\documentclass[acmsmall,nonacm]{acmart}

\acmJournal{PACMPL}
\acmVolume{1}
\acmNumber{CONF} 
\acmArticle{1}
\acmYear{2018}
\acmMonth{1}
\acmDOI{} 
\startPage{1}

\setcopyright{none}

\usepackage{booktabs}   
\usepackage{subcaption} 
                        
\usepackage[utf8]{inputenc}
\usepackage{thmtools}
\usepackage{tikz}

\usetikzlibrary{arrows,automata,positioning}

\newcommand\tbd[1]
{   \bigskip\par
                  	\fbox{\begin{minipage}{0.9\textwidth}
                  			#1
                  	\end{minipage}}
                  	\par\bigskip}
\newcommand\tbdcol[2]
{   \bigskip\par
             		\fbox{\begin{minipage}{0.9\textwidth}
             			\textcolor{#1}{#2}
             		\end{minipage}}
             		\par\bigskip}

\newcommand{\switchExtended}[2]{#2} 

\newcommand{\true}{\mathit{true}}
\newcommand{\false}{\mathit{false}}
\newcommand{\Hmu}{H_{\mu}}
\newcommand{\kvari}{G_k^n}
\newcommand{\AP}{\text{AP}}
\newcommand{\PA}{\text{PA}}
\newcommand{\TA}{\text{TA}}
\newcommand{\Paths}{\text{Paths}}
\newcommand{\Traces}{\text{Traces}}
\newcommand{\UNDECIDABLE}{\mathsf{UNDECIDABLE}}
\newcommand{\SPACE}{\mathsf{SPACE}}
\newcommand{\NSPACE}{\mathsf{NSPACE}}
\newcommand{\NLOGSPACE}{\mathsf{NLOGSPACE}}
\newcommand{\PSPACE}{\mathsf{PSPACE}}
\newcommand{\EXPSPACE}[1]{\mathsf{#1 EXPSPACE}}

\DeclareFontFamily{OT1}{pzc}{}
\DeclareFontShape{OT1}{pzc}{m}{it}{<-> s * [1.10] pzcmi7t}{}
\DeclareMathAlphabet{\mathpzc}{OT1}{pzc}{m}{it}

\begin{document}

\title[]{Automata and Fixpoints for Asynchronous Hyperproperties}         


\author{Jens Oliver Gutsfeld}
\affiliation{
  \department{Institut für Informatik}              
  \institution{Westfälische Wilhelms Universität Münster}            
  \streetaddress{Einsteinstraße 62}
  \city{Münster}
  \state{North Rhine-Westphalia}
  \postcode{48149}
  \country{Germany}                    
}
\email{jens.gutsfeld@uni-muenster.de}          

\author{Markus Müller-Olm}
\affiliation{
	\department{Institut für Informatik}              
	\institution{Westfälische Wilhelms Universität Münster}            
	\streetaddress{Einsteinstraße 62}
	\city{Münster}
	\state{North Rhine-Westphalia}
	\postcode{48149}
	\country{Germany}                    
}
\email{markus.mueller-olm@uni-muenster.de}          

\author{Christoph Ohrem}
\affiliation{
	\department{Institut für Informatik}              
	\institution{Westfälische Wilhelms Universität Münster}            
	\streetaddress{Einsteinstraße 62}
	\city{Münster}
	\state{North Rhine-Westphalia}
	\postcode{48149}
	\country{Germany}                    
}
\email{christoph.ohrem@uni-muenster.de}          

\begin{abstract}
Hyperproperties have received increasing attention in the last decade due to their importance e.g. for security analyses. 
Past approaches have focussed on synchronous analyses, i.e. techniques in which different paths are compared lockstepwise. 
In this paper, we systematically study asynchronous analyses for hyperproperties by introducing both a novel automata model (Alternating Asynchronous Parity Automata) and the temporal fixpoint calculus $\Hmu$, the first fixpoint calculus that can systematically express hyperproperties in an asynchronous manner and at the same time subsumes the existing logic HyperLTL.
We show that the expressive power of both models coincides over fixed path assignments.
The high expressive power of both models is evidenced by the fact that decision problems of interest are highly undecidable, i.e. not even arithmetical.
As a remedy, we propose approximative analyses for both models that also induce natural decidable fragments.
\end{abstract}



\maketitle

\section{Introduction}\label{sec:introduction}

Hyperproperties \cite{Clarkson2010} are a recent innovation in theoretical computer science.
While a traditional trace property (like liveness or safety) refers to single traces, a hyperproperty refers to \textit{sets of traces}.
Hyperproperties of interest include security properties like non-interference or observational determinism since it can only be inferred from combinations of traces and their relation to each other whether a system fulfills these properties. 
Analysis methods for hyperproperties have been proposed in many contexts, including abstract interpretation \cite{Mastroeni2017,Mastroeni2018}, runtime verification \cite{Finkbeiner2019}, synthesis \cite{Finkbeiner2020} and model checking \cite{Clarkson2014,Finkbeiner2015,Rabe2016,Gutsfeld2020}.
In model checking, several temporal logics for hyperproperties have been proposed, including \textit{hyperized} variants of LTL \cite{Clarkson2014, Finkbeiner2015, Rabe2016}, CTL$^*$ \cite{Clarkson2014, Finkbeiner2015, Rabe2016}, QPTL \cite{Rabe2016,Coenen2019} and PDL$-\Delta$ \cite{Gutsfeld2020}.
In all these logics, specifications are \textit{synchronous}, i.e.\ the modalities only allow for lockstepwise traversal of different paths.
The same is true of the automata-theoretic frameworks underlying the algorithms for these logics. 
However, the restriction to synchronous traversal of traces is a conceptual limitation of existing approaches \cite{Finkbeiner2017} that seems rather artificial.
Arguably, the ability to accommodate \textit{asynchronous} specifications is an important requirement for hyperproperty verification in various scenarios since many interesting properties require comparing traces at different points of time.
For instance, synchronous formulations of information flow security properties such as non-interference are often too strict for system abstractions with varying granularity of steps.
Here, a proper security analysis requires asynchronicity in order to match points of interest.
Asynchronous hyperproperties also arise naturally in the context of multithreaded programs where each thread is represented by a single trace as the overall system behaviour is an asynchronous interleaving of the individual traces.

In order to investigate the boundaries of automatic analysis of asynchronous hyperproperties induced by undecidability and complexity limitations, we propose both an automata-theoretic framework, Alternating Asynchronous Parity Automata (AAPA), and a temporal fixpoint calculus, $\Hmu$,  for asynchronous hyperproperties in this paper.
Our contribution is threefold: first of all, we show that both perspectives indeed coincide over fixed sets of paths by providing a direct and intuitive translation between AAPA and $\Hmu$ formulas.
Secondly, using this correspondence, we highlight the limitations of the analysis of asynchronous hyperproperties by showing that major problems of interest for both models (model checking, satisfiability, automata emptiness) are not even arithmetical. 
Thus, these problems are not only undecidable, but also exhaustive approximation analyses are impossible as they require recursive enumerability.
Finally, we consider natural semantic restrictions -- $k$-synchronicity and $k$-context-boundedness -- of both models that give rise to families of increasingly precise over- and underapproximate analyses.
Also, we identify settings where these analyses yield precise results.
We provide precise completeness results for all but one of the corresponding decision problems. 
Our complexity results for restricted classes of AAPA also shed new light on the classical theory of multitape automata over finite words as both the restrictions and the proofs can be directly transferred.

The rest of the paper is structured as follows: in \autoref{sec:preliminaries}, we provide some basic notation and recall the definitions of Alternating Parity Automata and Regular Transducers.
Then, in \autoref{sec:automata}, we introduce Alternating Asynchronous Parity Automata (AAPA) as a model for the asynchronous analysis of multiple input words and study their closure and decidability properties.
As the emptiness problem is undecidable, we discuss approximate analyses which lead to decidability for corresponding fragments.
We introduce $\Hmu$ as a novel fixpoint logic for hyperproperties in \autoref{sec:hypercalculus}. 
\autoref{sec:connection} establishes the connection between AAPA and $\Hmu$.
In \autoref{sec:modelchecking} and \autoref{sec:satisfiability}, this connection is used 
to transfer the approximate analyses of AAPA to $\Hmu$ and to obtain tight complexity bounds for corresponding decision problems.
We summarise the paper in \autoref{sec:conclusion}.
Due to lack of space, we have transferred some proofs to \switchExtended{an extended version \cite{}}{the appendix}.

\textbf{Related work:} 
Hyperproperties were systematically considered in \cite{Clarkson2010}.
The temporal logics HyperLTL and HyperCTL$^*$ were introduced in \cite{Clarkson2014} and efficient algorithms for them were developed in \cite{Finkbeiner2015}. 
The polyadic $\mu$-calculus \cite{Andersen1994} is directly related to hyperproperties.
It extends the modal $\mu$-calculus by branching over tuples of states instead of single states and has recently been considered in the context of so-called \textit{incremental hyperproperties} \cite{Milushev2013}.
This logic can express properties that are not expressible in HyperLTL and vice versa \cite{Rabe2016}.
The same relation holds between the polyadic $\mu$-calculus and $\Hmu$:
On the one hand, HyperLTL can be embedded into $\Hmu$ trivially, and on the other hand, $\Hmu$ is a linear time logic, while the polyadic $\mu$-calculus is a branching time logic, implying the logics are expressively incomparable.
The polyadic $\mu$-calculus was later reinvented \cite{Lange2015} under the name \textit{higher-dimensional $\mu$-calculus} \cite{Otto1999} and it was shown that every bisimulation-invariant property of finite graphs that can be decided in polynomial time can be expressed in it.

A different class of logics with the ability to express hyperproperties are the first- and second-order logics with equal-level predicate MPL[E], MSO[E], FOL[<,E] and S1S[E] \cite{Coenen2019,Finkbeiner2017,Thomas2011}.
We believe that $\Hmu$ can be embedded into the most powerful of these logics, S1S[E] and MSO[E].
Since MPL[E] and MSO[E] are branching time logics while $\Hmu$ is a linear time logic, just like FOL[<,E] and S1S[E], we restrict our further analysis to the relationship between $\Hmu$ and these latter two logics.
We believe that the expressive power of FOL[<,E] and $\Hmu$ is incomparable:
As for HyperLTL \cite{Bozzelli2015}, the property that an atomic proposition does not occur on a certain level in the tree (of traces) -- which is directly expressible in FOL[<,E] -- likely is not expressible in $\Hmu$.
On the other hand, for singleton trace sets, $\Hmu$ and FOL[<,E] degenerate to the linear time $\mu$-calculus and FOL[<], respectively, and it is known that FOL[<] - unlike the linear time $\mu$-calculus - cannot express all $\omega$-regular properties.
Notwithstanding S1S[E], we think the study of $\Hmu$ is of interest because (i) it is closer to logics traditionally used in model checking and (ii) there is no obvious way to define decidable approximate analyses for S1S[E] as we do for $\Hmu$.
Indeed, all results concerning S1S[E] we are aware of are undecidability results.

The logic $\Hmu$ proposed in the current paper is based on the linear-time $\mu$-calculus \cite{Vardi1988} and our model checking algorithms use a construction based on alternating parity word automata with holes in the flavour of \cite{Lange2005} while handling quantifiers via the constructions for HyperCTL$^*$ from \cite{Finkbeiner2015}.
AAPA are asynchronous $\omega$-automata with a parity acceptance condition.
Asynchronous automata on finite words were already introduced in the seminal paper by Rabin and Scott \cite{Rabin1959} on finite automata and later considered in many other contexts \cite{Geidmanis1987,Ibarra2013,Furia2014}.
On infinite words, Büchi automata on multiple input words were considered in the context of recursion theory and descriptive set theory \cite{Finkel2006,Finkel2016}.
As far as we are aware, automata of this type with a parity acceptance condition or alternation have not been studied yet and neither have algorithms for decidable restrictions and their exact complexity.
A different line of research discusses variants of asynchronous automata for the analysis of concurrent programs \cite{Zielonka1987,Muscholl1996,Peled1996}.
However, unlike AAPA, these models are concerned with language recognition for trace languages in the context of concurrent systems.
For AAPA (and H$_\mu$), we use two types of restrictions: $k$-synchronicity and $k$-context-boundedness. The first restriction has been discussed in the context of multitape automata \cite{Ibarra2013,Furia2014}, while the second restriction is inspired by a similar condition used in the analysis of concurrent programs \cite{Atig2009,Qadeer2005,Qadeer2008,Bansal2013}.
In \cite{Krebs2017}, Krebs et. al. consider a team semantics based approach to the verification of hyperproperties using variants of LTL with synchronous and asynchronous semantics.
Of course, there is a large body of work on the analysis of asynchronous systems, e.g. \cite{Ganty2009,Ganty2012,DurandGasselin2015,Esparza2016}. However, we are not aware of any such work concerning hyperproperties.


\section{Preliminaries}\label{sec:preliminaries}

Let $\AP$ be a finite set of atomic propositions.
A \textit{Kripke Structure} is a tuple $\mathcal{K} := (S, s_0, \delta, L)$ where $S$ is a finite set of states, $s_0 \in S$ is an initial state, $\delta \subseteq S \times S$ is a transition relation and $L : S \to 2^{\AP}$ is a labeling function.
We assume that there are no states without outgoing edges, that is for each $s \in S$, there is an $s' \in S$ with $(s,s') \in \delta$.
A path in a Kripke Structure $\mathcal{K}$ is an infinite sequence $s_0 s_1 ... \in S^{\omega}$ where $s_0$ is the initial state of $\mathcal{K}$ and $(s_{i}, s_{i+1})\in \delta$ holds for all $i \geq 0$.
We denote by $\Paths(\mathcal{K})$ the set of paths in $\mathcal{K}$ starting in $s_0$. 
A trace is an infinite sequence from the set $(2^{\AP})^{\omega}$.
For a path $s_0 s_1 ...$, the induced trace is given by $L(s_0) L(s_1) ...$. 
We write $\Traces(\mathcal{K})$ to denote the traces induced by paths of a Kripke Structure $\mathcal{K}$ starting in $s_0$.

Let $\Sigma$ be a finite input alphabet.
A (nondeterministic) regular transducer over $\Sigma$ is a tuple $\mathcal{T} = (Q,q_0,\gamma)$ where $Q$ is a finite set of control locations, $q_0 \in Q$ is an initial location and $\gamma: Q \times \Sigma \to 2^{Q \times \Sigma}$ is a transition function.
Given a word $w = w_1 ... w_n \in \Sigma^*$, a run of $\mathcal{T}$ on $w$ is an alternating sequence $q_0 v_1 q_1 v_1 ... v_n q_n$ such that $(q_{i+1},v_{i+1}) \in \gamma(q_i,w_{i+1})$ for all $0 \leq i < n$.
We then call $\mathcal{T}(w) := v_1 ... v_n \in \Sigma^*$ an output of $\mathcal{T}$ on $w$.
Intuitively, a nondeterministic regular transducer can be seen as a nondeterministic finite automaton (NFA) with output.

An Alternating Parity Automaton (APA) over $\Sigma$ is a tuple 
$\mathcal{A} = (Q, q_0, \rho, \Omega)$
such that $Q$ is a finite, non-empty set of control locations, $q_0 \in Q$ is an initial control location, $\rho: Q \times \Sigma \rightarrow \mathcal{B}^+(Q)$ is a function that maps control locations and input symbols	to positive boolean formulas over control locations and $\Omega: Q  \rightarrow \{0, 1, \dots k\}$ is a function that maps control locations to priorities.
We assume that every APA has two distinct states $\true$ with priority 0 and $\false$ with priority 1 such that $\rho(\true,\sigma) = \textit{true}$ and $\rho(\false,\sigma) = \textit{false}$ for all $\sigma \in \Sigma$.
If $\rho(q,\sigma)$ only consists of disjunctions for every $q$ and $\sigma$, we call an APA a Nondeterministic Parity Automaton (NPA) and denote $\rho(q,\sigma)$ as a set of control locations.
Additionally, we allow states $X$ with $\rho(X,\sigma) = \bot$ and $\Omega(X) = \bot$, which we call holes \cite{Lange2005}.
Intuitively, a hole is a state where the construction of an APA is not yet finished.
By $\mathcal{A}[X:\mathcal{A}']$, we denote the APA $\mathcal{A}$ where the hole $X$ is replaced by the automaton $\mathcal{A}'$.

A tree $T$ is a subset of $\mathbb{N}^{*}$ such that for every node $t \in \mathbb{N}^{*}$ and every positive integer $n \in \mathbb{N}$: $t \cdot n \in T$ implies (i) $t \in T$ (we then call $t \cdot n$ a child of $t$), and (ii) for every $0 < m < n$, $t \cdot m \in T$. 
We assume every node has at least one child.
A path in a tree $T$ is a sequence of nodes $t_0 t_1 ...$ such that $t_0 = \varepsilon$ and $t_{i+1}$ is a child of $t_i$ for all $i \in \mathbb{N}_{0}$.
A run of an APA $\mathcal{A}$ on an infinite word $w \in \Sigma^{\omega}$ is defined as a $Q$-labeled tree $(T,r)$ where $r: T \to Q$ is a labelling function such that $r(\varepsilon) = q_{0}$ and for every node $t \in T$ with children $t_1,...,t_k$, we have $1 \leq k \leq |Q|$ and the valuation assigning true to the control locations $r(t_1),...,r(t_k)$ and false to all other control locations satisfies $\rho(r(t),w(|t|))$.
A run $(T,r)$ is an accepting run iff for every path $t_1 t_2 ...$ in $T$, the lowest priority occuring infinitely often is even.
A word $w$ is accepted by $\mathcal{A}$ iff there is an accepting run of $\mathcal{A}$ on $w$.
The set of infinite words accepted by $\mathcal{A}$ is denoted by $\mathcal{L}(\mathcal{A})$.
Extending the notion of holes, we write $\mathcal{A}[X:\mathcal{L}]$ for a language $\mathcal{L} \subseteq \Sigma^{\omega}$ to denote $\mathcal{A}[X:\mathcal{A}']$ for some automaton $\mathcal{A}'$ with $\mathcal{L}(\mathcal{A}') = \mathcal{L}$.
We call an APA (resp. NPA) an Alternating Büchi Automaton (resp. Nondeterministic Büchi Automaton) iff its priorities are $0$ and $1$. We abbreviate these automata as ABA and NBA.
In the remainder of this paper, we use known theorems about parity and Büchi automata (without holes):

\begin{proposition}[\cite{Dax2008}]\label{thm:paritydealternation}
	For every APA $\mathcal{A}$ with $n$ states and $k$ priorities, there is a nondeterministic Büchi Automaton with $2^{\mathcal{O}(n \cdot k \cdot log\ n)}$ states accepting the same language.
\end{proposition}
\begin{proposition}\label{thm:paritycomplement}
	For every APA $\mathcal{A}$ with $n$ states and $k$ priorities, there is an APA $\overline{\mathcal{A}}$ with $n$ states and $k$ priorities that recognises the complement language.
\end{proposition}
\begin{proposition}\label{thm:parityemptiness}
The emptiness problem is $\PSPACE$-complete for APA and $\NLOGSPACE$-complete for NPA and NBA.
\end{proposition}

\autoref{thm:paritycomplement} and \autoref{thm:parityemptiness} can be found e.g.\ in \cite{Demri2016}.
On multiple occasions in this paper, we use a function for nested exponentials.
Specifically, we define $g_{c,p}(0,n) := p(n)$ and $g_{c,p}(d+1,n) := c^{g_{c,p}(d,n)}$ for a constant $c > 1$ and a polynomial $p$.
For $c = 2$ and $p = id$, i.e. the identity function, we omit the subscripts in order to improve readability.
By slight abuse of notation, we say that a function $f$ is in $\mathcal{O}(g(d,n))$ if $f$ is in $\mathcal{O}(g_{c,p}(d,n))$ for some constant $c > 1$ and polynomial $p$.
We straightforwardly extend this notion to multiple $g$ functions, where different constants $c > 1$ and polynomials $p$ can be used for the various $g$ functions.
Also, we write $\SPACE(g(d,n))$ as an abbreviation for $\bigcup_{c,p}\SPACE(g_{c,p}(d,n))$.
\section{Alternating Asynchronous Parity Automata}\label{sec:automata}

We introduce a new class of automata for the asynchronous traversal of multiple $\omega$-words.

\begin{definition}[Alternating Asynchronous Parity Automaton]
	Let $M = \{1, 2, \dots n\}$ be a set of directions and $\Sigma$ an input alphabet. 
	An Alternating Asynchronous Parity Automaton (AAPA) is a tuple $\mathcal{A} = (Q,\rho_0,\rho,\Omega)$ where
	\begin{itemize}
		\item $Q$ and $\Omega$ are the same as in an APA,
		\item $\rho_0 \in \mathcal{B}^+(Q)$ is a positive boolean combination of initial states, and
		\item $\rho: Q \times \Sigma \times M \to \mathcal{B}^+(Q)$ maps triples of control locations, input symbols and directions to positive boolean combinations of control locations.
	\end{itemize}
\end{definition}

Just as for APA, we call an AAPA where $\rho(q,\sigma,d)$ and $\rho_0$ only consist of disjunctions a Nondeterministic Asynchronous Parity Automaton (NAPA).
Compared to an APA, where a single word over $\Sigma$ is given as input, an AAPA has access to $n$ input words over $\Sigma$ and can perform steps on them individually.
The $M$ argument of the transition function indicates on which input word to progress.
Note that any APA can be seen as an AAPA with $n = 1$.
The definition of a run $T$ of an AAPA is similar to the one for a run of an APA, but with the following modifications:
\begin{itemize}
	\item the run is defined over $n$ input words $w_1,...,w_n \in \Sigma^{\omega}$ instead of a single word $w$,
	\item for each $t \in T$, we have $n$ offset counters $c_1^t,...,c_n^t$ starting at $c_i^{t} = 0$ for all $i$ and $t$ with $|t| \leq 1$,
	\item we have $\{r(t) | t \in T, |t| = 1\} \models \rho_0$, and
	\item when node $t \in T \setminus \{\varepsilon\}$ has children $t_1,...,t_k$, then there is a $d \in M$ such that (i) $c_d^{t_i} = c_d^t + 1$ and $c_{d'}^{t_i} = c_{d'}^{t}$ for all $i$ and $d' \neq d$, (ii) we have $1 \leq k \leq |Q|$ and (iii) the valuation assigning true to $r(t_1),...,r(t_k)$ and false to all other states satisfies $\rho(r(t),w_d(c_d^t),d)$.
\end{itemize}

These automata are particularly suitable for the analysis of our new logic $\Hmu$, which we introduce in the next section.
Indeed, AAPA and $\Hmu$ are able to express the same asynchronous restrictions on multiple $\omega$-words, as shown in \autoref{sec:connection}.
In order to compare AAPA to different automata models over a single input word, we interpret the $n$ input words $w_1,...,w_n$ with $w_i = w_i^0 w_i^1 ... \in \Sigma^{\omega}$ as a single word $w = (w_1^0,...,w_n^0)(w_1^1,...,w_n^1)... \in (\Sigma^n)^{\omega}$.
We introduce some notation for the asynchronous manipulation of such words.
For this purpose, let $v = (v_1,...,v_n) \in \mathbb{N}_0^n$ be a vector.
Then, we use $w[v] = (w_1^{v_1},...,w_n^{v_n})(w_1^{v_1 + 1},...,w_n^{v_n + 1})...$ to denote $w$ shifted left according to entries in $v$.
We write $v + v'$ and $v + e_i$ for standard vector addition with an arbitrary vector $v'$ and a unit vector $e_i = (0,...,0,1,0,...,0)$, respectively.

We establish some properties of AAPA and their nondeterministic counterparts.
First, a standard argument using alternation and priority shifts gives us the following result about AAPA:

\begin{theorem}\label{theorem:AAPAclosure}
AAPA are closed under union, intersection and complement. The constructions are linear in the size of the input automata.
\end{theorem}

However, for their nondeterministic counterpart, the above result does not hold:

\begin{theorem}
NAPA are not closed under intersection. In particular, it is undecidable whether there is a NAPA recognizing the intersection $\mathcal{L} (\mathcal{A}_1)  \cap  \mathcal{L} (\mathcal{A}_2)$ for two NAPA $\mathcal{A}_1$,$\mathcal{A}_2$. . 
\end{theorem} 

\begin{proof}
Let finite state Asynchronous Automata (AA) be NAPA on finite words which accept input words by reaching accepting states.
Fix two languages $\mathcal{L}_1, \mathcal{L}_2$  recognizable by AA over an alphabet $\Sigma$.
Let $\bot$ be a fresh symbol.
Then $\hat{\mathcal{L}_1} = \{w \bot^{\omega} \mid w \in \mathcal{L}_1\}$ and $\hat{\mathcal{L}_2} = \{w \bot^{\omega} \mid w \in \mathcal{L}_2\}$ are recognizable by two NAPA.
By a standard argument, from a NAPA recognizing $\hat{\mathcal{L}_1} \cap \hat{\mathcal{L}_2}$, we obtain an AA recognizing $\mathcal{L}_1 \cap \mathcal{L}_2$, but AA are not closed under intersection and it is undecidable whether there is an AA recognizing the intersection \cite{Furia2014}.
\end{proof}

Due to this difference in closure properties, we obtain a gap in expressivity:

\begin{corollary}
There is an AAPA such that no NAPA recognises the same language. It is undecidable whether an AAPA can be translated to a NAPA that recognises the same language.
\end{corollary}

A related result regarding the intersection of NAPA is:

\begin{theorem}
	The emptiness problem for the intersection of NAPA is undecidable.
\end{theorem}
\begin{proof}
	The proof is by reduction from the Post Correspondence Problem (PCP). 
	Let $I = (w_1, u_1) \dots (w_n, u_n)$ be a PCP instance. 
	We again choose $\bot$ as a fresh symbol.
	Let $\mathcal{L}_1$ be the language consisting of all possible concatenations of $(w_i, u_i)$ followed by $(\bot,\bot)^{\omega}$ and $\mathcal{L}_2$ be the language $\{(a_1,a_1)(a_2,a_2)\dots(a_l,a_l)(\bot,\bot)^{\omega} | a_i \in \Sigma, l \in \mathbb{N}_0\}$.
	$\mathcal{L}_1$ can be recognised by a NAPA with loops on an initial state that first read $w_i$ in direction $1$ and then read $u_i$ in direction $2$. 
	$\mathcal{L}_2$ can be recognised by inspecting a single symbol of both directions in turns and making sure that the same symbol is being read.
	Then the PCP instance has a solution iff $\mathcal{L}_1 \cap \mathcal{L}_2$ is non-empty.
\end{proof}

Since the intersection of two NAPA can be recognised by an AAPA, we immediately obtain:

\begin{corollary}
	The emptiness problem for AAPA is undecidable.
\end{corollary}

Furthermore, since AAPA can be used to decide inclusion problems for two-tape Büchi automata \cite{Finkel2009}, this result can be strengthened drastically.
We refer the interested reader to the \switchExtended{extended version}{appendix} for a detailed proof of the strengthened claim.

\begin{theorem}\label{thm:aapasigmaone}
The emptiness problem for AAPA is hard for the level $\Sigma_2^1$ of the analytical hierarchy.
\end{theorem}


When comparing NAPA to synchronous automata, it is not hard to see that the language $\mathcal{L} = \{(a,a,a)^n(b,a,a)^n(b,b,a)^n(b,b,b)^{\omega} | n \in \mathbb{N}_0\}$ can be recognised by a NAPA, while there is no synchronous parity automaton recognising $\mathcal{L}$ since it is not an $\omega$-regular language.
Despite this increase in expressive power, the emptiness problem for NAPA can still be reduced to an emptiness test on their synchronous counterparts.

\begin{theorem}\label{thm:napaemptiness}
The emptiness problem for NAPA is $\PSPACE$-complete.
\end{theorem}
\begin{proof}
Given a NAPA $\mathcal{A}$ with $m$ states and $n$ input words over the alphabet $\Sigma$, we construct a NPA $\mathcal{A}'$ with $m \cdot |\Sigma|^n$ states over  $\Sigma$ that stores the currently accessible input vector in its state and guesses the next input symbol for each direction.
For hardness, we can encode the configurations of a Turing machine whose space is bounded by a polynomial $p(x)$ in a NAPA with $n = p(x)$ directions.
Consistency of successive configurations can be checked locally such that there is no need to represent full configurations in the input alphabet or the state space.
Therefore, the size of the NAPA stays polynomial despite of the fact that the Turing machine has exponentially many configurations.
Thus, we can check for the existence of an accepting run of the Turing machine.
\end{proof}

Since the translation used in the proof is only exponential in $n$ and NAPA subsume synchronous Büchi automata, we obtain the following corollary:

\begin{corollary}
	For fixed $n$, the emptiness problem for NAPA is $\NLOGSPACE$-complete.
\end{corollary}

As emptiness of AAPA cannot be decided and as alternation elimination is not possible, we study analyses that consider well-specified subclasses of runs and identify fragments of AAPA for which these analyses are precise.

\subsection{$k$-synchronous analysis of AAPA}

\begin{definition}
	We call a run $T$ of an AAPA with $n$ input words $k$-synchronous for a $k \in \mathbb{N} \cup \{\infty\}$, if in every node $t$ in $T$, the offset counters $c_1^t,...,c_n^t$ satisfy $|c_i^t - c_j^t| \leq k$ for all $i$ and $j$.
\end{definition}

Intuitively, a $k$-synchronous run has the property that the AAPA can never be \textit{ahead} more than $k$ steps in one direction than in any other.
This gives rise to an approximate analysis where only the $k$-synchronous runs of an AAPA are considered.
Since a $k$-synchronous run is particularly a $k'$-synchronous run for all $k' \geq k$, the approximation improves with increasing $k$, capturing the whole semantics at $k = \infty$ for all AAPA.
Since an analysis with $k = \infty$ is impossible, we assume that $k < \infty$ in the remainder of this section.
We show that $k$-synchronous runs of an AAPA can be analysed via a reduction to APA:

\begin{theorem}\label{thm:aapaksynchronoustosynchronous}
	For every AAPA $\mathcal{A}$ over $\Sigma$ with $l$ priorities, $n$ input words and $m$ states, there is an APA over $\Sigma^n$ with $\mathcal{O}(l \cdot m \cdot |\Sigma|^{k \cdot n})$ states recognizing all words accepted by a $k$-synchronous run of $\mathcal{A}$.
\end{theorem}

\begin{proof}	
	We read input vectors synchronously and maintain a $k \cdot n$ window of the input words in the state space.
	Since in a $k$-synchronous run of an AAPA each direction can be ahead each other direction at most $k$ steps, no direction can leave this window without all rearmost directions performing steps.
	The content of the window can therefore be stored in an APA's state space.
	
	We simulate steps of the AAPA by moving markers forward in each row of the window.
	When the last marker leaves the rearmost column of the window, a new input vector is read and added to the front.
	A step that would leave the window must not be simulated since that would move that direction $k+1$ steps ahead of the rearmost direction. 
	Therefore, we move to $\false$ instead.
	This way, we ensure that each run of the APA simulates a $k$-synchronous run of the AAPA.
	
	Since an input vector is only read when all directions have left the last column, one step of the APA has to simulate multiple successive steps of the AAPA.
	This is done by a nondeterministic choice over all sequences of steps resulting in the rearmost column being erased.
	In order to correctly mirror the priorities, for each simulated sequence of steps, we move to a copy of the reached state annotated with the lowest priority encountered in this sequence.
	
	
	In order to fill the window, $k$ input vectors are read in an \textit{initialisation} phase of the APA.
	Overall, the size of the state spaces increases by a factor of $\mathcal{O}(l\cdot|\Sigma|^{k \cdot n})$.
\end{proof}

Note that \autoref{thm:aapaksynchronoustosynchronous} yields an underapproximation of an AAPA's behaviour.
By transitioning to $\true$ instead of $\false$ at a violation of $k$-synchronicity, we can instead perform an overapproximation that ignores non $k$-synchronous branches when determining whether a run is accepting.
Obviously, these approximations yield exact results for AAPA where all runs are $k$-synchronous.
We call such AAPA $k$-synchronous.
Synchronicity can be enforced by a syntactic restriction, namely that in all cycles in the transition graph, every direction occurs the same number of times.
The parameter $k$ is then induced by the largest difference of the number of occurences of two directions on any path in the transition graph.
We establish a tight complexity bound for this analysis:

\begin{theorem}\label{thm:synchronousemptiness}
The problem to decide whether there is a $k$-synchronous accepting run of an AAPA and thus the emptiness problem for k-synchronous AAPA is $\EXPSPACE{}$-complete.
\end{theorem}

\begin{proof}
Using \autoref{thm:aapaksynchronoustosynchronous}, we can construct an APA whose size is exponential in the size of the input.
Since APA can be tested for emptiness in $\PSPACE$, this establishes membership in $\EXPSPACE{}$.
For hardness, we reduce from the acceptance problem for deterministic exponentially space bounded Turing machines (DTMs).
Let $\mathcal{M}$ be a DTM with control locations $Q$ and space complexity $f(x) = 2^{p(x)}$ for some polynomial $p$.
For a given input $w$ of $\mathcal{M}$, we construct a $k$-synchronous AAPA $\mathcal{A}$ with $n := p(|w|) + 1$ and $k := 2$ as follows:

In the first direction, we successively read configurations of $\mathcal{M}$ separated by a marker.
We encode such configurations by words over $\{0,1\} \cup \{0,1\} \times Q$ where the occurrence of a tuple $(b,q)$
indicates that the head points to this position, that this position of the tape contains the bit $b$, and that the current state is $q$.
Presence of initial and final configuration can easily be checked via alternation.

Additionally, we need to check that successive configurations are constructed in accordance with the transition function of $\mathcal{M}$.
For this purpose, we have to count to exponentially large values in order to compare positions with the same (or neighbouring) index in successive configurations with each other.
Since we cannot store exponentially large counter values in the state, we instead construct a \textit{virtual counter gadget}: for every additional direction, we can enforce that it consist of the word $(01)^{\omega}$.
As we illustrate in \autoref{fig:countergadget}, the concatenation of the current values of these directions is interpreted as a counter.
Then, we conjunctively move to each position in a synchronous manner, save the value of the first direction in a state of the AAPA and initialise the counter: if the directions have value $1$, we advance them by one symbol and otherwise, we keep the position.
We then advance the first direction and increase our virtual counter by changing the appropriate bits via single advancements and maintaining the other bits by nondeterministically either maintaining the current symbol or advancing the direction by two positions.
Then, one of the nondeterministic choices preserves $2$-synchronicity.
When the counter has reached $2^{p(|w|)}$, we have found the matching tape cell in the next configuration and can check whether it is admissible by comparing its value with the value saved in the control state of the AAPA.

It is easy to see that $\mathcal{A}$ has a $2$-synchronous accepting run iff $\mathcal{M}$ accepts $w$.
\end{proof}
\begin{figure}
	\scalebox{.7}{
	\begin{tikzpicture}
		\node[] at (-6,1) (tm){DTM};
		\node[] at (-4,1) (tm1){1};
		\node[] at (-3,1) (tm2){\#};
		\node[] at (-2,1) (tm3){0};
		\node[] at (-1,1) (tm4){0};
		\node[] at (0,1) (tm5){1};
		\node[] at (1,1) (tm6){(1,q)};
		\node[] at (2,1) (tm7){0};
		\node[] at (3,1) (tm8){1};
		\node[] at (4,1) (tm9){1};
		\node[] at (5,1) (tm10){...};
		\node[] at (-6,0.5) (1bit){Bit 1};
		\node[] at (-4,0.5) (1bit1){0};
		\node[] at (-3,0.5) (1bit2){1};
		\node[] at (-2,0.5) (1bit3){0};
		\node[] at (-1,0.5) (1bit4){1};
		\node[] at (0,0.5) (1bit5){0};
		\node[] at (1,0.5) (1bit6){1};
		\node[] at (2,0.5) (1bit7){0};
		\node[] at (3,0.5) (1bit8){1};
		\node[] at (4,0.5) (1bit9){0};
		\node[] at (5,0.5) (1bit10){...};
		\node[] at (-6,0) (2bit){Bit 2};
		\node[] at (-4,0) (2bit1){0};
		\node[] at (-3,0) (2bit2){1};
		\node[] at (-2,0) (2bit3){0};
		\node[] at (-1,0) (2bit4){1};
		\node[] at (0,0) (2bit5){0};
		\node[] at (1,0) (2bit6){1};
		\node[] at (2,0) (2bit7){0};
		\node[] at (3,0) (2bit8){1};
		\node[] at (4,0) (2bit9){0};
		\node[] at (5,0) (2bit10){...};
		\node[] at (-6,-0.5) (3bit){Bit 3};
		\node[] at (-4,-0.5) (3bit1){0};
		\node[] at (-3,-0.5) (3bit2){1};
		\node[] at (-2,-0.5) (3bit3){0};
		\node[] at (-1,-0.5) (3bit4){1};
		\node[] at (0,-0.5) (3bit5){0};
		\node[] at (1,-0.5) (3bit6){1};
		\node[] at (2,-0.5) (3bit7){0};
		\node[] at (3,-0.5) (3bit8){1};
		\node[] at (4,-0.5) (3bit9){0};
		\node[] at (5,-0.5) (3bit10){...};
		\node[draw,red,minimum size= 12pt] at(0,1) (mark1){};
		\node[draw,red,minimum size= 12pt] at(1,0.5) (mark2){};
		\node[draw,red,minimum size= 12pt] at(0,0) (mark3){};
		\node[draw,red,minimum size= 12pt] at(-1,-0.5) (mark4){};
	\end{tikzpicture}
	}
	\caption{A $2$-synchronous virtual counter gadget with $3$ bits and current value $101$. The Turing machine is in state $q$ and its head is on the fourth bit of the current configuration.}
	\label{fig:countergadget}
\end{figure}

Note that in the proof of the lower bound detailed above, $k$ can be chosen as a fixed value.
This raises the question whether there is a construction for the emptiness test that is exponential only in $n$, but not in $k$.
Indeed, we can construct an APA that has a non-empty language if and only if the given AAPA has an accepting $k$-synchronous run. 
However, unlike the APA in the proof of \autoref{thm:aapaksynchronoustosynchronous}, the constructed APA only accepts a certain encoding of the word accepted by the $k$-synchrounous run instead of the word itself. 
More specifically, it expects the input word to consist of concatenations of input windows (with size $k\cdot n$) of the original input words, one window for each simulated step of the AAPA. 
For each direction it maintains a counter indicating its current position in the input window. 
Using alternation and additional counters, the APA can ensure that the succession of input windows is consistent.
A single step of the AAPA can then be simulated by reading the current copy of the input window in order to check that the direction inducing the step has the correct symbol.
An upper bound on the number of states of this APA is dominated by the $n$ counters up to $k$ indicating on what position in the window each direction is.
This results in a factor of $k^n$ on the number of states but stays polynomial for fixed $n$. As single steps of the AAPA are simulated separately, this construction also avoids the additional factor $l$.

Together with the fact that already the emptiness problem for APA is PSPACE-hard, these considerations lead to the following corollary.

\begin{corollary}
For fixed $n$, the emptiness problem for $k$-synchronous AAPA is $\PSPACE$-complete.
\end{corollary}

\subsection{$k$-context-bounded analysis of AAPA}

Since different words can only diverge from each other up to $k$ steps, the ability of AAPA to asynchronously traverse words is severely restricted in $k$-synchronous runs.
We thus consider a further class of runs where the positions on different words can diverge unboundedly.
For this restriction, we introduce the notion of a context:

\begin{definition}\label{definition:context}
	A context is a (possibly infinite) subpath $p = t_1 t_2 ...$ in a run of an AAPA over $w_1,...,w_n$ such that transitions between successive states all use the same direction, that is there is a $d \in M$ such that for all $i \in \{1,...,|p| -1\}$ we have $c_{d}^{t_{i+1}} = c_d^{t_i} + 1$.
	We call a run $T$ of an AAPA $k$-context-bounded if every path in $T$ consists of at most $k$ contexts.
\end{definition}


We propose an approximate analysis which checks only for the existence of a $k$-context-bounded accepting run of an AAPA.
\switchExtended{The}{In the appendix, we show that the} restriction that AAPA can only consider a single direction during each context is well-chosen in the sense that additionally allowing contexts in which a selection of directions is traversed synchronously leads to undecidability.

\begin{figure}
	\centering
	\scalebox{.7}{
		\begin{tikzpicture}
		\node[draw,circle,fill=green!20] at (0,0) (ep){$q_1$};
		\node[draw,circle,fill=red!20] at (-3,-1) (0){$q_2$};
		\node[draw,circle,fill=red!20] at (3,-1) (1){$q_3$};
		\node[draw,circle,fill=green!20] at (-5,-2) (00){$q_1$};
		\node[draw,circle,fill=red!20] at (-3,-2) (01){$q_2$};
		\node[draw,circle,fill=red!20] at (-1,-2) (02){$q_5$};
		\node[draw,circle,fill=green!20] at (1,-2) (10){$q_4$};
		\node[draw,circle,fill=red!20] at (5,-2) (11){$q_5$};
		\node[draw,circle,fill=green!20] at (-5,-3) (000){$q_4$};
		\node[draw,circle,fill=red!20] at (-3,-3) (010){$q_3$};
		\node[draw,circle,fill=green!20] at (-1,-3) (020){$q_6$};
		\node[draw,circle,fill=green!20] at (1,-3) (100){$q_3$};
		\node[draw,circle,fill=green!20] at (4,-3) (110){$q_4$};
		\node[draw,circle,fill=green!20] at (6,-3) (111){$q_6$};
		
		\node[] at (-5,-4) (dots){...};
		\node[] at (-3,-4) (dots){...};
		\node[] at (-1,-4) (dots){...};
		\node[] at (1,-4) (dots){...};
		\node[] at (4,-4) (dots){...};
		\node[] at (6,-4) (dots){...};
		
		\path[->] (ep) edge (0);
		\path[->] (ep) edge (1);
		\path[->] (0) edge (00);
		\path[->] (0) edge (01);
		\path[->] (0) edge (02);
		\path[->] (1) edge (10);
		\path[->] (1) edge (11);
		\path[->] (00) edge (000);
		\path[->] (01) edge (010);
		\path[->] (02) edge (020);
		\path[->] (10) edge (100);
		\path[->] (11) edge (110);
		\path[->] (11) edge (111);
		\end{tikzpicture}
	}
	\caption{A $3$-context-bounded run of an AAPA where green nodes move in direction $1$ and red nodes move in direction $2$. The corresponding guess is $(q_1,\{(q_2,\{q_1,q_6\}),(q_3,\{q_4,q_6\})\})$.}
	\label{illustration:guesses1}
\end{figure}
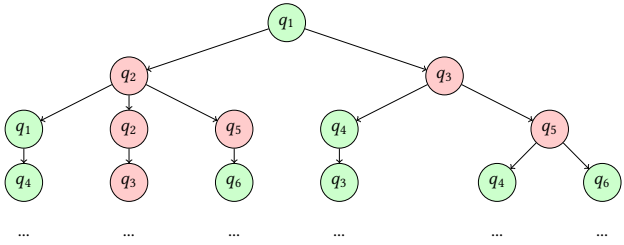

Our analysis is based on the observation that a run of an AAPA can be divided into maximal sections only reading a single direction and that in the case of a $k$-context-bounded run, the interaction between neighbouring sections can be described in a finite way.
Given such a description, each direction can be analysed independently from the others.
We illustrate this with an example with just two directions and two context switches.
Consider the run in \autoref{illustration:guesses1} where green states move in direction 1 and red states move in direction 2 and inspect the left part of the tree.
Instead of moving to $q_2$, the automaton responsible for direction 1 can directly move to $q_1$ and $q_6$ from $q_1$ if the automaton responsible for direction 2 ensures existence of an accepting run from $q_2$ that reaches only these two direction 1 states.
Consequently, the automaton responsible for direction 2 can safely cut off the branches between states $q_2$ and $q_1$ as well as between $q_5$ and $q_6$ since the automaton responsible for direction 1 ensures an accepting continuation from $q_1$ and $q_6$.
We furnish states with descriptions of this interplay between assumptions and guarantees and call such descriptions \textit{guesses}.
These guesses allow us to split the run from \autoref{illustration:guesses1} into two runs on single directions as shown in \autoref{illustration:guesses2}.
Guesses are nested and the nesting depth corresponds to the number of context switches that can still be performed in a $k$-context-bounded run.
They are constructed inductively:
In the most simple case, where no context switch can be made, the guess is empty.
If context switches are still possible, the guess is a set of states enriched with guesses with one context switch less.

\begin{figure}
	\centering
	\scalebox{.7}{
		\begin{tikzpicture}
			\node[draw,circle,fill=green!20] at (-3,0) (ep1){$\varepsilon$};
			\node[draw,circle,fill=green!20] at (-3,-1) (ep10){$q_1$};
			\node[draw,circle,fill=green!20] at (-5,-2) (ep100){$q_1$};
			\node[draw,circle,fill=green!20] at (-3,-2) (ep102){$q_6$};
			\node[draw,circle,fill=green!20] at (-1,-2) (ep103){$q_4$};
			\node[draw,circle,fill=green!20] at (-5,-3) (ep1000){$q_4$};
			\node[draw,circle,fill=green!20] at (-1,-3) (ep1030){$q_3$};
			
			\node[] at (-2.5,-0.7) (){$G_1$};
			
			\node[draw,circle,fill=red!20] at (5,0) (ep2){$\varepsilon$};
			\node[draw,circle,fill=red!20] at (3,-1) (ep20){$q_2$};
			\node[draw,circle,fill=red!20] at (7,-1) (ep21){$q_3$};
			\node[draw,circle,fill=red!20] at (2,-2) (ep200){$q_2$};
			\node[draw,circle,fill=red!20] at (4,-2) (ep201){$q_5$};
			\node[draw,circle,fill=red!20] at (6,-2) (ep210){$t$};
			\node[draw,circle,fill=red!20] at (8,-2) (ep211){$q_5$};
			\node[draw,circle,fill=red!20] at (2,-3) (ep2000){$q_3$};
			\node[draw,circle,fill=red!20] at (4,-3) (ep2010){$t$};
			\node[draw,circle,fill=red!20] at (8,-3) (ep2110){$t$};
			
			\node[] at (2.5,-0.7) (){$G_2$};
			\node[] at (7.5,-0.7) (){$G_3$};
			\node[] at (1.5,-1.7) (){$G_2$};
			\node[] at (4.5,-1.7) (){$G_2$};
			\node[] at (8.5,-1.7) (){$G_3$};
			\node[] at (2.5,-2.7) (){$G_2$};
			
			\node[] at (-5,-4) (dots){...};
			\node[] at (-3,-3) (dots){...};
			\node[] at (-1,-4) (dots){...};
			\node[] at (2,-4) (dots){...};
			\node[] at (4,-4) (dots){...};
			\node[] at (6,-3) (dots){...};
			\node[] at (8,-4) (dots){...};
			
			\path[->] (ep1) edge (ep10);
			\path[->] (ep10) edge (ep100);
			\path[->] (ep10) edge (ep102);
			\path[->] (ep10) edge (ep103);
			\path[->] (ep100) edge (ep1000);
			\path[->] (ep103) edge (ep1030);
			
			\path[->] (ep2) edge (ep20);
			\path[->] (ep2) edge (ep21);
			\path[->] (ep20) edge (ep200);
			\path[->] (ep20) edge (ep201);
			\path[->] (ep21) edge (ep210);
			\path[->] (ep21) edge (ep211);
			\path[->] (ep200) edge (ep2000);
			\path[->] (ep201) edge (ep2010);
			\path[->] (ep211) edge (ep2110);
		\end{tikzpicture}
	}
	\caption{Translation of the run into multiple synchronous runs with $G_1 = \{(q_2,G_2),(q_3,G_3)\}$, $G_2 = \{q_1,q_6\}$ and $G_3 = \{q_4,q_6\}$. As can be seen in nodes $01$ and $02$ in the green tree, one subtree starting in each node suffices to be adapted since they all start with the same offset on their word.}
	\label{illustration:guesses2}
\end{figure}
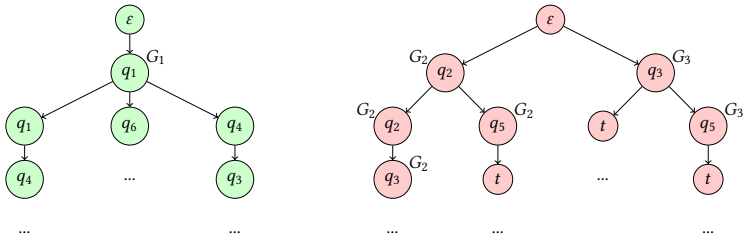

Formally, when analysing an AAPA $\mathcal{A} = (Q,\rho_0,\rho,\Omega)$ over $n$ words from $\Sigma^{\omega}$, we inductively define the set of guesses $G_i$ for $i$ context switches as follows: $G_0 := \{\bot\}$, $G_{i+1} := 2^{Q \times G_i}$ for $i \geq 0$,and $G := \bigcup_{i = 0}^{k-1} G_i$.
For simplicity, we assume that each state $q \in Q$ moves in a unique direction $d$ (i.e. the transition function yields $\false$ for other directions); the set of states that move in direction $d$ is called $Q_d$.
This increases the size of the state space by a factor of at most $n$.
As described above, we skip sections belonging to direction $d' \neq d$ when considering a direction $d$.
In order to extract the potential reentry points for directions $d$ from a guess, we use a frontier function $F_d: G \to G$ defined as follows:
\begin{equation*}
	F_d(g) = \{(q,g') \in g \mid q \in Q_d \} \cup \bigcup \{F_d(g') | (q,g') \in g, q \notin Q_d \} \,.
\end{equation*}
For the analysis, we define from $\mathcal{A}$ an APA-like structure $\mathcal{S} = (S,\rho_S,\Omega_S)$ over the alphabet $\Sigma$ with
\begin{align*}
	S &=\ Q \times G\,, \\
	\rho_S((q,g),\sigma) &= \bigvee\{\bigwedge Q'_d \times \{g\} \land 
	\bigwedge \{F_d(g'') \mid (q,g'') \in g' \} \mid \\
	&\qquad\qquad Q'_d \subseteq Q_d, g' \subseteq g, 
	Q'_d \cup \{q \mid \exists g'': (q,g'') \in g'\} \models \rho(q,\sigma,d) \} 
	\text{ for } q \in Q_d \,, \text{ and } \\
	\Omega_S((q,g)) &=\ \Omega(q) \,.
\end{align*}

Thus, the successors of a state $(q,g)$ in $\mathcal{S}$ are composed of two sets:
a set $Q'_d$ of states that belong to the same section and the set of reentry points extracted from the second component of the guesses in $g'$ annotating states for other directions in $g$. These sets are chosen such that $Q'_d$ and the states in the first component of the guesses in $g'$ satisfy the corresponding transition function of $\mathcal{A}$.

For this structure, we have:

\begin{lemma}\label{lem:contextboundedtest}
	There is a $k$-context-bounded accepting run in $\mathcal{A}$ iff there is a $g \in G$ such that $\mathcal{S}$ accepts from $\bigwedge \{ g' \in F_d(\{(q_0,g)\}) \mid q_0 \in Q_0 \}$ for all $d \in \{1,...,n\}$ and some $Q_0$ with $Q_0 \models \rho_0$.
\end{lemma}

With a little effort, this structure can be translated to an NPA over $\Sigma^n$ recognising the words accepted by $k$-context-bounded runs of the original AAPA $\mathcal{A}$. 
For this, transitions in each state $q\in Q_d$ of $\mathcal{S}$ are adjusted to the input alphabet $\Sigma^n$ by allowing arbitrary symbols in other directions than $d$.
The key insight is that the choice of the guess $g$ from \autoref{lem:contextboundedtest} can be integrated into an NPA after an alternation removal as a nondeterministic choice.
We parameterise the structure $\mathcal{S}$ in the initial guess $g$ (for which we assume w.l.o.g. $g \in G_{k-1}$) and obtain a structure $\mathcal{S}(g)$ of size $\mathcal{O}(g(k-2,|\mathcal{A}|))$, where the fixed $g \in G$ indicates the initial guess.
The initial state of $\mathcal{S}(g)$ can nondeterministically guess a state set $Q_0$ with $Q_0 \models \rho_0$ and then move to $\bigwedge_{d \in\{1,...,n\}}\bigwedge_{q_0 \in Q_0}\bigwedge F_d(\{(q_0,g)\})$ to obtain the desired test.
We eliminate alternation from $\mathcal{S}(g)$ to obtain an NPA $\mathcal{S}'(g)$ of size $\mathcal{O}(g(k-1,|\mathcal{A}|))$ for each guess $g \in G_k$.
Thus, an NPA which nondeterministically guesses a $g \in G$ and then moves to $\mathcal{S}'(g)$ performs the desired test from \autoref{lem:contextboundedtest}.
Since there are $|G| = \mathcal{O}(g(k-1,|\mathcal{A}|))$ possible guesses, this NPA's size asymptotically is $\mathcal{O}(g(k-1,|\mathcal{A}|))$ as well.
We obtain:

\begin{corollary}\label{cor:contextboundedtoapa}
For every AAPA $\mathcal{A}$, there is an NPA with $\mathcal{O}(g(k-1,|\mathcal{A}|))$ states recognizing all words accepted by a $k$-context-bounded run of $\mathcal{A}$.
\end{corollary}

Note that similar to the construction for \autoref{thm:aapaksynchronoustosynchronous}, a violation of $k$-context-boundedness leads to a transition to $\false$ in our construction (in this case by an empty disjunction in the transition function since there is no $g' \subseteq g$ fulfilling the conditions for $g  = \bot$).
This again yields an underapproximation of the AAPA's behaviour and can be transformed into an overapproximation by instead transitioning to $\true$.

Again, the developed analysis is precise for AAPA having only $k$-context-bounded runs, which we call $k$-context-bounded AAPA.
A syntactic restriction ensuring that an AAPA is context-bounded is that each strongly connected component (SCC) in its transition graph uses steps of a unique direction only.
The parameter $k$ is then induced by the maximum number of switches between directions in the DAG of SCCs.

We show that the size of this construction cannot asymptotically be improved.

\begin{theorem}\label{thm:contextboundedcomplexity}
	The problem to decide if there is a $k$-context-bounded accepting run of an AAPA $\mathcal{A}$ and thus the emptiness problem for $k$-context-bounded AAPA $\mathcal{A}$ is complete for $\EXPSPACE{(k-2)}$.
\end{theorem}

For the proof, we establish two helpful lemmas about AAPA with context bounds.
We make use of Stockmeyer's nested index encoding \cite{Stockmeyer1974}.
For a (finite) word $w = w_0 ... w_n$, the encoding is inductively defined by:
\begin{align*}
	stock_0(w) :=&\; [_0 w_0 ... w_n ]_0 \\
	stock_{k+1}(w) :=&\; [_{k+1} stock_k(bin(0)) w_0 ... stock_k(bin(n))  w_n ]_{k+1}
\end{align*}
Here, $bin(i)$ denotes a binary encoding of $i \in \mathbb{N}_0$.
For our proofs, it will be useful to use a least-significant-bit-first encoding for $bin(i)$.
Brackets $[_k$ and $]_k$ allow us to identify the beginning and end of a level $k$ encoding without necessarily knowing the exact length of the encoding.
For level $k \geq 1$, every letter of the word is preceded by its index in the word.
These indices are written in binary and thus form words of their own, which can be encoded as well.
The number of times the encoding of indices is nested is determined by the level of the encoding.
Since indices for a word of length $n$ can be encoded in words of length $log\ n$, every level of the encoding allows to encode exponentially larger words when starting from a fixed size.
Thus, with level $0$ encodings of length $m$, one can encode words of length $g(k,m)$ on level $k$.

The first lemma shows that by using this encoding, we can perform certain operations on sequences of words of length $g(k,m)$.
It is formulated in a generic way using regular transducers and regular languages in order to be applicable to nested index encodings as well as successions of Turing machine configurations.
The first is needed for its inductive proof; the second allows us to apply it for the proof of \autoref{thm:contextboundedcomplexity}.

The operations we need for our result are the following:
\begin{itemize}
	\item Checking whether the first and the last word of the sequence are contained in two given regular languages and
	\item checking whether each word (apart from the first one) is obtained by applying a given regular transducer to the previous word.
\end{itemize}
More concretely, for sequences of indices, we want to check:
\begin{itemize}
	\item Is the first index $0...0$?
	\item Is the last index $1...1$?
	\item Is each successor obtained by adding $1$ to the previous index?
\end{itemize}
For Turing machine configurations, the checks are:
\begin{itemize}
	\item Is the first configuration the starting configuration?
	\item Is the last configuration an accepting configuration?
	\item Is each configuration obtained by applying the Turing machine's transition function to its predecessor configuration? 
\end{itemize}
Note that for our reduction to be polynomial, each of the regular language acceptors and regular transducers have to be of polynomial size in the input size, which they indeed are.

With these considerations in mind, it is straightforward to see that an application of the next Lemma can be used to encode the acceptance problem of $g(k-2,p(n))$ space bounded Turing machines in $k$-context-bounded AAPA.
It is also easy to see that an application to sequences of indices can be used to check the validity of nested index encodings.
We now formulate the lemma:

\begin{lemma}\label{lem:contextboundedtransducer}
	Given a size bound $m$, an AAPA of size polynomial in $m$ with access to directions $i$ and $j$, as well as $k \geq 2$ context switches starting in a context for direction $i$ can:
	\begin{enumerate}
		\item enforce that the word in direction $i$ contains a level $k-2$ nested index encoding of all numbers between $0$ and $g(k-1,m)-1$ and
		\item enforce that the word in direction $j$ contains a sequence of level $k-1$ encoded words $w_1 w_2 ... w_l$, each of length $g(k-1,m)$, separated by markers, such that
		\begin{itemize}
			\item $w_1$ is contained in a given regular language with an acceptor of size $\mathcal{O}(m)$,
			\item $w_l$ is contained in another given regular language with an acceptor of size $\mathcal{O}(m)$, and
			\item for $h < l$, $w_{h+1}$ is obtained from $w_{h}$ by applying a regular transducer of size $\mathcal{O}(m)$.
		\end{itemize}
	\end{enumerate}
\end{lemma}

For the proof, we need the following lemma about AAPA:

\begin{lemma}\label{lem:contextboundedcomparison}
	Given a size bound $m$, an AAPA of size polynomial in $m$ with $k \geq 1$ context switches can check whether the level $k - 1$ encoded words of length $g(k - 1, m)$ on directions $i$ and $j$ differ.
\end{lemma}

\begin{proof}[Proof of \autoref{lem:contextboundedcomparison}]
	The lemma can by shown by an induction on $k$.
	
	In the \textbf{base case} $k = 1$, the words are level $0$ encoded and have length $m$.
	Since the number of context switches is restricted, it is not possible to check the two words for differing bits directly by moving the two directions forward one step at a time.
	Instead, we use alternation to perform this test.
	For every number of steps $s$ up to $m$, we disjunctively move $s$ steps forward in direction $i$, perform a context switch while memorising the number of steps and the last read symbol and then move the same number of steps in direction $j$ and compare the symbols.
	If the symbols differ, the words cannot be the same and we accept, otherwise we reject.
	If all disjunctive tests reject, there are no differing bits in the words and thus no accepting runs of the AAPA.
	
	In the \textbf{inductive step}, we have $k + 1$ context switches available to compare level $k \geq 1$ encoded words of length $g(k,m)$.
	Here, we cannot count the index of a symbol in the word in the state space since that would increase the AAPA's size beyond polynomial.
	We can, however, make use of the fact that the words are encoded such that each symbol in the word is preceded by its index in a level $k - 1$ encoding.
	Thus, we start on direction $i$ and disjunctively move to the start of some index where we expect the two words to differ.
	We then perform a context switch and conjunctively test for each position whether $(i)$ the indices differ or $(ii)$ the symbols differ.
	Test $(i)$ can be done using the induction hypothesis with the $k$ remaining context switches.
	Test $(ii)$ can be done using a single additional context switch by moving to the symbol, memorising it, then moving to the respective symbol on direction $i$ and comparing the two.
\end{proof}

\begin{proof}[Proof of \autoref{lem:contextboundedtransducer}]
	The lemma can be shown by an induction on $k$.
	
	The \textbf{base case} is $k = 2$.
	Item $(1)$ can be enforced by checking for the sequence of words bracketed by $[_0$ and $]_0$ in direction $i$, that (i) these words have length $m$, that (ii) the first one is $0...0$, that (iii) the last one is $1...1$ and that (iv) each successor is obtained from the previous word by adding $1$.
	These four checks can be done via alternation and a counter to obtain $(1)$.
	
	For $(2)$, we first need to check that direction $j$ contains level $1$ nested index encodings of words of length $2^m$.
	This can be done in the same way as the check in item $(1)$ with the difference that bits belonging to the encoded words themselves have to be skipped.
	These can, however, easily be found since the indices to be checked are bracketed by $[_0$ and $]_0$.
	$w_1$ and $w_l$ belonging to regular languages can be checked by their respective language acceptors, where bits not belonging to the words themselves can be skipped.
	The main difference between the two tests is that $l$ is not determined and thus the point where $l$ is reached has to be guessed nondeterministically.
	
	For the third item of $(2)$, the challenge is that corresponding positions in $w_h$ and $w_{h+1}$ are $\mathcal{O}(2^m)$ steps apart from each other due to the length of the words $w_h$ and thus cannot be matched by counting steps in the state of the AAPA as that would violate its size restriction.
	However, we can  use the two context switches and the fact that we have already checked item $(1)$ to ensure this item.
	A test for this starts in direction $i$ and conjunctively switches to direction $j$ at the start of each number from item $(1)$.
	The copy that performs the context switch before number $n$ then has to ensure the correct transduction of the $n$th bit from each word $w_h$ to $w_{h+1}$.
	For this purpose, it conjunctively moves to the start of each word $w_h$ and performs the transduction of bit $n$. 
	As the state space is not large enough to store the value $n$, this is done in the following way:
	for each position in $w_h$ it is checked whether (i) the correct transduction is being performed or (ii) the index does not match $n$.
	The latter is done by using \autoref{lem:contextboundedcomparison} with the remaining context switch.
	Note that in order to ensure that the former check can be performed, the control location of the transducer has to be tracked. 
	This can be done by enriching the input word in direction $j$ with states of the transducer in each bit and checking their correct succession during the transduction tests.
	Since the transducer's control location is available for each position, the copy can apply the transducer's transition function to the current bit and nondeterministically choose one of the possible tuples of new control location and output bit.
	It then checks for the new control location at the next position in $w_h$.
	It also checks for the output bit in $w_{h+1}$ by again disjunctively testing for each position whether the bit is present or the index is different from $n$.
	
	In the \textbf{inductive step}, we show the claim for $k + 1$ under the induction hypothesis $(IH)$ for $k$.
	For the proof of $(1)$ we swap the roles of direction $i$ and $j$ and use the induction hypothesis:
	as argued before, the first index found being $0...0$, the last index being $1...1$ and each index being obtained from the last by adding $1$ can be described by regular languages (resp. transducers of constant size). 
	Hence, claim $(1)$ for $k+1$ follows from claim $(2)$ for $k$.
	It remains to argue that claim $(2)$ for $k + 1$, which we will show next, and claim $(1)$ for $k$, which we have used here, are consistent with each other since they formulate different requirements for the same direction.
	This is the case since a higher level nested index encoding contains all level $k-1$ encoded indices from $0...0$ to $1...1$ in ascending order.
	They can easily be identified by the level $k - 1$ brackets $[_{k-1}$ and $]_{k-1}$.
	
	For claim $(2)$, we first have to show that direction $j$ contains a sequence of words correctly level $k$ encoded.
	Like in the base case, this is done in a similar way as claim $(1)$ is enforced, namely by checking presence of level $k - 1$ encodings of $0...0$ to $1...1$ in ascending order while skipping bits belonging to the encoded word itself.
	Note that $(IH)$ applies here since we can use the level $k - 1$ indices of the level $k$ encodings that are present in direction $i$ as we have shown already.
	The main difference to the argument for claim $(1)$ is that we have to start on direction $j$ instead of $i$, conjunctively move to the start of each word, and then perform a context switch to be able to use $(IH)$.
	This results in $k + 1$ context switches for this case and thus does not violate the context restriction.
	Also (just like in the base case), item one and two of claim $(2)$ can easily be achieved by applying the language acceptors to the bits on the highest level of the encodings.
	The third item of $(2)$ can also be shown as in the base case since after moving from direction $i$ to $j$, there are $k$ context switches left which allows us to apply \autoref{lem:contextboundedcomparison}.
\end{proof}

Now, we are finally able to show the desired result.

\begin{proof}[Proof of \autoref{thm:contextboundedcomplexity}]
Inclusion follows from the construction underlying \autoref{cor:contextboundedtoapa}.
For hardness, we reduce from the acceptance problem for  $g(k-2, p(n))$ space bounded Turing machines using \autoref{lem:contextboundedtransducer} as explained above.
\end{proof}

As mentioned in the introduction, our constructions and completeness proofs can also be applied to alternating asynchronous automata (multitape automata) on finite words: the definitions of $k$-synchronicity and $k$-context-boundedness carry over in a direct manner and our algorithms for the emptiness tests can be applied. 
Furthermore, our hardness proofs can be transferred because our reductions only require reachability of control states and are therefore not dependent on a parity condition.
Thus, we believe our results also shed new light on the theory of multitape automata.

\section{A $\mu$-calculus for Hyperproperties}\label{sec:hypercalculus}

We now define our new logic, $\Hmu$.
In order to capture asynchronous hyperproperties, we combine ideas from HyperLTL \cite{Finkbeiner2015}, the polyadic $\mu$-calculus \cite{Andersen1994}, and the linear time $\mu$-calculus \cite{Barringer1986,Vardi1988} in a novel fashion. 
From HyperLTL we take the ideas to quantify over path variables and to relativise atomic propositions to path variables. Inspired by the indexed next-operator of the polyadic $\mu$-calculus, we relativise the next-operator to progress only on a single path identified by a path variable. Finally, complex hyperproperties can be specified by fixpoints. In this way, we extend the means provided by the linear time $\mu$-calculus for specifying properties to hyperproperties. 
Note that none of the logics that inspired the design of $\Hmu$ is able to capture asynchronous hyperproperties.
We use the following syntax:

\begin{definition}[Syntax]
	Let $\AP$ be a set of atomic propositions, $N = \{\pi_1,...,\pi_n\}$ be a set of path variables and $\chi = \{X_1,...,X_m\}$ be a set of predicates.
	We define $\Hmu$ formulas over $\AP$, $N$ and $\chi$ by the following grammar:
	\begin{align*}
	\varphi&:= \exists \pi . \varphi | \forall \pi . \varphi | \psi \\
	\psi &:= a_\pi | X | \bigcirc_\pi \psi | \psi \lor \psi | \lnot \psi | \mu X . \psi
	\end{align*}
	where $a \in \AP$, $\pi \in N$ and $X \in \chi$.
	We call expressions generated by the nonterminal $\varphi$ quantified formulas and expressions generated by the nonterminal $\psi$ quantifier-free formulas.
\end{definition}

In this paper, we will use two representations of a $\Hmu$ formula $\varphi$.
The first and more common one is its syntax tree which we denote by $rep_t(\varphi)$.
For the second one, we compress the syntax tree into a syntax directed acyclic graph (DAG) where syntactically equivalent subformulas share the same nodes and write $rep_d(\varphi)$ for this representation.
This offers an exponentially more succinct representation for some families of formulas while not increasing the complexity of the algorithms we consider.
Then, we use $|\varphi|_t := |rep_t(\varphi)|$ and $|\varphi|_d := |rep_d(\varphi)|$ as two measures for the size of a formula.
Note that since the DAG is obtained from the syntax tree by compressing it, $|\varphi|_d \leq |\varphi|_t$ holds for all formulas $\varphi$.
Therefore, all complexity upper bounds in which the size of a formula is measured by $|\cdot|_d$ trivially transfer to complexity upper bounds for the other size measure.

We add common connectives as syntactic sugar: $\true \equiv a_{\pi} \lor \lnot a_{\pi}$, $\false \equiv \lnot \true$, $\psi \land \psi' \equiv \lnot (\lnot \psi \lor \lnot\psi')$, $\psi \rightarrow \psi' \equiv \lnot \psi \lor \psi'$, $\psi \leftrightarrow \psi' \equiv \psi \rightarrow \psi' \land \psi' \rightarrow \psi$ and $\nu X. \psi \equiv \lnot \mu X. \lnot \psi[\lnot X/X]$.

Using these additional connectives, we impose some additional constraints on the syntax of $\Hmu$ formulas.
First, we assume that in a quantified formula $\varphi$ all predicates are bound by a fixpoint operator.
We also assume a \textit{strictly guarded} form, where predicates $X$ are only allowed in scope of an even number of negations inside $\mu X. \varphi$ and are directly preceeded by $\bigcirc_{\pi}$ for some $\pi$.
The latter part of this can indeed be required without loss of generality: if one constructs a formula where there is no progress through $\bigcirc_{\pi}$ between $\mu X$ and $X$, the fixpoint can equally be eliminated; if there is progress through $\bigcirc_{\pi}$ for some $\pi$, the $\bigcirc_{\pi}$ operator can be moved inwards such that it directly occurs in front of $X$.
Second, a formula is in \textit{positive normal form} when $\lnot$ only occurs in front of atomic propositions and all bound predicates and path variables are distinct.
Finally, we say that a formula is in \textit{closed} form when all path variables and predicates are bound.

Since our logic extends the linear time $\mu$-calculus with path quantification from HyperLTL, the formula constructs behave in similar ways as they do in those two logics.
The main difference to the linear time $\mu$-calculus is that a formula reasons about a set of paths or traces instead of over a single path or trace.
Thus, constructs $a_{\pi}$ and $\bigcirc_{\pi} \psi$ are indexed to express that $a$ holds on path $\pi$ or that $\psi$ holds when path $\pi$ moves one step forward.
Indexing the $\bigcirc$-operator allows us to express asynchronous behaviour.
Path quantification $\exists \pi . \varphi$ or $\forall \pi . \varphi$ allows to require that for one or for all paths $\pi$, the set of paths obtained by adding $\pi$ to the previously considered set fulfills $\varphi$.
Boolean connectives are interpreted in the standard way.
Finally, the constructs $X$ and $\mu X. \psi$ allow us to formulate iterative properties by least fixpoint constructions. 

With this logic, it becomes possible to specify asynchronous hyperproperties, i.e.\ properties that do not rely on traversing different paths lockstepwise.
We now sketch a few scenarios in which this is useful.
One potential application of hyperlogics is the analysis of multithreaded software.
In this scenario, different path variables used in a formula refer to different threads of the system and both the interaction between threads as well as the specification are captured by the formula.
For example, let $\pi_0$ and $\pi_1$ refer to executions of two different threads of a system that synchronise via locking.
Then the formula $\mu X. (\psi_{\textit{error}} \lor (\psi_{\textit{move}_0} \land \bigcirc_{\pi_0} X) \lor (\psi_{\textit{move}_1} \land \bigcirc_{\pi_1} X))$ expresses that the two threads can reach an error state identified by a formula $\psi_{error}$ through a lock-sensitive interleaving.
In this example, the atomic proposition $\textit{lock}_{\pi_i}$ indicates that thread $i$ currently holds the lock and the formula $\psi_{\textit{move}_i} = \lnot \textit{lock}_{\pi_{1-i}} \lor \lnot \bigcirc_{\pi_i} \textit{lock}_{\pi_i}$ expresses that thread $i$ can perform a step.
Such a property clearly requires asynchronous traversal of the different paths.
More complex interaction strategies can be handled by modifying the formulas $\psi_{\textit{move}_i}$.

Asynchronicity is also useful in applications of hyperlogics in security, e.g.\ when steps observed by the environment do not correspond to the same number of steps in different paths of the model.
This occurs, for instance, in models of software systems that reflect internal computations.
In such a situation, the formula $\varphi = \exists \pi . \forall \pi' . \nu X . \mu Y . ((a_{\pi} \leftrightarrow a_{\pi'}) \land \bigcirc_{\pi}\bigcirc_{\pi'} X) \lor \bigcirc_{\pi'} Y$ expresses that there is a path $\pi$ such that for all paths $\pi'$, $\pi$ and $\pi'$ repeatedly agree on the atomic proposition $a$.
However, contrary to the HyperLTL formula $\exists \pi . \forall \pi' . \mathcal{G}(a_{\pi} \leftrightarrow a_{\pi'})$ a step on $\pi$ can be matched by an arbitrary number of steps on $\pi'$.
This illustrates how $\Hmu$ allows us to relate a path $\pi$
describing expected observable behaviour with paths $\pi'$ with additional unobservable steps.
A similar technique can be used to specify asynchronous variants of classical hyperproperties from the literature like generalised non-interference or observational determinism \cite{Clarkson2010}.
The formula $\varphi_{\textit{obs}} = \forall \pi \forall \pi' . ((\psi_{\textit{eqL}} \rightarrow \nu X . \mu Y . \psi_{\textit{eqL}} \land \bigcirc_{\pi}\bigcirc_{\pi'} X) \lor (\lnot \textit{obs}_{\pi} \land \bigcirc_{\pi} Y) \lor (\lnot \textit{obs}_{\pi'} \land \bigcirc_{\pi'} Y)$ with $\psi_{\textit{eqL}} = \bigwedge_{a \in L} a_{\pi} \leftrightarrow a_{\pi'}$, for instance, specifies an asynchronous variant of observational determinism, which intuitively states that a system appears to be deterministic to a low security user who can only see propositions from the set $L$.
More specifically, it states that all pairs of executions which agree on the atomic propositions visible to a low security observer at the beginning of their computation agree on these atomic propositions in all observable situations.
Compared to its synchronous counterpart, it skips over unobservable states in both executions.


Later, in \autoref{sec:modelchecking} and \autoref{sec:satisfiability}, we adapt the two families of approximate analyses, $k$-synchronous and $k$-context-bounded analyses, from AAPA to $\Hmu$.
We now illustrate the utility of the resulting analyses.
%
%
Since any violation of the properties $\varphi$ and $\varphi_{\textit{obs}}$ from above occurs after a finite number of context switches, a $k$-context-bounded analysis with sufficiently large $k$ can be used to disprove these properties.
Another example is the property that one cannot locally distinguish whether fragments of a computation belong to one trace or another, or more precisely that for every point in $\pi_1$ there is a point in $\pi_2$ such that the next $n$ steps are indistinguishable.
In order to make this property's encoding in $\Hmu$ more readable, we introduce some additional syntactic sugar: $\mathcal{G}_{\pi} \psi \equiv \nu X . (\psi \land \bigcirc_{\pi} X)$ expresses that when progressing on $\pi$, some property $\psi$ \textit{generally} holds; $\mathcal{F}_{\pi} \psi \equiv \mu X . (\psi \lor \bigcirc_{\pi} X)$ expresses that when progressing on $\pi$, some property $\psi$ \textit{finally} holds.
With this syntactic sugar, the described property can be encoded as $\mathcal{G}_{\pi_1} \mathcal{F}_{\pi_2} \bigwedge_{i \leq n}\bigwedge_{a \in \AP} (\bigcirc_{\pi_2}^i a_{\pi_2} \leftrightarrow \bigcirc_{\pi_1}^i a_{\pi_1})$.
In this example, the $k$-context-bounded analyses with $k \geq 3$ yield precise results.
For the usefulness of the $k$-synchronous analyses, let us consider the following example:
imagine we expect a property described by a formula $\psi$ on $n$ paths $\pi_1,...,\pi_n$ in a synchronous manner after an initialisation phase that takes a different number of steps on each path.
Then, the formula $\mu X . ( \bigvee_{i \leq n} (Init_{\pi_i} \land \bigcirc_{\pi_i} X) \lor (\bigwedge_{i \leq n} (\lnot Init_{\pi_i}) \land \psi(\pi_1,...,\pi_n)))$ describes this expectation in a natural way and can be precisely analysed by the $k$-synchronous setup, if the lengths of the initialisation phases differ by at most $k$.
While the  two of the above properties can also be expressed using synchronous formulas, this most likely requires exponentially larger and less intuitive formulas.

For the definition of the semantics, we introduce some notation.
We call a function $\Pi: N \to \Paths(\mathcal{K})$ for a set of path variables $N$ a path assignment and $\mathcal{V}: \chi \to \PA \to 2^{\mathbb{N}_0^n}$ a predicate valuation, where we use $\PA$ to denote the set of all path assignments.
By $\Pi[\pi \mapsto p]$ (resp. $\mathcal{V}[X \mapsto M]$) we denote a path assignment $\Pi'$ (resp. predicate valuation $\mathcal{V}'$) with $\Pi'(\pi') = \Pi(\pi')$ for $\pi' \neq \pi$ and $\Pi'(\pi) = p$ (resp. $\mathcal{V}'(X') = \mathcal{V}(X')$ for $X' \neq X$, $\mathcal{V}'(X) = M$).
Also, we extend the notion of word shifts according to vectors from $\omega$-words to path assignments:
given a path assignment $\Pi$ binding path variables $\pi_1,...,\pi_n$ and a vector $v = (v_1,...,v_n) \in \mathbb{N}_0^n$, we use $\Pi[v]$ to denote the path assignment where each path $\pi_i$ is shifted left according to $v_i$.
For the definition of fixpoints, we define an order $\sqsubseteq$ on functions $\PA \to 2^{\mathbb{N}_0^n}$ such that $\xi \sqsubseteq \xi'$ iff $\xi(\Pi) \subseteq \xi'(\Pi)$ for all $\Pi$.
In this way, $(\PA \to 2^{\mathbb{N}_0^n},\sqsubseteq)$ forms a complete lattice with $\bot \equiv \lambda \Pi . \emptyset$ as its smallest element.

We now define the semantics of $\Hmu$. 
The semantics is indexed by a $k \in \mathbb{N}_0 \cup \{\infty\}$ restricting how far different paths can diverge from each other.
In the semantics indexed by $k$, we only consider situations where the foremost path is ahead of the rearmost path at most $k$ steps, or more formally, we restrict ourselves to index tuples from the set $\kvari := \{(j_1,...,j_n) \in \mathbb{N}_0^n | \forall i,i'. |j_i - j_{i'}| \leq k\}$.
Note that for $k = \infty$, $\kvari$ contains all index tuples from $\mathbb{N}_0^n$.
We distinguish between a quantifier semantics and a path semantics for quantified and quantifier-free formulas, respectively.
We write $\Pi \models^{\mathcal{K}}_k \varphi$ to denote that a path assignment $\Pi$ in the context of a Kripke Structure $\mathcal{K}$ satisfies a quantified formula $\varphi$ in the $k$-semantics.
This is extended to Kripke Structures: we write $\mathcal{K} \models_k \varphi$ iff $\{\} \models_k^{\mathcal{K}} \varphi$ for the empty path assignment $\{\}$.
Path semantics on the other hand applies to quantifier-free formulas $\psi$ with possibly free predicates.
It is defined in the context of predicate valuations and captures which index combinations from $\kvari$ fulfill the formula for the given path assignment.
We write $(j_1,...,j_n) \in \llbracket \varphi \rrbracket^{\mathcal{V}}_k(\Pi)$ for $(j_1,...,j_n) \in \kvari$ to denote that in the context of a predicate valuation $\mathcal{V}$, when we consider a path assignment $\Pi$ mapping the variables $\pi_1,...,\pi_n$ to paths $p_1,...,p_n$, the combination of suffixes $p_1[j_1],...,p_n[j_n]$ satisfies the formula $\varphi$.
Since no restrictions are imposed on index tuples for $k = \infty$, we omit the subscript in this situation and write $\mathcal{K} \models_{\infty} \varphi$ as $\mathcal{K} \models \varphi$, $\Pi \models_{\infty}^{\mathcal{K}} \varphi$ as $\Pi \models^{\mathcal{K}} \varphi$ and $\llbracket \psi \rrbracket^{\mathcal{V}}_{\infty}$ as $\llbracket \psi \rrbracket^{\mathcal{V}}$, respectively.

\begin{definition}[Quantifier Semantics]\label{def:quantifiersemantics}
	\begin{align*}
	\Pi \models_k^{\mathcal{K}} \exists \pi . \varphi  \text{ iff }& \Pi[\pi \mapsto p] \models_k^{\mathcal{K}} \varphi \text{ for some } p \in Paths(\mathcal{K}) \\
	\Pi \models_k^{\mathcal{K}} \forall \pi . \varphi  \text{ iff }& \Pi[\pi \mapsto p] \models_k^{\mathcal{K}} \varphi \text{ for all } p \in Paths(\mathcal{K}) \\
	\Pi \models_k^{\mathcal{K}} \psi  \text{ iff }& (0,...,0) \in \llbracket \psi \rrbracket_{k}^{\mathcal{V}}(\Pi) \text{ for some } \mathcal{V}
	\end{align*}
	for a quantified formula $\varphi$ and a quantifier-free formula $\psi$.
\end{definition}
\begin{definition}[Path Semantics]\label{def:pathsemantics}
	\begin{align*}
	\llbracket a_{\pi_i} \rrbracket^{\mathcal{V}}_{k} :=\;& \lambda \Pi . \{(j_1,...,j_n) \in \kvari \mid a \in L(\Pi(\pi_i)(j_i))\} \\
	\llbracket X \rrbracket^{\mathcal{V}}_{k} :=\;& \mathcal{V}(X) \\
	\llbracket \bigcirc_{\pi_i} \psi \rrbracket^{\mathcal{V}}_{k} :=\;& \lambda \Pi . \{(j_1,...,j_n) \in \kvari \mid 
	(j_1,...,j_i + 1,...,j_n) \in \llbracket \psi \rrbracket^{\mathcal{V}}_{k}(\Pi)\} \\
	\llbracket \psi \lor \psi' \rrbracket^{\mathcal{V}}_{k} :=\;& \lambda \Pi .\llbracket \psi \rrbracket^{\mathcal{V}}_{k}(\Pi) \cup \llbracket \psi' \rrbracket^{\mathcal{V}}_{k}(\Pi) \\
	\llbracket \lnot \psi \rrbracket^{\mathcal{V}}_{k} :=\;& \lambda \Pi . \kvari \setminus \llbracket \psi \rrbracket^{\mathcal{V}}_{k}(\Pi) \\
	\llbracket \mu X . \psi \rrbracket^{\mathcal{V}}_{k} :=\;& \bigsqcap \{\xi : PA \to 2^{\kvari} \mid \xi \sqsupseteq \llbracket \psi \rrbracket^{\mathcal{V}[X \mapsto \xi]}_{k}\}
	\end{align*}
\end{definition}

We now establish some properties of $\Hmu$'s semantics.
The first one is that the semantics of $\mu X. \psi$ indeed characterises a fixpoint.

\begin{theorem}\label{thm:monotone}
	$\alpha(\xi) := \llbracket \psi \rrbracket^{\mathcal{V}[X \mapsto \xi]}_k$ is monotone for all $k$, $\mathcal{V}$, $X$ and $\psi$ in positive normal form.
\end{theorem}

The Knaster-Tarski fixpoint theorem \cite{Tarski1955,Cousot1979} then gives a constructive characterisation of the semantics of fixpoint formulas via ordinal numbers:
\begin{corollary}\label{cor:welldefined}
	$\llbracket \mu X. \psi \rrbracket^{\mathcal{V}}_k$ is the least fixpoint of $\alpha$.
	It can be characterised by its approximants $\bigsqcup_{\kappa \geq 0} \alpha^\kappa(\bot)$ with $\alpha^0(\xi) = \xi$, $\alpha^{\kappa+1}(\xi) = \alpha(\alpha^\kappa(\xi))$ and $\alpha^{\lambda}(\xi) = \lambda \Pi. \bigcup_{\kappa < \lambda} \alpha^\kappa(\xi)(\Pi)$ where $\kappa$ ranges over ordinals and $\lambda$ over limit ordinals.
\end{corollary}

The second property is that the $k$-semantics properly approximates the full semantics.

\begin{theorem}\label{thm:monotonek}
	$\beta(k) := \llbracket \psi \rrbracket^{\mathcal{V}}_k$ is monotone for all $\psi$ in positive normal form and $\mathcal{V}$.
\end{theorem}

\begin{corollary}\label{cor:synchronoushierarchy}
	For all Kripke Structures $\mathcal{K}$, formulas $\varphi$ in positive normal form and $k,k' \in \mathbb{N}_0 \cup \{\infty\}$ with $k \leq k'$, we have: $\mathcal{K} \models_{k} \varphi$ implies $\mathcal{K} \models_{k'} \varphi$.
\end{corollary}

For the sake of reasoning about satisfiability, we define a variant of the semantics for $\Hmu$ formulas on traces in the straightforward way.
Instead of considering path assignments $\Pi: N \to \Paths(\mathcal{K})$, we instead consider trace assignments $\Pi: N \to \mathcal{T}$ for a set of traces $\mathcal{T}$.
Existential and universal quantifiers then quantify traces $t \in \mathcal{T}$ instead of paths $p\in \Paths(\mathcal{K})$.
Also, atomic propositions require the proposition to be included in the respective set of atomic propositions directly instead of in the set obtained by applying the labelling function to a state.
We write $\mathcal{T} \models \varphi$ to denote that a set of traces $\mathcal{T}$ fulfills a formula.
By choosing $\mathcal{T} = \Traces(\mathcal{K})$, the two semantics coincide.
For a more formal specification of this semantics variation, we refer the reader to the \switchExtended{extended version \cite{}}{appendix}.

Given these two variants of $\Hmu's$ semantics, we consider two decision problems:
\begin{itemize}
	\item \textit{Model Checking}: given a closed $\Hmu$ formula $\varphi$ and a Kripke Structure $\mathcal{K}$, does $\mathcal{K} \models \varphi$ hold?
	\item \textit{Satisfiability}: given a closed $\Hmu$ formula $\varphi$, is there a non-empty set of traces $\mathcal{T}$ such that $\mathcal{T} \models \varphi$ holds?
\end{itemize}
For this purpose we first establish a connection between $\Hmu$ and AAPA.
This allows us to apply the results on AAPA that we have established already.
In order to transfer results for restricted classes of AAPA we define corresponding fragments of the logic $\Hmu$ next.

\begin{definition}
	We call a $\Hmu$ formula $k$-synchronous for a Kripke Structure $\mathcal{K}$ if the following condition holds: $\mathcal{K} \models \varphi$ implies $\mathcal{K} \models_k \varphi$.
	A $0$-synchronous formula is called synchronous.
\end{definition}

As a small technical detail, $\llbracket \bigcirc_{\pi} \psi \rrbracket^{\mathcal{V}}_0$, i.e. the $0$-semantics of $\bigcirc_{\pi} \psi$ is always empty.
Strictly speaking, this makes the above definition useless for the case $k = 0$.
In order to cure this defect, we allow a synchronous next operator $\bigcirc \psi$ that advances all direction simultaneously in synchronous formulas.

\begin{definition}[Synchronous syntactic fragment]
	A $\Hmu$ formula belongs to the synchronous fragment if it uses the synchronous next operator $\bigcirc$ instead of the indexed next operators $\bigcirc_{\pi}$.
\end{definition}

For the next two fragments, we need a notion of \textit{extended syntax tree} of a formula $\varphi$.
Thereby, we mean the (infinite) tree obtained by repeatedly substituting fixpoint variables by their respective formula in the syntax tree, i.e. for every fixpoint expression $\mu X . \psi$ (or $\nu X . \psi$), replacing every occurrence of $X$ in $\psi$ with $\psi$.
It is easy to see that membership in the two syntactic fragments defined next is decidable in polynomial time.

\begin{definition}[$k$-synchronous syntactic fragment]
	A $\Hmu$ formula belongs to the $k$-synchronous syntactic fragment if the difference between the number of occurrences of $\bigcirc_{\pi}$ and $\bigcirc_{\pi'}$ for $\pi \neq \pi'$ on any path starting in the root of the formula's extended syntax tree is at most $k$.
\end{definition}

\begin{definition}[$k$-context-bounded syntactic fragment]
	A $\Hmu$ formula is from the $k$-context-bounded syntactic fragment if on every path starting in the root of the formula's extended syntax tree, the number of switches between directions $\pi$ for $\bigcirc_{\pi}$ constructs is at most $k-1$.
\end{definition}

\section{Expressive Equality}\label{sec:connection}

In this section we establish the correspondence between AAPA and $\Hmu$, a relation that is essential for transferring the results from \autoref{sec:automata} to $\Hmu$.
Crucial for this correspondence is the fact that a path assignment $\Pi$ for paths $\pi_1,...,\pi_n$ over the set of states $S$ can be encoded into a word over the alphabet $S^n$.
Then, free path variables in a formula correspond to components of the input alphabet and free path predicates correspond to holes in the automaton, where a fitting semantics can be plugged in once the predicate is bound.
Because of this correspondence, we use the same name for a predicate and its corresponding hole in an automaton.

Thus, given a path assignment $\Pi$ with $\Pi(\pi_i) = s_i^0 s_i^1 ...$, we define its translation into a word over the alphabet $\mathcal{S} = S^n$ as $w_{\Pi} = (s_1^0,...,s_n^0)(s_1^1,...,s_n^1)... \in \mathcal{S}^{\omega}$.
While viewing path assignments as such words, we handle predicate valuations $\mathcal{V}$ by languages $\mathcal{L}(\mathcal{V}(X_j)(\Pi)) = \{(w_{\Pi}[v] \in \mathcal{S}^{\omega} | v \in \mathcal{V}(X_j)(\Pi)\}$.
For these languages to be well-defined, we need to restrict ourselves to \textit{well-formed} valuations with the property that for all vectors $v,v' \in \mathbb{N}_0^n$, path assignments $\Pi,\Pi'$ and predicates $X$, we have that $\Pi[v] = \Pi'[v']$ implies $v \in \mathcal{V}(X)(\Pi)$ iff $v' \in \mathcal{V}(X)(\Pi')$.
However, we can show via induction that only valuations with this property occur during fixpoint iterations of $\Hmu$.

We introduce the notion of $\mathcal{K}$-equivalence:

\begin{definition}[$\mathcal{K}$-equivalence]
	Given a Kripke Structure $\mathcal{K} = (S,s_0,\delta,L)$, a $\Hmu$ formula $\psi$ over $\{\pi_1,...,\pi_n\}$ with free predicates $X_1,...,X_m$ and an alternating (asynchronous) parity automaton $\mathcal{A}$ with holes $X_1,...,X_m$ over the alphabet $S^n$, we call $\mathcal{A}$ $\mathcal{K}$-equivalent to $\psi$, iff the following condition holds: for all path assignments $\Pi$, well-formed predicate valuations $\mathcal{V}$ and vectors $v \in \mathbb{N}_0^n$, we have $v \in \llbracket \psi \rrbracket^{\mathcal{V}}(\Pi)$ iff $w_{\Pi}[v] \in \mathcal{L}(\mathcal{A}[X_1 : \mathcal{L}(\mathcal{V}(X_1)(\Pi)) ,...,X_m : \mathcal{L}(\mathcal{V}(X_m)(\Pi))])$.
\end{definition}

The definition of $\mathcal{K}$-equivalence is straightforwardly extended to quantified formulas $\varphi$: $\mathcal{A}$ is called $\mathcal{K}$-equivalent to $\varphi$ iff for all $\Pi$ the statements (i) $\Pi \models^{\mathcal{K}} \varphi$ and (ii) $w_{\Pi} \in \mathcal{L}(\mathcal{A})$ are equivalent.

This notion allows us to formulate the following theorem:

\begin{theorem}\label{thm:aapahmuequivalence}
	Let $\mathcal{K}$ be a Kripke Structure.
	\begin{enumerate}
		\item For every quantifier-free $\Hmu$ formula $\psi$ in positive normal form, there is an AAPA $\mathcal{A}_{\psi}$ of linear size in $|\psi|_d$ (and hence also in $|\psi|_t$) such that $\mathcal{A}_{\psi}$ is $\mathcal{K}$-equivalent to $\psi$.
		\item For every AAPA $\mathcal{A}$ over the alphabet $\Sigma$, there is a quantifier-free $\Hmu$ formula $\psi_{\mathcal{A}}$ with $|\psi_{\mathcal{A}}|_d$ linear and $|\psi_{\mathcal{A}}|_t$ exponential in $|\mathcal{A}|$ such that $\mathcal{A}$ is $\mathcal{K}$-equivalent to $\psi_{\mathcal{A}}$.
	\end{enumerate}
\end{theorem}

We dedicate the following two subsections to the constructions underlying the proof of this theorem.
\switchExtended{}{A detailed correctness proof for these constructions can be found in the appendix.}

\subsection{From $\Hmu$ to AAPA: Construction for \autoref{thm:aapahmuequivalence}, Part 1}\label{subsec:hmutoaapa}

Intuitively, the AAPA $\mathcal{A}_{\psi}$ has a node for each node in $\psi$'s syntax DAG.
For most constructs, the node can straightforwardly check the semantics of the formula, either directly or by transitioning to nodes for subformulas in a suitable manner.
The most interesting cases are those for predicates and fixpoint expressions.
As in the definition of $\mathcal{K}$-equivalence, free predicates correspond to holes in the automaton and are thus translated to such.
In the same definition, bound predicates do not occur as holes in the automaton; indeed, the construction for fixpoints fills the corresponding holes.
This is done by a backwards edge to the start of the automaton.
Taking this backwards edge corresponds to one unfolding of a fixpoint iteration.
Depending on whether we have a least or greatest fixpoint, the iteration can be performed a finite or infinite number of times.
This is captured by assigning the predicate state a priority that is lower than any other in the automaton and that is odd for least and even for greatest fixpoints.

We inductively construct the AAPA $\mathcal{A}_{\psi} = (Q,\rho_0,\rho,\Omega)$ for a formula $\psi$.
In each step, we assume that automata for all subformulas of $\psi$ are already constructed.
The construction is linear in $|\psi|_d$ since in each step at most a constant number of states is added to the already existing partial automata and automata for syntactically (and thus semantically) equivalent subformulas can be shared.
We now describe each case of the inductive construction.
Recall that we assume that states $\true$ and $\false$ are present in every automaton.
We therefore do not mention them explicitly.

\begin{description}
\item[Case $\psi = a_{\pi_i}$:]${}$\\ 
	For atomic propositions, we let the automaton enter an accepting self-loop if the proposition is fulfilled at the current index and a rejecting self-loop otherwise. 
	We set $Q = \{a_{\pi_i}\}$ and $\rho_0 = a_{\pi_i}$.
	Furthermore, we let	$\rho(a_{\pi_i},s,d) = \true$ if $a \in L(s)$ and $d = i$ and let $\rho(a_{\pi_i},s,d) = \false$, 
	otherwise. 
	For the priority assignments, we set 
	 $\Omega(a_{\pi_i}) = 1$.
	The case $\psi = \lnot a_{\pi_i}$ is analogous.
\item[Case $\psi = X$:]${}$\\
	A predicate is transformed into a hole in the automaton: we set
	$Q = \{X\}$,
	$\rho_0 = X$,
	$\rho(X,s,d) = \bot$ and 
	$\Omega(X) = \bot$. 
\item[Case $\psi = \psi_1 \lor \psi_2$:]${}$\\
From the induction hypothesis, we have automata $\mathcal{A}_{\psi_i} = (Q_i,\rho_{0,i},\rho_i,\Omega_i)$ for the formulas $\psi_i$. 
We set 
	$Q = Q_1 \cup Q_2$ and 
	$\rho_0 = \rho_{0,1} \lor \rho_{0,2}$.
	The transition function is induced by the transition functions of the automata $\mathcal{A}_i$, i.e. 
	$\rho(q,s,d) = \rho_i(q,s,d)$ for $q \in Q_i$.
	For the priorities, we pick
	$\Omega(q) = \Omega_i(q)$ for $q \in Q_i$.
	This automaton accepts iff one of the automata $\mathcal{A}_i$ accepts.
\item[Case $\psi = \bigcirc_i \psi_1$:] ${}$\\
From the induction hypothesis, we obtain $\mathcal{A}_{\psi_1} = (Q_1,\rho_{0,1},\rho_1,\Omega_1)$ for $\psi_1$. 
We then set
	$Q = \{\psi\} \cup Q_1$ and 
	$\rho_0 = \psi$.
	Furthermore, we set $\rho(\psi,s,d) = 
	\rho_{0,1}$ for  $d = i$ and $\rho(\psi,s,d) = \false$ otherwise.
	For the states $q \in Q_1$ and priorities, we choose
	$\rho(q,s,d) = \rho_1(q,s,d)$,
	$\Omega(\psi) = 0$ and $\Omega(q) = \Omega_1(q)$ for $q \in Q_1$.
	The choice of priority for state $\psi$ does not matter, since every run visiting this node will either visit it only a finite number of times or visit states with lower or equal priority infinitely often.
	This is ensured by the construction for fixpoint formulas.
	
\item[Case $\psi = \mu X. \psi_1$ or $\psi = \nu X. \psi_1$:]${}$\\
By the induction hypothesis, we have an automaton $\mathcal{A}_{\psi_1} = (Q_1,\rho_{0,1},\rho_1,\Omega_1)$ for $\psi_1$. 
We can assume w.l.o.g. that there is only one hole for the path predicate $X$, since all holes have the same behaviour.
Let $p := min\{\Omega_1(q_i) | q_i \in Q_1\}$.
Let $p_{even} = p$ and $p_{odd} = p - 1$ if $p$ is even; otherwise, let $p_{even} = p - 1$ and $p_{odd} = p$.
Intuitively, $p_{even}$ (resp. $p_{odd}$) is an even (resp. odd) lower bound on the priorities occuring in $\mathcal{A}_{\psi}$.
For the case where the priority used in the automaton below is negative, all priorities in the automaton are shifted by a multiple of $2$ such that the lowest priority is $0$ or $1$.
We set $Q = Q_1$, $\rho_0 = X$, $\rho(X,s,d) = \rho_1(\rho_{0,1},s,d)$ and $\rho(q,s,d) = \rho_{1}(q,s,d)$ for $q \in Q_1 \setminus \{X\}$.
For the priorities, we choose $\Omega(q) = \Omega_1(q)$ for $q \in Q_1 \setminus \{X\}$.
The state $X$ is assigned the priority 	$\Omega(X) = p_{odd}$ if  $\psi$ is a $\mu$ formula and $\Omega(X) = p_{even}$ otherwise. 
In this definition, we use $\rho_1(\rho_{0,1},s,d)$ to denote a variant of $\rho_{0,1}$ in which every occurrence of $q$ is substituted by $\rho_1(q,s,d)$.
The automaton evaluates the subformula $\psi_1$, but switches to the start in case a predicate $X$ is encountered. The choice of priority for $X$ reflects that an infinite unfolding of a $\mu$-formula should lead to a rejecting run, while an infinite unfolding of a $\nu$-formula should lead to an accepting run.
\end{description}

\subsection{From AAPA to $\Hmu$: Construction for \autoref{thm:aapahmuequivalence}, Part 2}\label{subsec:aapatohmu}

Given an AAPA $\mathcal{A} = (Q, \rho_0, \rho, \Omega)$, we construct formulas $\psi_i^h$ by induction on $i$.
Our construction is inspired by a construction from \cite{Bozzelli2007} in the context of a fixpoint logic for visibly pushdown languages.\footnote{More precisely, the construction is found in the appendix of an extended version of that paper that was provided to us by the author.
While the paper states that the construction can be performed in linear time and thus with at most linear size increase, it is not discernable how the construction can be performed without an at worst exponential blowup when using the syntax tree of a formula as the basis for measuring its size.}
We fix an ordering $q_1,...,q_n$ on states of $Q$ such that $\Omega(q_i) \geq \Omega(q_j)$ for $i < j$ for non-hole states and holes are the last $m$ states of the ordering.
Intuitively, the formula $\psi_0^h$ describes the local behaviour of a state $q_h$ and the formula $\psi_i^h$ expresses the existence of an accepting run of $\mathcal{A}$ starting in $q_h$, where only states with a priority higher than $\Omega(q_i)$ are visited infinitely often.
For each state $q_i$, we introduce a predicate $X_i$ which is bound in the $i$-th step of the inductive construction.
In the construction, $i$ ranges from $0$ to $n-m$ and $h$ ranges from $1$ to $n$.
Therefore, when the construction is finished, only the hole states of $\mathcal{A}$ remain as unbound predicates in the formula and we can choose $\rho_0[q_1 / \psi_{n-m}^1]...[q_n / \psi_{n-m}^n]$ as the desired formula.

\textbf{Construction of $\psi_0^h$:}
for holes, that is for $h > n-m$, the formula $\psi_0^h$ is given as the predicate $X_h$.
For non-holes, that is for $h \leq n-m$, we first construct $\hat{\rho}(q_h,\sigma,d)$ from $\rho(q_h,\sigma,d)$ by substituting every occurrence of a state $q_l$ with $X_l$.
$\psi_0^h$ is then given as $\bigvee_{\sigma \in \Sigma} \bigvee_{d \in M} (\sigma_{\pi_d} \land \bigcirc_d \hat{\rho}(q_h,\sigma,d))$ to describe that some $\sigma$ is currently being read in direction $d$ and in the next step we are in some combination of states satisfying $\rho(q_h,\sigma,d)$ after moving on in direction $d$.

\textbf{Construction of $\psi_i^h$ for $i > 0$:}
we assume that $\psi_{i-1}^h$ for all $h$ is already constructed.
As a first step, we construct the formula $\psi_i^i$ by binding the predicate $X_i$.
We differentiate two cases based on the priority of state $q_i$.
If $\Omega(q_i)$ is odd, then $\psi_i^i$ is given as $\mu X_i . \psi_{i-1}^i$.
In the other case, where $\Omega(q_i)$ is even, we construct $\psi_i^i$ as $\nu X_i . \psi_{i-1}^i$.
Then we construct $\psi_i^h$ for all $h \neq i$ by substituting $\psi_i^i$ for every occurrence of the predicate $X_i$ in the previous formula, that is: $\psi_i^h = \psi_{i-1}^h [X_i / \psi_i^i]$.

\textbf{Size of the construction:}

In the base case of the construction, a number of syntax DAG nodes is created that is linear in the size of $\rho$.
In the inductive step, only a single node is added, that is the node for the fixpoint in $\psi_i^i$.
The other formulas in this step, namely $\psi_i^h$ for $h \neq i$, can be obtained from $\psi_{i-1}^h$ by redirecting edges in the syntax DAG.
Since only a single node is added in each step and there is a step for each node of the automaton, $|Q|$ nodes are added in total.
Finally, when combining the formulas, at most $|\rho_0|$ nodes are added.
Thus, we have $|\psi_{\mathcal{A}}|_d = \mathcal{O}(|\rho| + |Q| + |\rho_0|) = \mathcal{O}(|\mathcal{A}|)$.

For the second measure of size, $|\psi_{\mathcal{A}}|_t$, the formula must be represented as a syntax tree.
Here, the substitution in the construction of $\psi_i^h$ for $h \neq i$ cannot be performed without adding nodes to the syntax tree since the nodes representing $\psi_i^i$ must be duplicated.
This results in a linear increase in each of the steps, which leads to an overall exponential size increase in the worst case.

\subsection{The Consequences of \autoref{thm:aapahmuequivalence}}

Since the emptiness problem for AAPA is $\Sigma_2^1$-hard and we can effectively build a $\Hmu$ formula for every AAPA and check it against a structure that generates $\Sigma^{\omega}$, we obtain:

\begin{theorem}
Model checking $\Hmu$ against a Kripke model is $\Sigma_2^1$-hard.
\end{theorem}

Likewise, the emptiness problem of AAPA can be reduced to the satisfiability problem for $\Hmu$:

\begin{theorem}\label{thm:satisfiability}
The satisfiability problem for $\Hmu$ is $\Sigma_2^1$-hard.
\end{theorem}

\section{Approximate Model Checking and Decidable Fragments}\label{sec:modelchecking}

\begin{table}
	\centering
		\begin{tabular}{l l l}
			\toprule
			fragment & complexity & \\
			\midrule
			full & $\UNDECIDABLE$ & \\
			synchronous & $\SPACE(g(d,|\varphi|_d))$-complete & $\NSPACE(g(d-1,|\mathcal{K}|))$-complete\\
			$k$-synchronous & $\SPACE(g(d+1,|\varphi|_d))$-complete & $\NSPACE(g(d-1,|\mathcal{K}|))$-complete\\
			$k$-context-bounded & $\SPACE(g(d+k-2,|\varphi|_d))$ & $\NSPACE(g(d-1,|\mathcal{K}|))$ \\
			& $\SPACE(g(\tilde{d}+k-2,|\varphi|_d))$-hard & $\tilde{d} = min(d,1)$ \\
			\bottomrule
		\end{tabular}
	\caption{Complexity results for model checking.} 
	\label{table:complexity}
\end{table}

The $\Sigma_2^1$-hardness of model checking and satisfiability imply that these problems can neither be algorithmically solved nor can a complete approximation procedure (e.g. via bounded model checking) be developed for the full logic.
However, the concepts of $k$-synchronicity and $k$-context-boundedness have already proved to lead to decidability for AAPA.
We therefore consider similarly restricted analyses for $\Hmu$ and explain how they relate to the matching analyses of AAPA.

For model checking $\Hmu$ against Kripke structures, two main problems have to be solved: the quantifier-free part of the input formula has to be suitably represented and the representation has to take the quantifiers into account.
As representation, we use the AAPA derived from the construction of \autoref{subsec:hmutoaapa}. 
The main idea for handling quantifiers is to take a nondeterministic automaton representing the inner formula
and to form a product with the Kripke structure in order to check for the existence of appropriate paths in the structure.
Additionally, the automata need to be complemented in the case of universal quantifiers.
Both $k$-synchronous and $k$-context-bounded AAPA are suitable for this purpose as we can dealternate them to NPA. 
In this way, we lift our approximate analyses from automata to formulas.
Moreover, we even obtain precise results on several fragments of the logic.
This stems from the following facts:

\begin{theorem}\label{thm:constructionproperties}
	\begin{itemize}
			\item The automaton $\mathcal{A}_{\psi}$ for a synchronous $\Hmu$ formula $\psi$ can be transformed into an APA of asymptotically the same size.
			\item The automaton $\mathcal{A}_{\psi}$ for a $k$-synchronous $\Hmu$ formula $\psi$ is a $k$-synchronous AAPA.
			\item The automaton $\mathcal{A}_{\psi}$ for a $k$-context-bounded $\Hmu$ formula $\psi$ is a $k$-context-bounded AAPA.
	\end{itemize}
\end{theorem}

Let us first consider the $k$-synchronous analysis.
It can be applied to formulas by using the following fact:

\begin{theorem}\label{thm:approximateconstructionproperties}
	For a quantifier-free formula $\psi$, a path assignment $\Pi$ and a well-formed predicate valuation $\mathcal{V}$, the automaton $\mathcal{A}_{\psi}$ from \autoref{thm:aapahmuequivalence} has a $k$-synchronous accepting run over $w_{\Pi}[v]$ with holes filled according to $\mathcal{V}$ iff $v \in \llbracket \psi \rrbracket_{k}^{\mathcal{V}}(\Pi)$.
\end{theorem}

This generalises the definition of $\mathcal{K}$-equivalence.
The proof relies on the observation that each node in $\mathcal{A}_{\psi}$'s runs corresponds to a subformula of $\psi$ and that for each such subformula, the offset counters that appear in an accepting run correspond to the tuples in the semantics of $\psi$.
Thus, the restriction induced by the $k$-semantics and $k$-synchronous runs deliver the same tuples.

\autoref{thm:approximateconstructionproperties} implies that we can use the $k$-synchronous analysis underlying \autoref{thm:synchronousemptiness} on the automaton $\mathcal{A}_{\psi}$ for a formula $\psi$ to determine that formula's $k$-semantics.
Since the approximations given by the $k$-semantics improve with increasing $k$ for quantifier-free (\autoref{thm:monotonek}) as well as quantified (\autoref{cor:synchronoushierarchy}) formulas, this procedure yields increasingly precise analyses of arbitrary $\Hmu$ formulas.
Indeed, any approximate analysis parameterised in some $k$ on AAPA that supplies us, for increasing $k$, with an increasing number of offsets $v$ from which there is an accepting run for a fixed $w_{\Pi}$ can be used to approximate $\Hmu$ as well.
This is due to the fact that such an analysis directly corresponds to a monotone semantics approximator as in \autoref{thm:monotonek} and is thus applicable in a proof like the one from \autoref{cor:synchronoushierarchy}.
In particular, this applies to the $k$-context-bounded analysis underlying \autoref{thm:contextboundedcomplexity} using the above argument about the correspondence between offset counters and tuples in the semantics.

\subsection{Constructions for Quantifiers}

Given a $k$-synchronous  or $k$-context-bounded AAPA for a quantified formula, the quantifier can be handled in the same way as for HyperLTL in \cite{Finkbeiner2015}.

\textbf{Construction for existential quantifiers:}
we perform the construction for a formula $\exists \pi_{n+1} . \varphi$ and a structure $\mathcal{K} = (S,s_0,\delta,L)$.
Our construction simulates one input component of $\mathcal{A}_{\varphi}$ in the state space of $\mathcal{A}_{\exists_{\pi_{n+1}} . \varphi}$.
In order to allow us to do so, we assume $\mathcal{A}_{\varphi}$ to be given as an NPA.
For the fragments we consider, a translation to NPA is possible, as seen in \autoref{sec:automata}.
Therefore, $\mathcal{A}_{\varphi}$ has the form $(Q_1,q_{0,1},\rho_1,\Omega_1)$ and input alphabet $S^{n+1}$.
We construct the NPA $\mathcal{A}_{\exists \pi_{n+1} . \varphi} = (Q,q_0,\rho,\Omega)$ with input alphabet $S^n$,
 	states $Q = Q_1 \times S$,
	initial state $q_0 = (q_{0,1},s_0)$,
	transition function $\rho((q,s),\mathpzc{s}) = \{(q',s') \in Q_1 \times S \mid q' \in \rho_1(q,\mathpzc{s} + s), s' \in \delta(s)\}$,
	and priorities $\Omega((q,s)) = \Omega_1(q)$.
	
As mentioned, the last component of $\mathcal{S}$ is now simulated in the last component of the state space of $\mathcal{A}_{\exists \pi_{n+1} . \varphi}$.
Simultaneously this component makes sure that transitions are taken according to the transition function $\delta$ of $\mathcal{K}$.
Choosing $(q_{0,1},s_0)$ as the starting state ensures that the path which is simulated in this way starts in the initial state $s_0$ of $\mathcal{K}$.

\textbf{Construction for universal quantifiers:}
in order to handle universal quantifiers, we complement the automaton $\mathcal{A}_{\varphi}$ by swapping $\land$ and $\lor$ in transitions and increasing all priorities by one.
We then construct the automaton $\exists \pi_{n+1}. \lnot\varphi$, and complement it again by the same procedure.
Note that when several of these constructions are combined, double negations can be cancelled out to avoid unnecessary complementation constructions.

\subsection{Complexity Analysis}

In order to analyse the size of the resulting automata, we need a notion of alternation depth, which is defined as the number of switches between $\exists$ and $\forall$ quantifiers in the quantifier-prefix of the formula.
Since our quantifier construction involves transforming the automaton into an NPA, we formulate our evaluation of $\mathcal{A}_{\varphi}$'s size for the resulting NPA: 

\begin{lemma}\label{lem:constructionsize}
	The NPA $\mathcal{A}_{\varphi}$ for a closed formula $\varphi$ with alternation depth $d$ has size:
	\begin{itemize}
		\item $\mathcal{O}(g(d,|\mathcal{K}| \cdot g(1,|\varphi|_d)))$ in a synchronous analysis,
		\item $\mathcal{O}(g(d,|\mathcal{K}| \cdot g(2,|\varphi|_d)))$ in a $k$-synchronous analysis, and
		\item $\mathcal{O}(g(d,|\mathcal{K}| \cdot g(k-1,|\varphi|_d)))$ in a $k$-context-bounded analysis.
	\end{itemize}
\end{lemma}

For the decision problems corresponding to our analyses, we obtain the following complexity results.
The upper bounds also apply for the respective analyses.

\begin{theorem}
	Model checking a closed synchronous $\Hmu$ formula $\varphi$ with alternation depth $d$ against a Kripke Structure is complete for $\SPACE(g(d,|\varphi|_d))$, $\SPACE(g(d,|\varphi|_t))$ and $\NSPACE(g(d-1,|\mathcal{K}|))$\footnote{
	For the case of $d = 0$, we use the definition from \cite{Finkbeiner2015}, where $\NSPACE(g(-1, n))$ was defined as $\NLOGSPACE$.
	For $d > 0$, we can use Savitch's Theorem to see that the problems are actually complete for the deterministic space classes.
	}.
\end{theorem}

\begin{proof}
	The upper bounds follow from \autoref{lem:constructionsize} and the $\NLOGSPACE$ complexity of emptiness tests on NPA.
	For hardness, we can reduce from HyperLTL Model Checking \cite{Rabe2016} to obtain the desired results.
\end{proof}

\begin{theorem}\label{thm:ksynchronousmodelchecking}
	Model checking a closed $k$-synchronous $\Hmu$ formula $\varphi$ with alternation depth $d$ against a Kripke Structure is complete for $\SPACE(g(d+1,|\varphi|_d))$ and $\NSPACE(g(d-1,|\mathcal{K}|))$. 
\end{theorem}

\begin{proof}
The upper bounds follow from \autoref{lem:constructionsize} and the $\NLOGSPACE$ complexity of emptiness tests on NPA.

For hardness, we reduce from the acceptance problem for $g(k,p(n))$ space bounded deterministic Turing machines.
The reduction is based on Sistla's classical yardstick construction for the satisfiability problem of Quantified Propositional Temporal Logic (QPTL) \cite{Sistla1983} and the adaptation of that reduction to HyperLTL \cite{Rabe2016}.
More concretely, in the reduction for HyperLTL, a formula $\varphi_{k,m}(P_x, P_y)$ is constructed such that $P_x$ and $P_y$ are true at exactly one point and the points at which both propositions are fulfilled are exactly $N_{k,m}$ steps away where $N_{k,m} \geq g(k,m)$.
While in that reduction, an alternation free formula is constructed for a polynomial $N_{k,m}$, we can instead construct a quantifier free formula for an exponential $N_{k,m}$ by building a $2$-synchronous AAPA as in the proof of \autoref{thm:synchronousemptiness} that can check for two propositions whether they are exponentially many indices away from each other.
For this AAPA, we can obtain a $2$-synchronous $\Hmu$ formula in polynomial time.
We can then inductively construct the formulas as in the proof for HyperLTL and translate it to $\Hmu$.
Note that since HyperLTL progresses on different paths synchronously, the inductive construction can easily be adapted to preserve $2$-synchronicity of the formula.
Accordingly, we need one quantifier alternation less in our $\Hmu$ formula than in the corresponding HyperLTL or QPTL formulas to build a yardstick of length $N_{k,m}$ and obtain the desired lower bound.

The hardness claim for fixed formulas can be obtained by a direct reduction from HyperLTL model checking \cite{Rabe2016}.
\end{proof}

By the space hierarchy theorem, the completeness for $\SPACE(g(d+1,|\varphi|_d))$ implies that when measuring size of formulas $\varphi$ by $|\cdot|_t$, at least space $\mathcal{O}(g(d,|\varphi|_t))$ is needed. Note that this does not imply hardness for the corresponding class since the implication is based on an exponential time reduction, but hardness requires polynomial time reductions. A similar reasoning applies to later theorems in which completeness results based on $|\cdot|_d$ are given.

\begin{theorem}\label{thm:contextboundedmodelchecking}
	Model checking a closed $k$-context-bounded $\Hmu$ formula $\varphi$ with alternation depth $d$ against a Kripke Structure is in $\SPACE(g(d+k-2,|\varphi|_d))$ and $\NSPACE(g(d-1,|\mathcal{K}|))$ and is hard for $\SPACE(g(\tilde{d}+k-2,|\varphi|_d))$ where $\tilde{d} = min(d,1)$. 
\end{theorem}

\begin{proof}
	The upper bounds follow from \autoref{lem:constructionsize} and the $\NLOGSPACE$ complexity of emptiness tests on NPA.
	For hardness, we show the claim for $d = 0$ and $d \geq 1$ separately.
	
	In the first case where $d = 0$, we reduce from the emptiness problem for $k$-context-bounded AAPA.
	More specifically, an AAPA $\mathcal{A}$ is non-empty iff $\mathcal{K} \models \exists \pi_1 \dots \exists \pi_n . \psi_{\mathcal{A}}$ for the structure $\mathcal{K}$ whose traces correspond to arbitrary words over $\mathcal{A}$'s input alphabet.
	
	In the second case, where $d \geq 1$, we reduce from the acceptance problem for $g(k-1,p(n))$ space bounded deterministic Turing machines.
	Similar to the proof of \autoref{thm:contextboundedcomplexity}, we show a more general result about regular transductions on level $k-1$ Stockmeyer encoded words as in \autoref{lem:contextboundedtransducer}.
	For a regular transducer and two regular language acceptors of size polynomial in $n$, we construct a Kripke Structure $\mathcal{K}$ and a formula $\varphi$ such that $\mathcal{K} \models \varphi$ iff there is a sequence of level $k-1$ Stockmeyer encoded words $w_0,w_1,...,w_l$ such that
	\begin{enumerate}
		\item[(a)] $w_0$ is accepted by the first regular language acceptor,
		\item[(b)] $w_l$ is accepted by the second regular language acceptor, and
		\item[(c)] for all $i < l$, $w_{i+1}$ is obtained from $w_{i}$ by applying the regular transducer $\mathcal{T} = (Q,q_o,\gamma)$.
	\end{enumerate}

	As the first step, we construct the Kripke Structure $\mathcal{K}$.
	Its set of atomic propositions is given as $AP = \{0,1\} \cup \{[_i , ]_i | i \leq k\} \cup Q$, where $\{0,1\}$ is the alphabet for the word itself, $\{[_i , ]_i | i \leq k\}$ is used for the Stockmeyer encodings and $Q$ is the set of states of the regular transducer.
	For this set of atomic propositions, we choose $\mathcal{K}$ as the structure that produces arbitrary traces over $AP$.
	
	To make the formula $\varphi$ more readable, we introduce some modalities as syntactic sugar that can easily be encoded in $\Hmu$.
	First, we use the LTL-like modality $\mathcal{G}_{\pi}$  that was previously defined as well as the Until-modality $\mathcal{U}_{\pi}$ that can be defined analogously.
	Also, we introduce two variations of $\bigcirc_{\pi}$ namely $\bigcirc_{\pi}^{symbol}$ and $\bigcirc_{\pi}^{word}$.
	In a sequence of Stockmeyer encoded words as inspected here, $\bigcirc_{\pi}^{symbol} \psi$ expects $\psi$ to hold at the next encoded symbol and $\bigcirc_{\pi}^{word} \psi$ expects $\psi$ to hold at the start of the next encoded word in the sequence.
	These modalities can straightforwardly be encoded in $\Hmu$.
	Finally, we use $\bigcirc_{\pi}^{word} \psi$ to encode modalities $\mathcal{G}_{\pi}^{word} \psi$ and $\mathcal{F}_{\pi}^{word} \psi$ which express that $\psi$ holds for each or one encoded word on $\pi$, respectively.
	
	To increase readability of the formula even further, we also introduce some auxiliary formulas:
	\begin{itemize}
		\item $\psi_{stock}^{k-1}(\pi,\pi')$ is obtained from \autoref{lem:contextboundedtransducer} and expresses that the next word of length $g(k-1,p(n))$ on $\pi$ represents a level $k-1$ Stockmeyer encoded word with the help of $\pi'$ and $k-1$ context switches
		\item $\psi_{same}^{k-2}(\pi,\pi')$ is obtained from \autoref{lem:contextboundedcomparison} and expresses that the first level $k-2$ Stockmeyer encoded words on $\pi$ and $\pi'$ match each other with $k-1$ context switches (starting with a $\pi$ context)
		\item $\psi_{start}(\pi)$ and $\psi_{end}(\pi)$ state that the first encoded word on $\pi$ is accepted by the first and second regular language acceptor respectively
		\item $\psi_{enc}(\pi)$ checks that $\pi$ respects a specific encoding. This means:
		\begin{itemize}
			\item On every index, exactly one atomic proposition from $\{0,1\} \cup \{[_i , ]_i | i \leq k\}$ is true
			\item On every index of an encoded word, exactly one atomic proposition from $Q$ is true
			\item The first $Q$-poposition in every encoded word is $q_0$
		\end{itemize}
	\end{itemize}
	
	Now we are able to construct the formula $\varphi$:
	\begin{align*}
		&\exists \pi_{seq} . \exists \pi_{help} . \forall \pi_{index} .\quad \psi_{stock}^{k-2}(\pi_{index},\pi_{help}) \rightarrow &\quad(1) \\
		&\qquad \psi_{enc}(\pi_{seq}) \land \mathcal{G}_{\pi_{seq}}^{word}\psi_{stock}^{k-1}(\pi_{seq},\pi_{help}) \;\land &\quad(2)\\
		&\qquad \psi_{start}(\pi_{seq}) \land \mathcal{F}_{\pi_{seq}}^{word} \psi_{end}(\pi_{seq}) \;\land &\quad(3) \\
		&\qquad \mathcal{G}_{\pi_{seq}} \psi_{same}^{k-2}(\pi_{seq},\pi_{index}) \rightarrow \\ 
		&\qquad \qquad \bigcirc_{\pi_{seq}}^{symbol} \bigvee_{i \in \{0,1\},q \in Q} i_{\pi_{seq}} \land q_{\pi_{seq}} \;\land \bigvee_{(q',i') \in \gamma(q,i)} \bigcirc_{\pi_{seq}}^{symbol} q'_{\pi_{seq}} \land &\quad(4) \\
		&\qquad \qquad \qquad \bigcirc_{\pi_{seq}}^{word} \lnot \psi_{same}^{k-2}(\pi_{seq},\pi_{index}) \mathcal{U}_{\pi_{seq}} \psi_{same}^{k-2}(\pi_{seq},\pi_{index}) \land \bigcirc_{\pi_{seq}}^{symbol} i'_{\pi_{seq}} &\quad
	\end{align*}
	Intuitively, this formula quantifying the three paths $\pi_{seq},\pi_{help}$ and $\pi_{index}$ can be understood in the following way:
	on $\pi_{seq}$, we check for the existence of the sequence of Stockmeyer encoded words.
	We use $\pi_{help}$ as a path yielding indices of lower level encodings to be used for checking the presence of Stockmeyer encodings as in \autoref{lem:contextboundedtransducer} or more specifically the translation of the AAPA referred to in that lemma into $\Hmu$.
	The path $\pi_{index}$ is universally quantified and is used to obtain all Stockmeyer indices for a level $k$ encoding.
	The condition in $(1)$ ensures that it is indeed a level $k-1$ Stockmeyer encoded word and each such word over $\{0,1\}$ represents an index for a level $k$ encoding.
	In $(2)$, we ensure that $\pi_{seq}$ has the required encoding.
	Part $(3)$ of the formula expresses the two conditions $(a)$ and $(b)$ that the sequence of words must fulfill.
	Finally, $(4)$ expresses condition $(c)$ of the sequence, i.e. that each word is obtained from the previous one by applying the regular transducer.
	This is done by checking the transductions for each index of the words separately (which is achieved by the universal quantification over $\pi_{index}$):
	whenever we find the index we are currently checking, we determine the current symbol of the word and the state the transducer is in.
	We choose one of $\mathcal{T}$'s possible transitions for the current symbol and state to determine its output and next state.
	The next state is then expected at the next symbol of the current encoded word.
	The output on the other hand is expected the next time this index occurs in the sequence.
\end{proof}

A few remarks are in order about these last results.

First, since the formula constructed for the reduction for \autoref{thm:contextboundedmodelchecking} works in a similar way as the automaton from \autoref{lem:contextboundedtransducer}, one could ask whether the hardness estimates can be inductively lifted to $\tilde{d} = d$
in a similar way as was done in \autoref{lem:contextboundedtransducer}.
This would require the formula $\psi_{same}$ used in our construction to handle longer nested index encodings.
As the number of context switches can not be increased further, this seems to require quantification.
Indeed, by adding suitable prenex alternating quantifiers, a corresponding formula could be defined.
However, this formula could no longer be used inside fixpoints since this would lead to formulas with non-prenex quantifiers and it is unclear whether such alternating quantifiers could be moved to the front of the formula without changing the semantics.
	
Secondly, while the inability to inductively lift the hardness result to an arbitrary number of quantifier alternations raises the question whether quantifier alternations increase the complexity of a context bounded analysis at all, our result shows that this is at least the case for the first one.
A complete answer to this question remains for future work.
	
Finally, the $k$-synchronous and $k$-context-bounded analyses can be combined by interpreting the quantifier free part of the given formula as a boolean combination of subformulas each of which is analysed with one of these analyses.
Our setup can easily handle this since the resulting automata for the analyses of these formulas are all synchronous automata which can be combined straightforwardly.
In this setting, the hardness proof for the $k$-context-bounded analysis of automata can be lifted with quantifier alternations in the same way as in the proof of \autoref{thm:ksynchronousmodelchecking} since the subformulas added in the inductive step of the proof are synchronous.
This yields a hardness result for the combined analysis that matches the upper bound for the $k$-context-bounded analysis. 

\section{Approximate Satisfiability Checking for Fragments of the Logic}\label{sec:satisfiability}

\begin{table}
	\centering
		\begin{tabular}{l l}
			\toprule
			fragment & complexity \\
			\midrule
			full & $\UNDECIDABLE$ \\
			alternation free synchronous & $\PSPACE$-complete \\
			alternation free $k$-synchronous & $\EXPSPACE{}$-complete \\
			alternation free $k$-context-bounded & $\EXPSPACE{(k-2)}$-complete \\
			$\exists^* \forall^*$ synchronous & $\EXPSPACE{}$-complete \\
			$\exists^* \forall^*$ $k$-synchronous & $\EXPSPACE{}$-complete \\
			$\exists^* \forall^*$ $k$-context-bounded & $\EXPSPACE{(k-2)}$-complete \\
			\bottomrule
		\end{tabular}
	\caption{Complexity results for satisfiability when representing a formula with $rep_d(\cdot)$.} 
	\label{table:satisfiability}
\end{table}

Since certain combinations of path quantifiers lead to undecidability already for synchronous hyperlogics like HyperLTL \cite{Finkbeiner2016}, we consider satisfiability for formulas with restricted quantifier prefixes only.
We say that a formula is in the $\exists^*$ fragment of $\Hmu$ if it only has $\exists$ quantifiers in its quantifier prefix.
The $\forall^*$ fragment is definied analogously.
Furthermore, we say that a formula is in the $\exists^* \forall^*$ fragment if its quantifier prefix has the form $\exists \pi_1 ... \exists \pi_n \forall \pi_1' ... \forall \pi_m'$.
Note that the same restrictions have been considered in \cite{Finkbeiner2016} for HyperLTL in order to obtain decidable fragments.
However, as our proof of \autoref{thm:satisfiability} shows, a quantifier prefix with just existential quantifiers suffices for undecidability in the case of $\Hmu$, if the quantifier-free part of the formula is not restricted or approximated as well.
Thus, just like for model checking, we make use of the approximate analyses developed in \autoref{sec:automata} for the quantifier-free part of the formula.
We now present complexity results for the decision problems corresponding to the resulting approximate analyses.
The upper bounds apply to these analyses as well.
An overview can be found in \autoref{table:satisfiability}.

\begin{theorem}\label{thm:alternationfreesatisfiability}
	The satisfiability problem for the $\exists^*$- and $\forall^*$-fragment of
	\begin{enumerate}
		\item synchronous $\Hmu$ is $\PSPACE$-complete
		\item $k$-synchronous $\Hmu$ is $\EXPSPACE{}$-complete
		\item $k$-context-bounded $\Hmu$ is $\EXPSPACE{(k-2)}$-complete
	\end{enumerate}
	when using $rep_d(\cdot)$ to represent formulas.
	The statement in (1) also holds for $rep_t(\cdot)$.
\end{theorem}

\begin{proof}
	Using \autoref{thm:aapahmuequivalence}, the three problems become interreducible to the emptiness problems for the corresponding AAPA restrictions.
	This is done in the following way:
	
	(i) An $\exists^*$ formula $\exists \pi_1 ... \exists \pi_n . \psi$ is satisfiable iff $\mathcal{A}_{\psi}$ is non-empty.
	(ii) A $\forall^*$ formula $\forall \pi_1 ... \forall \pi_n . \psi$ is satisfiable iff $\mathcal{A}_{\lnot \psi}$ is empty.
	(iii) An AAPA $\mathcal{A}$ is non-empty iff the formula $\exists \pi_1 ... \exists \pi_n . \psi_{\mathcal{A}}$ is satisfiable.
	(iv) An AAPA $\mathcal{A}$ is empty iff the formula $\forall \pi_1 ... \forall \pi_n . \lnot \psi_{\mathcal{A}}$ is satisfiable.
	
	This yields both upper and lower bounds for each of the problems using $rep_d(\cdot)$ as representation.
	The hardness result for item (1) and $rep_t(\cdot)$ can be obtained by reducing from the satisfiability problem for alternation-free HyperLTL instead.
\end{proof}

\begin{theorem}
	The satisfiability problem for the $\exists^*\forall^*$-fragment of
	\begin{enumerate}
		\item synchronous $\Hmu$ is $\EXPSPACE{}$-complete
		\item $k$-synchronous $\Hmu$ is $\EXPSPACE{}$-complete
		\item $k$-context-bounded $\Hmu$ is $\EXPSPACE{(k-2)}$-complete
	\end{enumerate}
	when using $rep_d(\cdot)$ to represent formulas.
\end{theorem}

\begin{proof}
	For the upper bounds, we adapt the idea from \cite{Finkbeiner2016} to test satisfiability of HyperLTL formulas only with a minimal set of traces: the set of traces chosen for the existential quantifiers.
	In such a set, the universal quantifiers can be instantiated in every possible combination and thus eliminated.
	More specifically, a HyperLTL formula $\exists \pi_1 ... \exists \pi_n \forall \pi_1' ... \forall \pi_m' . \psi$ is transformed into the equisatisfiable formula $\exists \pi_1 ... \exists \pi_n . \bigwedge_{j_1=1}^n ... \bigwedge_{j_m=1}^n \psi[\pi_{j_1}/\pi_1']...[\pi_{j_m}/\pi_m']$, which can then be tested for satisfiability using the same method as for $\exists^*$ formulas.
	
	However, in our setting, a direct instantiation of the universal quantifiers via substitution is possible only for synchronous $\Hmu$ formulas since tests for atomic propositions can occur with different offsets otherwise.
	For $k$-synchronous and $k$-context-bounded formulas, we have to incorporate the conjunctive test for each of these \textit{arrangements}, i.e. each set of substitutions, directly into our analysis of the corresponding automata.
	
	In the case of $k$-synchronous formulas $\exists \pi_1 ... \exists \pi_n \forall \pi_1' ... \forall \pi_m'. \psi$, we first create $n^m$ copies of $\mathcal{A}_{\psi}$ on $n+m$ input words, one for each arrangement.
	Then, for each arrangement $a$, we transform one copy of the AAPA into an APA $\mathcal{A}_{\psi}(a)$ with $n$ input directions and size $\mathcal{O}(|\psi|_d^2 \cdot |\Sigma|^{n \cdot k})$.
	This is done using a variation of the procedure from \autoref{thm:aapaksynchronoustosynchronous} which eliminates $m$ input directions and substitutes the corresponding checks according to the arrangement $a$.
	When we substitute a path $\pi_i$ with a path $\pi_j$, we have to make sure that all moves that were previously performed on direction $j$ are now performed on direction $i$ without manipulating the moves that were previously made on $i$.
	Thus, we introduce a second input marker for $\pi_j$ on direction $i$ in the $k \cdot n$ window from the proof of \autoref{thm:aapaksynchronoustosynchronous}.
	It is advanced whenever a symbol from direction $j$ is read and does not affect the other markers on the same direction.
	These additional $m$ markers do not asymptotically increase the size of the construction, thus we obtain the previously mentioned size.
	To perform the satisfiability test for the original formula, we now perform an emptiness test on the APA that makes a conjunctive move to $\mathcal{A}_{\psi}(a)$ for all arrangements $a$.
	Since there are $n^m$ different arrangements, this APA has size $\mathcal{O}(n^m \cdot |\psi|_d^2 \cdot |\Sigma|^{k \cdot n})$ and the test can be performed in $\EXPSPACE{}$.
	
	For a $k$-context-bounded analysis, we also have to incorporate the conjunctive test for all arrangements into the analysis.
	We first construct the structure $\mathcal{S}(g)$ for each guess $g$ without identifying any directions. 
	We combine this structure with an arrangement $a$ by replacing the test from \autoref{lem:contextboundedtest} with  $\bigwedge_{d'/d \in a} \bigwedge \{g' \in F_{d'}(\{q_0,g\}) | q_0 \in Q_0\}$ for all $d$ and some $Q_0$ with $Q_0 \models \rho_0$.
	Here, we use $d'/d \in a$ to denote that in the arrangement $a$, $d'$ is substituted by $d$.
	When using this notation, we assume that every arrangement contains the substitution $d/d$ such that the original paths are still considered.
	This integration of $a$ into $\mathcal{S}(g)$ to obtain $\mathcal{S}(a,g)$ does not increase its size beyond $\mathcal{O}(g(k-2,|\psi|_d))$ asymptotically.
	When combining the tests for all arrangements $a_1, \dots, a_{n^m}$, we have to keep in mind that each test can use a different guess.
	Thus, we construct APA $\mathcal{S}(g_1,...,g_{n^m})$ that are parameterised in guesses $g_1,...,g_{n^m}$ and conjunctively move into $\mathcal{S}(a_i,g_i)$. 
	Since they consist of $n^m$ APA of size $\mathcal{O}(g(k-2,|\psi|_d))$ and an initial state, they asymptotically have size $\mathcal{O}(g(k-2,|\psi|_d))$ as well.
	We remove alternation from these parameterised APA to obtain NPA $\mathcal{S}'(g_1,...,g_{n^m})$ of size $\mathcal{O}(g(k-1,|\psi|_d))$.
	The final NPA, which we test for emptiness to solve the satisfiability problem, nondeterministically guesses $g_1$ to $g_{n^m}$ and moves to $\mathcal{S}'(g_1,...,g_{n^m})$.
	Since there are $|G| = \mathcal{O}(g(k-1,|\psi|_d))$ possible guesses, we have $|G|^{n^m} = \mathcal{O}(g(k-1,|\psi|_d))$ possible combinations of guesses. 
	Thus, the final NPA also has an asymptotical size of $\mathcal{O}(g(k-1,|\psi|_d))$.
	This yields an emptiness test in $\EXPSPACE{(k-2)}$.
	
	For the lower bound, we need two different reductions.
	A reduction from $\exists^* \forall^*$ HyperLTL satisfiability yields the first and second lower bound.
	The third one is obtained from the fact that a $k$-context-bounded $\exists^*$ formula is especially an $\exists^* \forall^*$ formula.
\end{proof}
\section{Conclusion}\label{sec:conclusion}

In this paper, we introduced Alternating Asynchronous Parity Automata (AAPA) and the novel fixpoint logic $\Hmu$ as tools for the analysis of asynchronous hyperproperties.
We showed the most interesting decision problems for both models to be highly undecidable in general, but exhibited families of increasingly precise under- and overapproximations for both AAPA and $\Hmu$ and presented asymptotically optimal algorithms for most corresponding decision problems.
We also identified syntactic fragments where these analyses yield precise results.

Several questions remain for future work.
Firstly, while we have established an equivalence between AAPA and $\Hmu$ formulas over fixed path assignments, an interesting question is whether there is a natural model of tree automata possibly extending AAPA and equivalent to the full logic with quantifiers, analogous to the correspondence between the modal $\mu$-calculus and Alternating Parity Tree Automata \cite{Emerson1991}.
This could possibly lead to a more direct automata-theoretic approach to $\Hmu$ model checking.
Secondly, it would be interesting to identify further approximate analyses and corresponding decidable fragments.

\begin{acks}
	This work was partially funded by DFG project Model-Checking of Navigation Logics (MoNaLog) (MU 1508/3).
	We thank the reviewers for their helpful comments and Roland Meyer and Sören van der Wall for valuable discussions.
	We also thank Laura Bozzelli for providing us with an extended version of \cite{Bozzelli2007}.
\end{acks}
\bibliography{sections/conclusion/citations}


\begin{thebibliography}{49}


\ifx \showCODEN    \undefined \def \showCODEN     #1{\unskip}     \fi
\ifx \showDOI      \undefined \def \showDOI       #1{#1}\fi
\ifx \showISBNx    \undefined \def \showISBNx     #1{\unskip}     \fi
\ifx \showISBNxiii \undefined \def \showISBNxiii  #1{\unskip}     \fi
\ifx \showISSN     \undefined \def \showISSN      #1{\unskip}     \fi
\ifx \showLCCN     \undefined \def \showLCCN      #1{\unskip}     \fi
\ifx \shownote     \undefined \def \shownote      #1{#1}          \fi
\ifx \showarticletitle \undefined \def \showarticletitle #1{#1}   \fi
\ifx \showURL      \undefined \def \showURL       {\relax}        \fi
\providecommand\bibfield[2]{#2}
\providecommand\bibinfo[2]{#2}
\providecommand\natexlab[1]{#1}
\providecommand\showeprint[2][]{arXiv:#2}

\bibitem[\protect\citeauthoryear{Andersen}{Andersen}{1994}]%
        {Andersen1994}
\bibfield{author}{\bibinfo{person}{Henrik~Reif Andersen}.}
  \bibinfo{year}{1994}\natexlab{}.
\newblock \bibinfo{booktitle}{\emph{A polyadic modal $\mu$-calculus}}.
\newblock \bibinfo{type}{{T}echnical {R}eport} ID-TR: 1994-195.
  \bibinfo{institution}{Dept. of Computer Science, Technical University of
  Denmark, Copenhagen}.
\newblock


\bibitem[\protect\citeauthoryear{Atig, Bouajjani, and Qadeer}{Atig
  et~al\mbox{.}}{2009}]%
        {Atig2009}
\bibfield{author}{\bibinfo{person}{Mohamed~Faouzi Atig}, \bibinfo{person}{Ahmed
  Bouajjani}, {and} \bibinfo{person}{Shaz Qadeer}.}
  \bibinfo{year}{2009}\natexlab{}.
\newblock \showarticletitle{Context-Bounded Analysis for Concurrent Programs
  with Dynamic Creation of Threads}. In \bibinfo{booktitle}{\emph{Tools and
  Algorithms for the Construction and Analysis of Systems, 15th International
  Conference, {TACAS} 2009, Held as Part of the Joint European Conferences on
  Theory and Practice of Software, {ETAPS} 2009, York, UK, March 22-29, 2009.
  Proceedings}} \emph{(\bibinfo{series}{Lecture Notes in Computer Science})},
  \bibfield{editor}{\bibinfo{person}{Stefan Kowalewski} {and}
  \bibinfo{person}{Anna Philippou}} (Eds.), Vol.~\bibinfo{volume}{5505}.
  \bibinfo{publisher}{Springer}, \bibinfo{pages}{107--123}.
\newblock
\urldef\tempurl%
\url{https://doi.org/10.1007/978-3-642-00768-2\_11}
\showDOI{\tempurl}


\bibitem[\protect\citeauthoryear{Bansal and Demri}{Bansal and Demri}{2013}]%
        {Bansal2013}
\bibfield{author}{\bibinfo{person}{Kshitij Bansal} {and}
  \bibinfo{person}{St{\'{e}}phane Demri}.} \bibinfo{year}{2013}\natexlab{}.
\newblock \showarticletitle{Model-Checking Bounded Multi-Pushdown Systems}. In
  \bibinfo{booktitle}{\emph{Computer Science - Theory and Applications - 8th
  International Computer Science Symposium in Russia, {CSR} 2013, Ekaterinburg,
  Russia, June 25-29, 2013. Proceedings}}. \bibinfo{pages}{405--417}.
\newblock
\urldef\tempurl%
\url{https://doi.org/10.1007/978-3-642-38536-0\_35}
\showDOI{\tempurl}


\bibitem[\protect\citeauthoryear{Barringer, Kuiper, and Pnueli}{Barringer
  et~al\mbox{.}}{1986}]%
        {Barringer1986}
\bibfield{author}{\bibinfo{person}{Howard Barringer}, \bibinfo{person}{Ruurd
  Kuiper}, {and} \bibinfo{person}{Amir Pnueli}.}
  \bibinfo{year}{1986}\natexlab{}.
\newblock \showarticletitle{A Really Abstract Concurrent Model and its Temporal
  Logic}. In \bibinfo{booktitle}{\emph{Conference Record of the Thirteenth
  Annual {ACM} Symposium on Principles of Programming Languages, St. Petersburg
  Beach, Florida, USA, January 1986}}. \bibinfo{publisher}{{ACM} Press},
  \bibinfo{pages}{173--183}.
\newblock
\urldef\tempurl%
\url{https://doi.org/10.1145/512644.512660}
\showDOI{\tempurl}


\bibitem[\protect\citeauthoryear{Bozzelli}{Bozzelli}{2007}]%
        {Bozzelli2007}
\bibfield{author}{\bibinfo{person}{Laura Bozzelli}.}
  \bibinfo{year}{2007}\natexlab{}.
\newblock \showarticletitle{Alternating Automata and a Temporal Fixpoint
  Calculus for Visibly Pushdown Languages}. In \bibinfo{booktitle}{\emph{CONCUR
  2007 -- Concurrency Theory}}, \bibfield{editor}{\bibinfo{person}{Lu{\'i}s
  Caires} {and} \bibinfo{person}{Vasco~T. Vasconcelos}} (Eds.).
  \bibinfo{publisher}{Springer Berlin Heidelberg}, \bibinfo{address}{Berlin,
  Heidelberg}, \bibinfo{pages}{476--491}.
\newblock
\showISBNx{978-3-540-74407-8}


\bibitem[\protect\citeauthoryear{Bozzelli, Maubert, and Pinchinat}{Bozzelli
  et~al\mbox{.}}{2015}]%
        {Bozzelli2015}
\bibfield{author}{\bibinfo{person}{Laura Bozzelli}, \bibinfo{person}{Bastien
  Maubert}, {and} \bibinfo{person}{Sophie Pinchinat}.}
  \bibinfo{year}{2015}\natexlab{}.
\newblock \showarticletitle{Unifying hyper and epistemic temporal logics}. In
  \bibinfo{booktitle}{\emph{International Conference on Foundations of Software
  Science and Computation Structures}}. Springer, \bibinfo{pages}{167--182}.
\newblock


\bibitem[\protect\citeauthoryear{Clarkson, Finkbeiner, Koleini, Micinski, Rabe,
  and S{\'a}nchez}{Clarkson et~al\mbox{.}}{2014}]%
        {Clarkson2014}
\bibfield{author}{\bibinfo{person}{Michael~R. Clarkson}, \bibinfo{person}{Bernd
  Finkbeiner}, \bibinfo{person}{Masoud Koleini}, \bibinfo{person}{Kristopher~K.
  Micinski}, \bibinfo{person}{Markus~N. Rabe}, {and} \bibinfo{person}{C{\'e}sar
  S{\'a}nchez}.} \bibinfo{year}{2014}\natexlab{}.
\newblock \showarticletitle{Temporal Logics for Hyperproperties}. In
  \bibinfo{booktitle}{\emph{Principles of Security and Trust}},
  \bibfield{editor}{\bibinfo{person}{Mart{\'i}n Abadi} {and}
  \bibinfo{person}{Steve Kremer}} (Eds.). \bibinfo{publisher}{Springer Berlin
  Heidelberg}, \bibinfo{address}{Berlin, Heidelberg},
  \bibinfo{pages}{265--284}.
\newblock
\showISBNx{978-3-642-54792-8}


\bibitem[\protect\citeauthoryear{Clarkson and Schneider}{Clarkson and
  Schneider}{2010}]%
        {Clarkson2010}
\bibfield{author}{\bibinfo{person}{Michael~R. Clarkson} {and}
  \bibinfo{person}{Fred~B. Schneider}.} \bibinfo{year}{2010}\natexlab{}.
\newblock \showarticletitle{Hyperproperties}.
\newblock \bibinfo{journal}{\emph{J. Comput. Secur.}} \bibinfo{volume}{18},
  \bibinfo{number}{6} (\bibinfo{date}{Sept.} \bibinfo{year}{2010}),
  \bibinfo{pages}{1157--1210}.
\newblock
\showISSN{0926-227X}
\urldef\tempurl%
\url{http://dl.acm.org/citation.cfm?id=1891823.1891830}
\showURL{%
\tempurl}


\bibitem[\protect\citeauthoryear{Coenen, Finkbeiner, Hahn, and Hofmann}{Coenen
  et~al\mbox{.}}{2019}]%
        {Coenen2019}
\bibfield{author}{\bibinfo{person}{Norine Coenen}, \bibinfo{person}{Bernd
  Finkbeiner}, \bibinfo{person}{Christopher Hahn}, {and} \bibinfo{person}{Jana
  Hofmann}.} \bibinfo{year}{2019}\natexlab{}.
\newblock \showarticletitle{The Hierarchy of Hyperlogics}. In
  \bibinfo{booktitle}{\emph{34th Annual {ACM/IEEE} Symposium on Logic in
  Computer Science, {LICS} 2019, Vancouver, BC, Canada, June 24-27, 2019}}.
  \bibinfo{pages}{1--13}.
\newblock
\urldef\tempurl%
\url{https://doi.org/10.1109/LICS.2019.8785713}
\showDOI{\tempurl}


\bibitem[\protect\citeauthoryear{Cousot and Cousot}{Cousot and Cousot}{1979}]%
        {Cousot1979}
\bibfield{author}{\bibinfo{person}{Patrick Cousot} {and}
  \bibinfo{person}{Radhia Cousot}.} \bibinfo{year}{1979}\natexlab{}.
\newblock \showarticletitle{Constructive versions of Tarski's fixed point
  theorems.}
\newblock \bibinfo{journal}{\emph{Pacific J. Math.}} \bibinfo{volume}{82},
  \bibinfo{number}{1} (\bibinfo{year}{1979}), \bibinfo{pages}{43--57}.
\newblock
\urldef\tempurl%
\url{https://projecteuclid.org:443/euclid.pjm/1102785059}
\showURL{%
\tempurl}


\bibitem[\protect\citeauthoryear{Dax and Klaedtke}{Dax and Klaedtke}{2008}]%
        {Dax2008}
\bibfield{author}{\bibinfo{person}{Christian Dax} {and} \bibinfo{person}{Felix
  Klaedtke}.} \bibinfo{year}{2008}\natexlab{}.
\newblock \showarticletitle{Alternation elimination by complementation}. In
  \bibinfo{booktitle}{\emph{International Conference on Logic for Programming
  Artificial Intelligence and Reasoning}}. Springer, \bibinfo{pages}{214--229}.
\newblock


\bibitem[\protect\citeauthoryear{Demri, Goranko, and Lange}{Demri
  et~al\mbox{.}}{2016}]%
        {Demri2016}
\bibfield{author}{\bibinfo{person}{St{\'{e}}phane Demri},
  \bibinfo{person}{Valentin Goranko}, {and} \bibinfo{person}{Martin Lange}.}
  \bibinfo{year}{2016}\natexlab{}.
\newblock \bibinfo{booktitle}{\emph{Temporal Logics in Computer Science:
  Finite-State Systems}}.
\newblock \bibinfo{publisher}{Cambridge University Press}.
\newblock
\showISBNx{9781107028364}
\urldef\tempurl%
\url{https://doi.org/10.1017/CBO9781139236119}
\showDOI{\tempurl}


\bibitem[\protect\citeauthoryear{Durand{-}Gasselin, Esparza, Ganty, and
  Majumdar}{Durand{-}Gasselin et~al\mbox{.}}{2015}]%
        {DurandGasselin2015}
\bibfield{author}{\bibinfo{person}{Antoine Durand{-}Gasselin},
  \bibinfo{person}{Javier Esparza}, \bibinfo{person}{Pierre Ganty}, {and}
  \bibinfo{person}{Rupak Majumdar}.} \bibinfo{year}{2015}\natexlab{}.
\newblock \showarticletitle{Model Checking Parameterized Asynchronous
  Shared-Memory Systems}. In \bibinfo{booktitle}{\emph{Computer Aided
  Verification - 27th International Conference, {CAV} 2015, San Francisco, CA,
  USA, July 18-24, 2015, Proceedings, Part {I}}}
  \emph{(\bibinfo{series}{Lecture Notes in Computer Science})},
  \bibfield{editor}{\bibinfo{person}{Daniel Kroening} {and}
  \bibinfo{person}{Corina~S. Pasareanu}} (Eds.), Vol.~\bibinfo{volume}{9206}.
  \bibinfo{publisher}{Springer}, \bibinfo{pages}{67--84}.
\newblock
\urldef\tempurl%
\url{https://doi.org/10.1007/978-3-319-21690-4\_5}
\showDOI{\tempurl}


\bibitem[\protect\citeauthoryear{Emerson and Jutla}{Emerson and Jutla}{1991}]%
        {Emerson1991}
\bibfield{author}{\bibinfo{person}{E.~Allen Emerson} {and}
  \bibinfo{person}{Charanjit~S. Jutla}.} \bibinfo{year}{1991}\natexlab{}.
\newblock \showarticletitle{Tree Automata, Mu-Calculus and Determinacy
  (Extended Abstract)}. In \bibinfo{booktitle}{\emph{32nd Annual Symposium on
  Foundations of Computer Science, San Juan, Puerto Rico, 1-4 October 1991}}.
  \bibinfo{pages}{368--377}.
\newblock
\urldef\tempurl%
\url{https://doi.org/10.1109/SFCS.1991.185392}
\showDOI{\tempurl}


\bibitem[\protect\citeauthoryear{Esparza, Ganty, and Majumdar}{Esparza
  et~al\mbox{.}}{2016}]%
        {Esparza2016}
\bibfield{author}{\bibinfo{person}{Javier Esparza}, \bibinfo{person}{Pierre
  Ganty}, {and} \bibinfo{person}{Rupak Majumdar}.}
  \bibinfo{year}{2016}\natexlab{}.
\newblock \showarticletitle{Parameterized Verification of Asynchronous
  Shared-Memory Systems}.
\newblock \bibinfo{journal}{\emph{J. {ACM}}} \bibinfo{volume}{63},
  \bibinfo{number}{1} (\bibinfo{year}{2016}), \bibinfo{pages}{10:1--10:48}.
\newblock
\urldef\tempurl%
\url{https://doi.org/10.1145/2842603}
\showDOI{\tempurl}


\bibitem[\protect\citeauthoryear{Finkbeiner}{Finkbeiner}{2017}]%
        {Finkbeiner2017}
\bibfield{author}{\bibinfo{person}{Bernd Finkbeiner}.}
  \bibinfo{year}{2017}\natexlab{}.
\newblock \showarticletitle{Temporal Hyperproperties}.
\newblock \bibinfo{journal}{\emph{Bulletin of the {EATCS}}}
  \bibinfo{volume}{123} (\bibinfo{year}{2017}).
\newblock


\bibitem[\protect\citeauthoryear{Finkbeiner and Hahn}{Finkbeiner and
  Hahn}{2016}]%
        {Finkbeiner2016}
\bibfield{author}{\bibinfo{person}{Bernd Finkbeiner} {and}
  \bibinfo{person}{Christopher Hahn}.} \bibinfo{year}{2016}\natexlab{}.
\newblock \showarticletitle{Deciding Hyperproperties}. In
  \bibinfo{booktitle}{\emph{{CONCUR} 2016}}. \bibinfo{pages}{13:1--13:14}.
\newblock
\urldef\tempurl%
\url{https://doi.org/10.4230/LIPIcs.CONCUR.2016.13}
\showDOI{\tempurl}


\bibitem[\protect\citeauthoryear{Finkbeiner, Hahn, Lukert, Stenger, and
  Tentrup}{Finkbeiner et~al\mbox{.}}{2020}]%
        {Finkbeiner2020}
\bibfield{author}{\bibinfo{person}{Bernd Finkbeiner},
  \bibinfo{person}{Christopher Hahn}, \bibinfo{person}{Philip Lukert},
  \bibinfo{person}{Marvin Stenger}, {and} \bibinfo{person}{Leander Tentrup}.}
  \bibinfo{year}{2020}\natexlab{}.
\newblock \showarticletitle{Synthesis from hyperproperties}.
\newblock \bibinfo{journal}{\emph{Acta Informatica}} \bibinfo{volume}{57},
  \bibinfo{number}{1-2} (\bibinfo{year}{2020}), \bibinfo{pages}{137--163}.
\newblock
\urldef\tempurl%
\url{https://doi.org/10.1007/s00236-019-00358-2}
\showDOI{\tempurl}


\bibitem[\protect\citeauthoryear{Finkbeiner, Hahn, Stenger, and
  Tentrup}{Finkbeiner et~al\mbox{.}}{2019}]%
        {Finkbeiner2019}
\bibfield{author}{\bibinfo{person}{Bernd Finkbeiner},
  \bibinfo{person}{Christopher Hahn}, \bibinfo{person}{Marvin Stenger}, {and}
  \bibinfo{person}{Leander Tentrup}.} \bibinfo{year}{2019}\natexlab{}.
\newblock \showarticletitle{Monitoring hyperproperties}.
\newblock \bibinfo{journal}{\emph{Formal Methods Syst. Des.}}
  \bibinfo{volume}{54}, \bibinfo{number}{3} (\bibinfo{year}{2019}),
  \bibinfo{pages}{336--363}.
\newblock
\urldef\tempurl%
\url{https://doi.org/10.1007/s10703-019-00334-z}
\showDOI{\tempurl}


\bibitem[\protect\citeauthoryear{Finkbeiner, Rabe, and
  S{\'{a}}nchez}{Finkbeiner et~al\mbox{.}}{2015}]%
        {Finkbeiner2015}
\bibfield{author}{\bibinfo{person}{Bernd Finkbeiner},
  \bibinfo{person}{Markus~N. Rabe}, {and} \bibinfo{person}{C{\'{e}}sar
  S{\'{a}}nchez}.} \bibinfo{year}{2015}\natexlab{}.
\newblock \showarticletitle{Algorithms for Model Checking HyperLTL and
  HyperCTL$^*$}. In \bibinfo{booktitle}{\emph{{CAV} 2015}}.
  \bibinfo{pages}{30--48}.
\newblock
\urldef\tempurl%
\url{https://doi.org/10.1007/978-3-319-21690-4\_3}
\showDOI{\tempurl}


\bibitem[\protect\citeauthoryear{Finkel}{Finkel}{2006}]%
        {Finkel2006}
\bibfield{author}{\bibinfo{person}{Olivier Finkel}.}
  \bibinfo{year}{2006}\natexlab{}.
\newblock \showarticletitle{On the Accepting Power of 2-Tape B{\"{u}}chi
  Automata}. In \bibinfo{booktitle}{\emph{{STACS} 2006, 23rd Annual Symposium
  on Theoretical Aspects of Computer Science, Marseille, France, February
  23-25, 2006, Proceedings}}. \bibinfo{pages}{301--312}.
\newblock
\urldef\tempurl%
\url{https://doi.org/10.1007/11672142\_24}
\showDOI{\tempurl}


\bibitem[\protect\citeauthoryear{Finkel}{Finkel}{2016}]%
        {Finkel2016}
\bibfield{author}{\bibinfo{person}{Olivier Finkel}.}
  \bibinfo{year}{2016}\natexlab{}.
\newblock \showarticletitle{Infinite games specified by 2-tape automata}.
\newblock \bibinfo{journal}{\emph{Ann. Pure Appl. Logic}}
  \bibinfo{volume}{167}, \bibinfo{number}{12} (\bibinfo{year}{2016}),
  \bibinfo{pages}{1184--1212}.
\newblock
\urldef\tempurl%
\url{https://doi.org/10.1016/j.apal.2016.05.005}
\showDOI{\tempurl}


\bibitem[\protect\citeauthoryear{Finkel and Lecomte}{Finkel and
  Lecomte}{2009}]%
        {Finkel2009}
\bibfield{author}{\bibinfo{person}{Olivier Finkel} {and}
  \bibinfo{person}{Dominique Lecomte}.} \bibinfo{year}{2009}\natexlab{}.
\newblock \showarticletitle{Decision problems for Turing machines}.
\newblock \bibinfo{journal}{\emph{Inf. Process. Lett.}} \bibinfo{volume}{109},
  \bibinfo{number}{23-24} (\bibinfo{year}{2009}), \bibinfo{pages}{1223--1226}.
\newblock
\urldef\tempurl%
\url{https://doi.org/10.1016/j.ipl.2009.09.002}
\showDOI{\tempurl}


\bibitem[\protect\citeauthoryear{Furia}{Furia}{2014}]%
        {Furia2014}
\bibfield{author}{\bibinfo{person}{Carlo~A. Furia}.}
  \bibinfo{year}{2014}\natexlab{}.
\newblock \showarticletitle{Asynchronous Multi-Tape Automata Intersection:
  Undecidabiliy and Approximation}.
\newblock \bibinfo{journal}{\emph{CoRR}}  \bibinfo{volume}{abs/1206.4860}
  (\bibinfo{year}{2014}).
\newblock
\showeprint[arxiv]{1206.4860v5}
\urldef\tempurl%
\url{http://arxiv.org/abs/1206.4860v5}
\showURL{%
\tempurl}


\bibitem[\protect\citeauthoryear{Ganty and Majumdar}{Ganty and
  Majumdar}{2012}]%
        {Ganty2012}
\bibfield{author}{\bibinfo{person}{Pierre Ganty} {and} \bibinfo{person}{Rupak
  Majumdar}.} \bibinfo{year}{2012}\natexlab{}.
\newblock \showarticletitle{Algorithmic verification of asynchronous programs}.
\newblock \bibinfo{journal}{\emph{{ACM} Trans. Program. Lang. Syst.}}
  \bibinfo{volume}{34}, \bibinfo{number}{1} (\bibinfo{year}{2012}),
  \bibinfo{pages}{6:1--6:48}.
\newblock
\urldef\tempurl%
\url{https://doi.org/10.1145/2160910.2160915}
\showDOI{\tempurl}


\bibitem[\protect\citeauthoryear{Ganty, Majumdar, and Rybalchenko}{Ganty
  et~al\mbox{.}}{2009}]%
        {Ganty2009}
\bibfield{author}{\bibinfo{person}{Pierre Ganty}, \bibinfo{person}{Rupak
  Majumdar}, {and} \bibinfo{person}{Andrey Rybalchenko}.}
  \bibinfo{year}{2009}\natexlab{}.
\newblock \showarticletitle{Verifying liveness for asynchronous programs}. In
  \bibinfo{booktitle}{\emph{Proceedings of the 36th {ACM} {SIGPLAN-SIGACT}
  Symposium on Principles of Programming Languages, {POPL} 2009, Savannah, GA,
  USA, January 21-23, 2009}}, \bibfield{editor}{\bibinfo{person}{Zhong Shao}
  {and} \bibinfo{person}{Benjamin~C. Pierce}} (Eds.).
  \bibinfo{publisher}{{ACM}}, \bibinfo{pages}{102--113}.
\newblock
\urldef\tempurl%
\url{https://doi.org/10.1145/1480881.1480895}
\showDOI{\tempurl}


\bibitem[\protect\citeauthoryear{Geidmanis}{Geidmanis}{1987}]%
        {Geidmanis1987}
\bibfield{author}{\bibinfo{person}{Dainis Geidmanis}.}
  \bibinfo{year}{1987}\natexlab{}.
\newblock \showarticletitle{On the Capabilities of Alternating and
  Nondeterministic Multitape Automata}. In
  \bibinfo{booktitle}{\emph{Fundamentals of Computation Theory, International
  Conference FCT'87, Kazan, USSR, June 22-26, 1987, Proceedings}}.
  \bibinfo{pages}{150--154}.
\newblock
\urldef\tempurl%
\url{https://doi.org/10.1007/3-540-18740-5\_35}
\showDOI{\tempurl}


\bibitem[\protect\citeauthoryear{Gutsfeld, M{\"{u}}ller{-}Olm, and
  Ohrem}{Gutsfeld et~al\mbox{.}}{2020}]%
        {Gutsfeld2020}
\bibfield{author}{\bibinfo{person}{Jens~Oliver Gutsfeld},
  \bibinfo{person}{Markus M{\"{u}}ller{-}Olm}, {and} \bibinfo{person}{Christoph
  Ohrem}.} \bibinfo{year}{2020}\natexlab{}.
\newblock \showarticletitle{Propositional Dynamic Logic for Hyperproperties}.
  In \bibinfo{booktitle}{\emph{31st International Conference on Concurrency
  Theory, {CONCUR} 2020, September 1-4, 2020, Vienna, Austria (Virtual
  Conference)}} \emph{(\bibinfo{series}{LIPIcs})},
  \bibfield{editor}{\bibinfo{person}{Igor Konnov} {and} \bibinfo{person}{Laura
  Kov{\'{a}}cs}} (Eds.), Vol.~\bibinfo{volume}{171}.
  \bibinfo{publisher}{Schloss Dagstuhl - Leibniz-Zentrum f{\"{u}}r Informatik},
  \bibinfo{pages}{50:1--50:22}.
\newblock
\urldef\tempurl%
\url{https://doi.org/10.4230/LIPIcs.CONCUR.2020.50}
\showDOI{\tempurl}


\bibitem[\protect\citeauthoryear{Ibarra and Tr{\^{a}}n}{Ibarra and
  Tr{\^{a}}n}{2013}]%
        {Ibarra2013}
\bibfield{author}{\bibinfo{person}{Oscar~H. Ibarra} {and}
  \bibinfo{person}{Nicholas~Q. Tr{\^{a}}n}.} \bibinfo{year}{2013}\natexlab{}.
\newblock \showarticletitle{How to synchronize the Heads of a Multitape
  Automaton}.
\newblock \bibinfo{journal}{\emph{Int. J. Found. Comput. Sci.}}
  \bibinfo{volume}{24}, \bibinfo{number}{6} (\bibinfo{year}{2013}),
  \bibinfo{pages}{799--814}.
\newblock
\urldef\tempurl%
\url{https://doi.org/10.1142/S0129054113400194}
\showDOI{\tempurl}


\bibitem[\protect\citeauthoryear{Jr.}{Jr.}{1987}]%
        {Rogers1967}
\bibfield{author}{\bibinfo{person}{Hartley~Rogers Jr.}}
  \bibinfo{year}{1987}\natexlab{}.
\newblock \bibinfo{booktitle}{\emph{Theory of recursive functions and effective
  computability (Reprint from 1967)}}.
\newblock \bibinfo{publisher}{{MIT} Press}.
\newblock
\showISBNx{978-0-262-68052-3}
\urldef\tempurl%
\url{http://mitpress.mit.edu/catalog/item/default.asp?ttype=2\&tid=3182}
\showURL{%
\tempurl}


\bibitem[\protect\citeauthoryear{Krebs, Meier, Virtema, and Zimmermann}{Krebs
  et~al\mbox{.}}{2017}]%
        {Krebs2017}
\bibfield{author}{\bibinfo{person}{Andreas Krebs}, \bibinfo{person}{Arne
  Meier}, \bibinfo{person}{Jonni Virtema}, {and} \bibinfo{person}{Martin
  Zimmermann}.} \bibinfo{year}{2017}\natexlab{}.
\newblock \showarticletitle{Team Semantics for the Specification and
  Verification of Hyperproperties}.
\newblock \bibinfo{journal}{\emph{CoRR}}  \bibinfo{volume}{abs/1709.08510}
  (\bibinfo{year}{2017}).
\newblock


\bibitem[\protect\citeauthoryear{Lange}{Lange}{2005}]%
        {Lange2005}
\bibfield{author}{\bibinfo{person}{Martin Lange}.}
  \bibinfo{year}{2005}\natexlab{}.
\newblock \showarticletitle{Weak Automata for the Linear Time $\mu$-Calculus}.
  In \bibinfo{booktitle}{\emph{Verification, Model Checking, and Abstract
  Interpretation}}, \bibfield{editor}{\bibinfo{person}{Radhia Cousot}} (Ed.).
  \bibinfo{publisher}{Springer Berlin Heidelberg}, \bibinfo{address}{Berlin,
  Heidelberg}, \bibinfo{pages}{267--281}.
\newblock
\showISBNx{978-3-540-30579-8}


\bibitem[\protect\citeauthoryear{Lange}{Lange}{2015}]%
        {Lange2015}
\bibfield{author}{\bibinfo{person}{Martin Lange}.}
  \bibinfo{year}{2015}\natexlab{}.
\newblock \showarticletitle{The Arity Hierarchy in the Polyadic
  {\(\mu\)}-Calculus}. In \bibinfo{booktitle}{\emph{{FICS}}}
  \emph{(\bibinfo{series}{{EPTCS}})}, Vol.~\bibinfo{volume}{191}.
  \bibinfo{pages}{105--116}.
\newblock


\bibitem[\protect\citeauthoryear{Mastroeni and Pasqua}{Mastroeni and
  Pasqua}{2017}]%
        {Mastroeni2017}
\bibfield{author}{\bibinfo{person}{Isabella Mastroeni} {and}
  \bibinfo{person}{Michele Pasqua}.} \bibinfo{year}{2017}\natexlab{}.
\newblock \showarticletitle{Hyperhierarchy of Semantics - {A} Formal Framework
  for Hyperproperties Verification}. In \bibinfo{booktitle}{\emph{Static
  Analysis - 24th International Symposium, {SAS} 2017, New York, NY, USA,
  August 30 - September 1, 2017, Proceedings}} \emph{(\bibinfo{series}{Lecture
  Notes in Computer Science})}, \bibfield{editor}{\bibinfo{person}{Francesco
  Ranzato}} (Ed.), Vol.~\bibinfo{volume}{10422}. \bibinfo{publisher}{Springer},
  \bibinfo{pages}{232--252}.
\newblock
\urldef\tempurl%
\url{https://doi.org/10.1007/978-3-319-66706-5\_12}
\showDOI{\tempurl}


\bibitem[\protect\citeauthoryear{Mastroeni and Pasqua}{Mastroeni and
  Pasqua}{2018}]%
        {Mastroeni2018}
\bibfield{author}{\bibinfo{person}{Isabella Mastroeni} {and}
  \bibinfo{person}{Michele Pasqua}.} \bibinfo{year}{2018}\natexlab{}.
\newblock \showarticletitle{Verifying Bounded Subset-Closed Hyperproperties}.
  In \bibinfo{booktitle}{\emph{Static Analysis - 25th International Symposium,
  {SAS} 2018, Freiburg, Germany, August 29-31, 2018, Proceedings}}
  \emph{(\bibinfo{series}{Lecture Notes in Computer Science})},
  \bibfield{editor}{\bibinfo{person}{Andreas Podelski}} (Ed.),
  Vol.~\bibinfo{volume}{11002}. \bibinfo{publisher}{Springer},
  \bibinfo{pages}{263--283}.
\newblock
\urldef\tempurl%
\url{https://doi.org/10.1007/978-3-319-99725-4\_17}
\showDOI{\tempurl}


\bibitem[\protect\citeauthoryear{Milushev and Clarke}{Milushev and
  Clarke}{2013}]%
        {Milushev2013}
\bibfield{author}{\bibinfo{person}{Dimiter Milushev} {and}
  \bibinfo{person}{Dave Clarke}.} \bibinfo{year}{2013}\natexlab{}.
\newblock \showarticletitle{Incremental Hyperproperty Model Checking via
  Games}. In \bibinfo{booktitle}{\emph{Proceedings of the 18th Nordic
  Conference on Secure IT Systems - Volume 8208}} (Ilulissat, Greenland)
  \emph{(\bibinfo{series}{NordSec 2013})}. \bibinfo{publisher}{Springer-Verlag
  New York, Inc.}, \bibinfo{address}{New York, NY, USA},
  \bibinfo{pages}{247--262}.
\newblock
\showISBNx{978-3-642-41487-9}
\urldef\tempurl%
\url{https://doi.org/10.1007/978-3-642-41488-6_17}
\showDOI{\tempurl}


\bibitem[\protect\citeauthoryear{Muscholl}{Muscholl}{1996}]%
        {Muscholl1996}
\bibfield{author}{\bibinfo{person}{Anca Muscholl}.}
  \bibinfo{year}{1996}\natexlab{}.
\newblock \showarticletitle{On the Complementation of Asynchronous Cellular
  B{\"{u}}chi Automata}.
\newblock \bibinfo{journal}{\emph{Theor. Comput. Sci.}} \bibinfo{volume}{169},
  \bibinfo{number}{2} (\bibinfo{year}{1996}), \bibinfo{pages}{123--145}.
\newblock


\bibitem[\protect\citeauthoryear{Otto}{Otto}{1999}]%
        {Otto1999}
\bibfield{author}{\bibinfo{person}{Martin Otto}.}
  \bibinfo{year}{1999}\natexlab{}.
\newblock \showarticletitle{Bisimulation-invariant {PTIME} and
  higher-dimensional {\(\mathrm{\mu}\)}-calculus}.
\newblock \bibinfo{journal}{\emph{Theor. Comput. Sci.}} \bibinfo{volume}{224},
  \bibinfo{number}{1-2} (\bibinfo{year}{1999}), \bibinfo{pages}{237--265}.
\newblock


\bibitem[\protect\citeauthoryear{Peled and Penczek}{Peled and Penczek}{1996}]%
        {Peled1996}
\bibfield{author}{\bibinfo{person}{D. Peled} {and} \bibinfo{person}{W.
  Penczek}.} \bibinfo{year}{1996}\natexlab{}.
\newblock \showarticletitle{Using Asynchronous B{\"u}chi Automata for Efficient
  Automatic Verification of Concurrent Systems}. In
  \bibinfo{booktitle}{\emph{Protocol Specification, Testing and Verification
  XV: Proceedings of the Fifteenth IFIP WG6.1 International Symposium on
  Protocol Specification, Testing and Verification, Warsaw, Poland, June
  1995}}. \bibinfo{publisher}{Springer US}, \bibinfo{address}{Boston, MA},
  \bibinfo{pages}{315--330}.
\newblock


\bibitem[\protect\citeauthoryear{Qadeer}{Qadeer}{2008}]%
        {Qadeer2008}
\bibfield{author}{\bibinfo{person}{Shaz Qadeer}.}
  \bibinfo{year}{2008}\natexlab{}.
\newblock \showarticletitle{The Case for Context-Bounded Verification of
  Concurrent Programs}. In \bibinfo{booktitle}{\emph{Model Checking Software,
  15th International {SPIN} Workshop, Los Angeles, CA, USA, August 10-12, 2008,
  Proceedings}} \emph{(\bibinfo{series}{Lecture Notes in Computer Science})},
  \bibfield{editor}{\bibinfo{person}{Klaus Havelund}, \bibinfo{person}{Rupak
  Majumdar}, {and} \bibinfo{person}{Jens Palsberg}} (Eds.),
  Vol.~\bibinfo{volume}{5156}. \bibinfo{publisher}{Springer},
  \bibinfo{pages}{3--6}.
\newblock
\urldef\tempurl%
\url{https://doi.org/10.1007/978-3-540-85114-1\_2}
\showDOI{\tempurl}


\bibitem[\protect\citeauthoryear{Qadeer and Rehof}{Qadeer and Rehof}{2005}]%
        {Qadeer2005}
\bibfield{author}{\bibinfo{person}{Shaz Qadeer} {and} \bibinfo{person}{Jakob
  Rehof}.} \bibinfo{year}{2005}\natexlab{}.
\newblock \showarticletitle{Context-Bounded Model Checking of Concurrent
  Software}. In \bibinfo{booktitle}{\emph{Tools and Algorithms for the
  Construction and Analysis of Systems, 11th International Conference, {TACAS}
  2005, Held as Part of the Joint European Conferences on Theory and Practice
  of Software, {ETAPS} 2005, Edinburgh, UK, April 4-8, 2005, Proceedings}}
  \emph{(\bibinfo{series}{Lecture Notes in Computer Science})},
  \bibfield{editor}{\bibinfo{person}{Nicolas Halbwachs} {and}
  \bibinfo{person}{Lenore~D. Zuck}} (Eds.), Vol.~\bibinfo{volume}{3440}.
  \bibinfo{publisher}{Springer}, \bibinfo{pages}{93--107}.
\newblock
\urldef\tempurl%
\url{https://doi.org/10.1007/978-3-540-31980-1\_7}
\showDOI{\tempurl}


\bibitem[\protect\citeauthoryear{Rabe}{Rabe}{2016}]%
        {Rabe2016}
\bibfield{author}{\bibinfo{person}{Markus~N. Rabe}.}
  \bibinfo{year}{2016}\natexlab{}.
\newblock \emph{\bibinfo{title}{A temporal logic approach to Information-flow
  control}}.
\newblock \bibinfo{thesistype}{Ph.D. Dissertation}. \bibinfo{school}{Saarland
  University}.
\newblock


\bibitem[\protect\citeauthoryear{Rabin and Scott}{Rabin and Scott}{1959}]%
        {Rabin1959}
\bibfield{author}{\bibinfo{person}{Michael~O. Rabin} {and}
  \bibinfo{person}{Dana~S. Scott}.} \bibinfo{year}{1959}\natexlab{}.
\newblock \showarticletitle{Finite Automata and Their Decision Problems}.
\newblock \bibinfo{journal}{\emph{{IBM} Journal of Research and Development}}
  \bibinfo{volume}{3}, \bibinfo{number}{2} (\bibinfo{year}{1959}),
  \bibinfo{pages}{114--125}.
\newblock
\urldef\tempurl%
\url{https://doi.org/10.1147/rd.32.0114}
\showDOI{\tempurl}


\bibitem[\protect\citeauthoryear{Sistla}{Sistla}{1983}]%
        {Sistla1983}
\bibfield{author}{\bibinfo{person}{Aravinda~Prasad Sistla}.}
  \bibinfo{year}{1983}\natexlab{}.
\newblock \emph{\bibinfo{title}{Theoretical Issues in the Design and
  Verification of Distributed Systems}}.
\newblock \bibinfo{thesistype}{Ph.D. Dissertation}.
  \bibinfo{school}{Carnegie-Mellon University}, \bibinfo{address}{USA}.
\newblock


\bibitem[\protect\citeauthoryear{Spelten, Thomas, and Winter}{Spelten
  et~al\mbox{.}}{2011}]%
        {Thomas2011}
\bibfield{author}{\bibinfo{person}{Alex Spelten}, \bibinfo{person}{Wolfgang
  Thomas}, {and} \bibinfo{person}{Sarah Winter}.}
  \bibinfo{year}{2011}\natexlab{}.
\newblock \showarticletitle{Trees over Infinite Structures and Path Logics with
  Synchronization}. In \bibinfo{booktitle}{\emph{Proceedings 13th International
  Workshop on Verification of Infinite-State Systems, {INFINITY} 2011, Taipei,
  Taiwan, 10th October 2011}} \emph{(\bibinfo{series}{{EPTCS}})},
  \bibfield{editor}{\bibinfo{person}{Fang Yu} {and} \bibinfo{person}{Chao
  Wang}} (Eds.), Vol.~\bibinfo{volume}{73}. \bibinfo{pages}{20--34}.
\newblock
\urldef\tempurl%
\url{https://doi.org/10.4204/EPTCS.73.5}
\showDOI{\tempurl}


\bibitem[\protect\citeauthoryear{Stockmeyer}{Stockmeyer}{1974}]%
        {Stockmeyer1974}
\bibfield{author}{\bibinfo{person}{Larry~Joseph Stockmeyer}.}
  \bibinfo{year}{1974}\natexlab{}.
\newblock \emph{\bibinfo{title}{The complexity of decision problems in automata
  theory and logic}}.
\newblock \bibinfo{thesistype}{Ph.D. Dissertation}. \bibinfo{school}{MIT}.
\newblock


\bibitem[\protect\citeauthoryear{Tarski}{Tarski}{1955}]%
        {Tarski1955}
\bibfield{author}{\bibinfo{person}{Alfred Tarski}.}
  \bibinfo{year}{1955}\natexlab{}.
\newblock \showarticletitle{A lattice-theoretical fixpoint theorem and its
  applications.}
\newblock \bibinfo{journal}{\emph{Pacific J. Math.}} \bibinfo{volume}{5},
  \bibinfo{number}{2} (\bibinfo{year}{1955}), \bibinfo{pages}{285--309}.
\newblock
\urldef\tempurl%
\url{https://projecteuclid.org:443/euclid.pjm/1103044538}
\showURL{%
\tempurl}


\bibitem[\protect\citeauthoryear{Vardi}{Vardi}{1988}]%
        {Vardi1988}
\bibfield{author}{\bibinfo{person}{Moshe~Y. Vardi}.}
  \bibinfo{year}{1988}\natexlab{}.
\newblock \showarticletitle{A Temporal Fixpoint Calculus}. In
  \bibinfo{booktitle}{\emph{{POPL}}}. \bibinfo{publisher}{{ACM} Press},
  \bibinfo{pages}{250--259}.
\newblock


\bibitem[\protect\citeauthoryear{Zielonka}{Zielonka}{1987}]%
        {Zielonka1987}
\bibfield{author}{\bibinfo{person}{Wieslaw Zielonka}.}
  \bibinfo{year}{1987}\natexlab{}.
\newblock \showarticletitle{Notes on Finite Asynchronous Automata}.
\newblock \bibinfo{journal}{\emph{{ITA}}} \bibinfo{volume}{21},
  \bibinfo{number}{2} (\bibinfo{year}{1987}), \bibinfo{pages}{99--135}.
\newblock


\end{thebibliography}
\clearpage
\appendix

\section{Missing proofs from Section \ref{sec:automata}}






\subsection{Recursion Theory of AAPA}

We briefly outline some elementary notions of recursion theory and the theory of analytic sets. We refer the reader to \cite{Rogers1967} for a thorough introduction.

A $2$-tape Büchi automaton is a sextuple $\mathcal{T} = (K, \Sigma_1, \Sigma_2, \Delta, q_0, F)$ where $K$ is a finite set of states, $\Sigma_1, \Sigma_2$ are finite alphabets, $\Delta$ is a finite subset of $K \times \Sigma_1^* \times \Sigma_2^* \times K$, $q_0$ is the initial state and $F \subseteq K$ is the set of final states.
A computation $\mathcal{C}$ of $\mathcal{T}$ is an infinite sequence of transitions $(q_0, u_1, v_1, q_1)(q_1, u_2, v_2, q_2) \dots$.
A computation is accepting if a state $q \in F$ is visited infinitely often. 
The input word then is $u = u_1 u_2 \dots$ and the output word is $v = v_1 v_2 \dots$.
The infinitary rational relation $\mathcal{R}(\mathcal{T}) \subseteq \Sigma_1^{\omega} \times \Sigma_2^{\omega}$ accepted by $\mathcal{T}$ is the set of tuples $(u, v)$ for which there is an accepting computation of $\mathcal{T}$. A 2-tape Büchi automaton can be considered a NAPA in which transitions are allowed to depend on input words and emit output words instead of single symbols only.

Let $\Sigma_0^1 = \Pi_0^1$ be the set of formulas of second order arithmetic with no set quantifiers. A formula in the language of second order arithmetic is $\Sigma_{n+1}^1$ if it is logically equivalent to a formula of the form $\exists X_1 \dots \exists X_n \psi$ where $\psi$ is $\Pi_{n}^1$ and $\Pi_{n+1}^1$ if it is logically equivalent to a formula of the form $\forall X_1 \dots \forall X_n \psi$ where $\psi$ is $\Sigma_{n}^1$. As usual, capital notation for variables indicates that they are second order variables. A set of natural numbers is said to be $\Sigma_n^1$ (resp. $\Pi_n^1$) if there is a $\Sigma_n^1$ (resp. $\Pi_n^1$) formula defining it. 
Given two sets $A, B \subseteq \mathbb{N}$, we say that $A$ is $1$-reducible to $B$ (written $A \leq_1 B$) if there is a total (i), computable (ii) and injective (iii) function $f: \mathbb{N} \rightarrow \mathbb{N}$ such that $A = f^{-1}(B)$ (iv). 
A set $A \subseteq \mathbb{N}$ of natural numbers is called $\Sigma_n^1$-hard (resp. $\Pi_n^1$-hard) if every $\Sigma_n^1$ (resp.  $\Pi_n^1$) set $B$, $B \leq_1 A$ holds. $A$ is called $\Sigma_n^1$-complete (resp. $\Pi_n^1$-complete) if $A$ is $\Sigma_n^1$-hard (resp. $\Pi_n^1$-hard) and $A$ is a $\Sigma_n^1$ (resp. $\Pi_n^1$) set. 
If $A$ is $\Sigma_n^1$-hard, then $\overline{A}$ is $\Pi_n^1$-hard and vice versa.
We will make use of the following fact:

\begin{proposition}[\cite{Finkel2009}]
For two-tape Büchi automata, the inclusion problem, i.e. the language $\mathcal{L} = \{(\mathcal{T}, \mathcal{T}') \mid \mathcal{R}(\mathcal{T}) \subseteq \mathcal{R}(\mathcal{T}')\}$ is $\Pi_2^1$-complete. Thus, the complement, $\overline{\mathcal{L}}$, is $\Sigma_2^1$-complete.
\end{proposition}

Using this fact, we are now ready to classify the recursion-theoretic strength of AAPA:

\begin{proof}[Proof of \autoref{thm:aapasigmaone}]
We reduce from the complement of the inclusion problem for two-tape Büchi automata $\mathcal{T}, \mathcal{T}'$.
For this purpose, we construct an AAPA $\mathcal{A}$ with two tapes.
Trivially, $\mathcal{T}$ can be converted to AAPA since the transitions depending on multiple input symbols can be simulated stepwise.
Furthermore, by \autoref{theorem:AAPAclosure}, AAPA are closed under complement.
We can thus build an AAPA accepting $\mathcal{R}(\mathcal{T})$, an AAPA accepting  $\overline{\mathcal{R}(\mathcal{T}')}$ and use conjunctive alternation to enforce that an input tuple is accepted by both automata, resulting in $\mathcal{A}$.
The reduction outlined above is a $1$-reduction since it is obviously total (i) and computable(ii), every tuple of two-tape Büchi automaton is assigned a unique AAPA (iii) and  $\mathcal{L}(\mathcal{A})$ is non-empty iff $\mathcal{R}(\mathcal{T}) \cap \overline{\mathcal{R}(\mathcal{T}')}$ is non-empty (iv). 
\end{proof}

Finally, the following well-known fact illustrates why $\Sigma_2^1$-hard problems are highly intractable and not subject to exhaustive approximation analyses:

\begin{proposition}
No $\Sigma_2^1$ hard problem is arithmetical. In particular, no $\Sigma_2^1$-hard problem is recursively enumerable or co-enumerable.
\end{proposition}

\subsection{Proof of \autoref{lem:contextboundedtest}}

\begin{proof}
	For the first direction, assume that there is a $k$-context-bounded accepting run $T$ of $\mathcal{A}$.
	We use the accepting run on $(w_1,...,w_n)$ to construct accepting runs of $\mathcal{S}$ on $w_d$ starting in $\bigwedge \{ g' \in F_d(\{(q_0,g)\}) \mid q_0 \in Q_0 \}$, where $g$ is as well constructed from the accepting run.
	
	As the first step, we construct $g$ from the run.
	We divide the run into maximal connected sections such that each section moves forward in a single direction.
	Since $T$ is $k$-context-bounded, a path in it can switch between different sections $k-1$ times at most.
	The guess $g$ is now defined recursively.
	On the first level, we choose the set of states where $T$ enters different sections than the states entered from the states $q_0 \in Q_0$ belong to.
	These states $q$ are then each combined with a set of states where the subtree $T'$ of $T$ starting in $q$ enters a different section than $q$ belongs to.
	Repeat the process until a \textit{bottom section} is reached after at most $k-1$ steps, where the states do not have to be combined with further guesses, yielding a $g \in G$.
	
	Let $d \in \{1,...,n\}$ be an arbitrary direction.
	We show that $\mathcal{S}$ has an accepting run from $\bigwedge \{ g' \in F_d(\{(q_0,g)\}) \mid q_0 \in Q_0 \}$ on $w_d$.
	Note that due to the way $F_d$ is defined, it yields a set of states $\{q_1,...,q_m\}$ with corresponding guesses $g_1,...,g_m$ such that $T$ enters sections with direction $d$ first exactly in states $q_1,...,q_m$.
	Therefore, when reaching some $q_i$, $T$ has not moved forward in direction $d$ yet, and an accepting run from $\bigwedge \{ g' \in F_d(\{(q_0,g)\}) \mid q_0 \in Q_0 \}$ on $w_d$ consists of the conjunction of accepting runs from $q_i$ for all $i$.
	These runs are given by taking the subtree $T_i$ of $T$ starting in $q_i$ and erasing every section not belonging to direction $d$ in it.
	Whenever a section belonging to direction $d$ is disconnected from $T_i$ in this way, we reapply it at the point where the connecting sections were erased.
	After this process is completed, we insert $\true$ loops at the end of every finite path.
	This way, we indeed have a run of $\mathcal{S}$ since erasure of non-$d$ sections and reapplication of $d$-sections corresponds to conjunctive moves into $F_d(g'')$ in the definition of $\rho_S$.
	Insertion of $\true$ loops corresponds to empty sets $F_d(g'')$ and thus empty conjunctions (which are equivalent to moves to $\true$) in the definition of $\rho$.
	Also, it is indeed an accepting run.
	Infinite paths in the run that are constructed from infinite paths in $T_i$ are obtained by erasing finite subpaths.
	Therefore the parity-condition stays fulfilled.
	Paths ending in $\true$ loops are fulfilled by default.
	
	For the other direction, assume that there is a $g \in G$ such that $\mathcal{S}$ accepts $w_d$ from $\bigwedge \{ g' \in F_d(\{(q_0,g)\}) \mid q_0 \in Q_0 \}$ for all $d \in \{1,...,n\}$.
	We use all accepting runs on $w_d$ to construct a $k$-context-bounded accepting run of $\mathcal{A}$ on $(w_1,...,w_n)$.
	
	This works in exactly the opposite way multiple runs were created from a single one in the first direction.
	First we notice, that acceptance from $\bigwedge \{ g' \in F_d(\{(q_0,g)\}) \mid q_0 \in Q_0 \}$ is induced by a set of runs from $q_i$ for each $(q_i,g_i) \in \bigcup_{q_0 \in Q_0} F_d(\{(q_0,g)\})$.
	We erase $\true$ loops induced by empty $F_d(g'')$ sets.
	Next, we cut off transitions to $F_d(g'')$ and obtain even more partial accepting runs in $\mathcal{S}$.
	Given these, we reconnect them according to the guesses we have made in $g$:
	Whereever a transition to $F_d(g'')$ for some $(q',g'')$ was removed, instead transition to $q'$.
	Additionally, each labelling with $(q,g)$ is replaced with $q$.
	Then, we have indeed obtained a run of $\mathcal{A}$, since $\rho_S$ is built in a way that given a set of pairs $(q_i,g_i)$ in $g$, we have a transition $\bigwedge Q'_d \times \{g\} \land \bigwedge_i F_{\gamma(q)}(g_i)$ in $\mathcal{S}$ iff we have a transition $\bigwedge Q'_d \land \bigwedge_i q_i$ in $\mathcal{A}$.
	It is $k$-context-bounded because guesses in $g$ were only made $k-1$ levels deep and reconnection of partial runs on different directions was only performed according to the guesses.
	Also, it is indeed an accepting run since all infinite paths have infinite subpaths in one of the partial accepting runs and thus obey the parity-condition.
\end{proof}

\subsection{Reasons for restricting contexts to a single direction}

We now elaborate on the remark about $k$-context-bounded AAPA with additional synchronous steps.
We call an AAPA $\mathcal{A}$ $k$-sync-context-bounded iff it switches between sync-contexts at most $k-1$ times, where a sync-context can either be a context or a path where all directions are advanced synchronously.
More formally, we refine \autoref{definition:context}.
We call a  sequence of nodes $t_i, ..., t_j$ in a run over words $w_1, ..., w_n$ a \textit{synchronous block} if for every direction $d$, we have $c_j^d = c_i^d + 1$, i.e. every direction has been progressed by exactly one step.
A sync-context is a (possibly infinite) path $p = t_1 t_2 ...$ in a run of an AAPA over $w_1,...,w_n$ such that transitions between successive states all use the same direction and otherwise, all directions are advanced.
That means there is a $d \in M$ such that for all $i \in \{1,...,|p|\}$, either  $c_{i+1}^d = c_i^d + 1$ or otherwise, $p$ is a concatenation of synchronous blocks.
We call a run $T$ of an AAPA $k$-sync-context-bounded, if every path in $T$ switches between different sync-contexts at most $k-1$ times.
	
\begin{theorem}
The problem to decide whether there is a $k$-sync-context-bounded accepting run of an AAPA and thus the emptiness problem for $k$-sync-context-bounded AAPA is undecidable.
\end{theorem}

\begin{proof}
Let $\mathcal{M}$ be a deterministic Turing machine. For $\mathcal{M}$, we build a $3$-sync-context-bounded AAPA $\mathcal{A}$ recognizing an encoding of accepting runs of $\mathcal{M}$ with two directions as follows:

One part of $\mathcal{A}$ checks synchronously whether both directions contain the same input.
It can also check whether the sequences in the directions represent valid configurations separated by a marker, whether the first configuration is the initial configuration and whether an accepting state is eventually reached.

To check whether the sequence of configurations in the directions is constructed in accordance with the transition function of $\mathcal{M}$, we combine asynchronous with synchronous contexts in a second part of $\mathcal{A}$:
we first progress both directions to the start of a configuration via conjunctive alternation.
Then, a context switch is performed to asynchronously advance one of the directions to the next configuration.
After a second context switch, we can iterate through both configurations synchronously while checking that they satisfy the transition relation of $\mathcal{M}$.
Note that this approach of checking the transition relation does not require a size bound on the configurations of $\mathcal{M}$ since any configuration is at most one tape cell larger than its predecessor.
It also shows that only two context switches are sufficient for undecidability in the presence of synchronous contexts.
\end{proof}

\section{Missing proofs from Section \ref{sec:hypercalculus}}

\subsection{Proof of \autoref{thm:monotone}}

\begin{proof}	
	We show the claims by induction over the structure of $\psi$.
	Therefore we assume $\psi$ is given in positive normal form.
	
	Let $X$ be an arbitrary predicate, $\mathcal{V}$ be an arbitrary predicate valutaion, $\Pi$ be an arbitrary path assignment and $k \in \mathbb{N}_0 \cup \{\infty\}$.
	Let $\xi \sqsubseteq \xi'$ for $\xi,\xi' : PA \to 2^{\kvari}$
	in the following cases.
	For monotonicity, we show that $\alpha(\xi) \sqsubseteq \alpha(\xi')$ in each case.
	
	\textbf{Case 1:} $\psi = a_{\pi_i}$ or $\psi = \lnot a_{\pi_i}$ (both cases are analogous, only doing the first)
	\begin{align*}
		\alpha(\xi) &= \llbracket a_{\pi_i} \rrbracket^{\mathcal{V}[X \mapsto \xi]}_k \\
		&= \lambda \Pi .\{(j_1,...,j_n) \in \kvari | a \in L(\Pi(\pi_i)(j_i))\} \\
		&= \llbracket a_{\pi_i} \rrbracket^{\mathcal{V}[X \mapsto \xi']}_k \\
		&= \alpha(\xi')
	\end{align*}

	\textbf{Case 2:} $\psi = Y$
	\begin{itemize}
		\item Case 2.1: $X = Y$
		\begin{align*}
			\alpha(\xi) &= \llbracket Y \rrbracket^{\mathcal{V}[X \mapsto \xi]}_k \\
			&= \mathcal{V}[X \mapsto \xi](Y) \\
			&= \mathcal{V}[Y \mapsto \xi](Y) \\
			&= \xi \\
			&\sqsubseteq \xi' \\
			&= \mathcal{V}[Y \mapsto \xi'](Y) \\
			&= \mathcal{V}[X \mapsto \xi'](Y) \\
			&= \llbracket Y \rrbracket^{\mathcal{V}[X \mapsto \xi']}_k \\
			&= \alpha(\xi')
		\end{align*}

		\item Case 2.2: $X \neq Y$
		\begin{align*}
			\alpha(\xi) &= \llbracket Y \rrbracket^{\mathcal{V}[X \mapsto \xi]}_k \\
			&= \mathcal{V}[X \mapsto \xi](Y) \\
			&= \mathcal{V}[X \mapsto \xi'](Y) \\
			&= \llbracket Y \rrbracket^{\mathcal{V}[X \mapsto \xi']}_k \\
			&= \alpha(\xi')
		\end{align*}
	\end{itemize}

	\textbf{Case 3:} $\psi = \bigcirc_{\pi_i} \psi'$
	\begin{itemize}
		\item By induction hypothesis, $\alpha'(\xi) := \llbracket \psi' \rrbracket^{\mathcal{V}[X \mapsto \xi]}_k$ is monotone
		\begin{align*}
			\alpha(\xi) =& \llbracket \bigcirc_{\pi_i} \psi' \rrbracket^{\mathcal{V}[X \mapsto M]}_k \\
			=& \lambda \Pi. \{(j_1,...,j_n)\in \kvari | \\
			& (j_1,...,j_i +1,...,j_n) \in \llbracket \psi' \rrbracket^{\mathcal{V}[X \mapsto \xi]}(\Pi)\}_k \\
			=& \lambda \Pi. \{(j_1,...,j_n)\in \kvari | \\
			& (j_1,...,j_i +1,...,j_n) \in g'(\xi)(\Pi)\} \\
			\sqsubseteq& \lambda \Pi. \{(j_1,...,j_n)\in \kvari | \\
			&(j_1,...,j_i +1,...,j_n) \in g'(\xi')(\Pi)\} \\
			=& \lambda \Pi. \{(j_1,...,j_n)\in \kvari | \\
			&(j_1,...,j_i +1,...,j_n) \in \llbracket \psi' \rrbracket^{\mathcal{V}[X \mapsto \xi']}_k(\Pi)\} \\
			=& \llbracket \bigcirc_{\pi_i} \psi' \rrbracket^{\mathcal{V}[X \mapsto \xi']}_k \\
			=& \alpha(\xi')
		\end{align*}
	\end{itemize}
	
	\textbf{Case 4:} $\psi = \psi' \lor \psi''$
	\begin{itemize}
		\item By induction hypothesis, $\alpha'(\xi) := \llbracket \psi' \rrbracket^{\mathcal{V}[X \mapsto \xi]}_k$ and $\alpha''(\xi) := \llbracket \psi'' \rrbracket^{\mathcal{V}[X \mapsto \xi]}_k$ are monotone
		\begin{align*}
			\alpha(\xi) &= \llbracket \psi' \lor \psi'' \rrbracket^{\mathcal{V}[X \mapsto \xi]}_k \\
			&= \lambda \Pi. (\llbracket \psi' \rrbracket^{\mathcal{V}[X \mapsto \xi]}_k(\Pi) \cup \llbracket \psi'' \rrbracket^{\mathcal{V}[X \mapsto \xi]}_k(\Pi)) \\
			&= \llbracket \psi' \rrbracket^{\mathcal{V}[X \mapsto \xi]}_k \sqcup \llbracket \psi'' \rrbracket^{\mathcal{V}[X \mapsto \xi]}_k \\
			&\sqsubseteq \llbracket \psi' \rrbracket^{\mathcal{V}[X \mapsto \xi']}_k \sqcup \llbracket \psi'' \rrbracket^{\mathcal{V}[X \mapsto \xi']}_k \\
			&= \lambda \Pi. (\llbracket \psi' \rrbracket^{\mathcal{V}[X \mapsto \xi']}_k(\Pi) \cup \llbracket \psi'' \rrbracket^{\mathcal{V}[X \mapsto \xi']}_k(\Pi)) \\
			&= \llbracket \psi' \lor \psi'' \rrbracket^{\mathcal{V}[X \mapsto \xi']}_k \\
			&= \alpha(\xi')
		\end{align*}
	\end{itemize}

	\textbf{Case 5:} $\psi = \psi' \land \psi''$
	\begin{itemize}
		\item Monotonicity can be shown analogous to case $4$.
	\end{itemize}

	\textbf{Case 6:} $\psi = \mu Y. \psi'$ or $\psi = \nu Y. \psi'$
	\begin{itemize}
		\item Since we are assuming positive normal form, all bound path predicates are distinct. Therefore $X \neq Y$
		\item By induction hypothesis, $\alpha'(\xi) := \llbracket \psi' \rrbracket^{\mathcal{V}[Y \mapsto \xi'][X \mapsto \xi]}_k$ is monotone for all $\xi'$
		\item Doing the case for a least fixpoint, the other case is analogous
		\begin{align*}
			\alpha(\xi) &= \llbracket \mu Y . \psi' \rrbracket^{\mathcal{V}[X \mapsto \xi]}_k \\
			&= \bigsqcap \{\xi'' : PA \to 2^{\kvari} | \xi'' \sqsupseteq \llbracket \psi' \rrbracket^{\mathcal{V}[X \mapsto \xi][Y \mapsto \xi'']}_k\} \\
			&= \bigsqcap \{\xi'' : PA \to 2^{\kvari} | \xi'' \sqsupseteq \llbracket \psi' \rrbracket^{\mathcal{V}[Y \mapsto \xi''][X \mapsto \xi]}_k\} \\
			&\sqsubseteq \bigsqcap \{\xi'' : PA \to 2^{\kvari} | \xi'' \sqsupseteq \llbracket \psi' \rrbracket^{\mathcal{V}[Y \mapsto \xi''][X \mapsto \xi']}_k\} \\
			&= \bigsqcap \{\xi'' : PA \to 2^{\kvari} | \xi'' \sqsupseteq \llbracket \psi' \rrbracket^{\mathcal{V}[X \mapsto \xi'][Y \mapsto \xi'']}_k\} \\
			&= \llbracket \mu Y . \psi' \rrbracket^{\mathcal{V}[X \mapsto \xi']}_k \\
			&= \alpha(\xi')
		\end{align*}
	\end{itemize}
\end{proof}

\subsection{Proof of \autoref{cor:welldefined}}

\begin{proof}
	Both claims ensue from the fact that $(PA \to 2^{\mathbb{N}_0^n},\sqsubseteq)$ is a complete lattice, \autoref{thm:monotone} and Knaster-Tarski's fixpoint theorem.
\end{proof}

\subsection{Proof of \autoref{thm:monotonek}}

\begin{proof}
	Mostly straightforward structural induction on $\psi$ using the fact that $G_k \subseteq G_{k'}$ for $k \leq k'$.
	In the fixpoint case $\psi = \mu X. \psi_1$, we need to establish that $\beta_k(\xi) = \llbracket \psi_1 \rrbracket^{\mathcal{V}[X \mapsto \xi]}_k \sqsubseteq \beta_{k'}(\xi) = \llbracket \psi_1 \rrbracket^{\mathcal{V}[X \mapsto \xi]}_{k'}$ and therefore $lfp(\beta_k) \sqsubseteq lfp(\beta_{k'})$.
\end{proof}

\subsection{Proof of \autoref{cor:synchronoushierarchy}}

\begin{proof}
	We show that $\Pi \models_{k}^{\mathcal{K}} \varphi$ implies $\Pi \models_{k'}^{\Pi} \varphi$.
	The claim then follows immediately.
	Fix a Kripke Structure $\mathcal{K}$, a formula $\varphi$ and some $k,k'$ with $k \leq k'$.
	
	For an existential quantifier $\exists \pi . \varphi$, we have to show $\Pi[\pi \mapsto p] \models_{k}^{\mathcal{K}} \varphi$ for some $p \in Paths(\mathcal{K})$ implies $\Pi[\pi \mapsto p'] \models_{k'}^{\mathcal{K}} \varphi$ for some $p' \in Paths(\mathcal{K})$.
	Indeed, for $p = p'$, the claim follows from the induction hypothesis.
	
	The case for a universal quantifier is analogous.
	
	For a quantifier-free formula $\psi$, we have to show that $(0,...,0) \in \llbracket \psi \rrbracket^{\mathcal{V}}_{k}(\Pi)$ for some $\mathcal{V}$ implies $(0,...,0) \in \llbracket \psi \rrbracket^{\mathcal{V}'}_{k'}(\Pi)$ for some $\mathcal{V}'$.
	Indeed, by \autoref{thm:monotonek}, the claim holds for $\mathcal{V}' = \mathcal{V}$.
\end{proof}

\subsection{Formal semantics on traces}

Let $\mathcal{T}$ be a set of traces.
We call a function $\Pi: N \to \mathcal{T}$ a trace assignment and denote by $\TA$ the set of all trace assignments.
Then we use $\mathcal{V}: \chi \to \TA \to 2^{\mathbb{N}_0^n}$ to denote a predicate valuation.
Manipulations on these functions are defined as for path assignments.


We again differentiate between semantics for the two types of formulas: quantifier semantics and trace semantics.
For a quantified formula $\varphi$, we write $\mathcal{T} \models_k \varphi$ to denote that the set of traces $\mathcal{T}$ fulfills the formula $\varphi$, 
i.e. iff $\{\} \models_k^{\mathcal{T}} \varphi$ for the empty trace assignment $\{\}$.
For a quantifier-free formula $\psi$, we instead consider a semantics similar to the path semantics from \autoref{def:pathsemantics}, with the difference that we consider trace assignments instead of path assignments.

\begin{definition}[Quantifier semantics]
	\begin{align*}
	\Pi \models_k^{\mathcal{T}} \exists \pi . \varphi  \text{ iff }& \Pi[\pi \mapsto t] \models_k^{\mathcal{T}} \varphi \text{ for some } t \in \mathcal{T} \\
	\Pi \models_k^{\mathcal{T}} \forall \pi . \varphi  \text{ iff }& \Pi[\pi \mapsto t] \models_k^{\mathcal{T}} \varphi \text{ for all } t \in \mathcal{T} \\
	\Pi \models_k^{\mathcal{T}} \psi \text{ iff }& (0,...,0) \in \llbracket \psi \rrbracket_{k}^{\mathcal{V}}(\Pi) \text{ for some } \mathcal{V}
	\end{align*}
	for a quantified formula $\varphi$ and a quantifier-free formula $\psi$.
\end{definition}
\begin{definition}[Trace Semantics]
	\begin{align*}
		\llbracket a_{\pi_i} \rrbracket^{\mathcal{V}}_{k} :=& \lambda \Pi . \{(j_1,...,j_n) \in \kvari | a \in \Pi(\pi_i)(j_i)\} \\
		\llbracket X \rrbracket^{\mathcal{V}}_{k} :=& \mathcal{V}(X) \\
		\llbracket \bigcirc_{\pi_i} \varphi \rrbracket^{\mathcal{V}}_{k} :=& \lambda \Pi . \{(j_1,...,j_n) \in \kvari | \\
		&(j_1,...,j_i + 1,...,j_n) \in \llbracket \varphi \rrbracket^{\mathcal{V}}_{k}(\Pi)\} \\
		\llbracket \varphi \lor \varphi' \rrbracket^{\mathcal{V}}_{k} :=& \lambda \Pi .\llbracket \varphi \rrbracket^{\mathcal{V}}_{k}(\Pi) \cup \llbracket \varphi' \rrbracket^{\mathcal{V}}_{k}(\Pi) \\
		\llbracket \lnot \varphi \rrbracket^{\mathcal{V}}_{k} :=& \lambda \Pi . \kvari \setminus \llbracket \varphi \rrbracket^{\mathcal{V}}_{k}(\Pi) \\
		\llbracket \mu X . \varphi \rrbracket^{\mathcal{V}}_{k} :=& \bigsqcap \{\xi : TA \to 2^{\kvari} | \xi \sqsupseteq \llbracket \varphi \rrbracket^{\mathcal{V}[X \mapsto \xi]}_{k}\}
	\end{align*}
\end{definition}

\section{Missing proofs from Section \ref{sec:connection}}

\subsection{Proof of claim about well-formed valuations}

\begin{theorem}
	If $\mathcal{V}$ is a well-formed valuation, then $\llbracket \psi \rrbracket^{\mathcal{V}}_{k'}$ is well-formed, i.e. for all vectors $v,v'$ and path assignments $\Pi,\Pi'$ with $\Pi[v] = \Pi'[v']$ it holds: $v \in \llbracket \psi \rrbracket^{\mathcal{V}}_{k'} (\Pi)$ iff $v' \in \llbracket \psi \rrbracket^{\mathcal{V}}_{k'}(\Pi')$.
\end{theorem}

\begin{proof}
	We show the claim by a structural induction on $\psi$.
	Let $v,v'$ therefore be arbitrary vectors and $\Pi,\Pi'$ be arbitrary path assignments with $\Pi[v] = \Pi'[v']$ in the following cases.
	
	\textbf{Atomic propositions:} this case follows immediately from the assumption that $\Pi[v] = \Pi'[v']$.
	
	\textbf{Predicates:} this case follows immediately from the assumption that $\mathcal{V}$ is well-formed.
	
	\textbf{Next:} this case follows from the assumption that $\Pi[v] = \Pi'[v']$ and the induction hypothesis.
	
	\textbf{Boolean connecitves:} this case follows immediately from the assumption that $\Pi[v] = \Pi'[v']$ and the induction hypothesis.
	
	\textbf{Fixpoints:} we use the approximant characterisation $\bigsqcup_{\kappa \geq 0} \alpha^\kappa (\bot)$ for $\psi$ and show the claim by a transfinite induction on $\kappa$.
	We thus use $(IH1)$ for the structural induction's hypothesis and $(IH2)$ for the transfinite induction's hypothesis.

	In the \textit{base case} of the transfinite induction, $\kappa = 0$, we trivially have $v \in \alpha^0(\bot)(\Pi) = \emptyset$ iff $v' \in \alpha^0(\bot)(\Pi') = \emptyset$.
	
	In the \textit{inductive step} of the transfinite induction, $\kappa \mapsto \kappa+1$, we have $v \in \alpha^{\kappa+1}(\bot)(\Pi) = \llbracket \psi \rrbracket_{k}^{\mathcal{V}[X \mapsto \alpha^k(\bot)]}(\Pi)$ and show that $v' \in \alpha^{\kappa+1}(\bot)(\Pi') = \llbracket \psi \rrbracket_{k}^{\mathcal{V}[X \mapsto \alpha^k(\bot)]}(\Pi')$.
	Therefore notice that $(IH2)$ implies that $\mathcal{V}[X \mapsto \alpha^\kappa(\bot)]$ is a well-formed valuation.
	Then the claim follows immediately from $(IH1)$.
	
	In the \textit{limit case} of the transfinite induction, $\kappa < \lambda \mapsto \lambda$, let $v \in \alpha^{\lambda}(\bot)(\Pi) = \bigcup_{\kappa < \lambda} \alpha^\kappa(\bot)(\Pi)$.
	Therefore there is a $\kappa < \lambda$ such that $v \in \alpha^\kappa(\bot)(\Pi)$.
	The induction hypothesis $(IH2)$ then implies that $v' \in \alpha^\kappa(\bot)(\Pi')$ and thus $v' \in \alpha^{\lambda}(\bot)(\Pi')$.
\end{proof}

If we now consider starting from a trivially well-formed valuation like $\lambda X . \bot$, then we can see from the transfinite induction in the fixpoint case of the proof, that only well-formed valuations occur in fixpoint iterations of $\Hmu$.

\subsection{Proof of \autoref{thm:aapahmuequivalence} part 1}

In this section, we prove that the construction presented in \autoref{subsec:hmutoaapa} indeed results in an automaton that is $\mathcal{K}$-equivalent to $\psi$.

\begin{proof}
	By structural induction over the structure of $\psi$.
	Let $v = (v_1,...,v_n) \in \mathbb{N}_0^n$, let $\Pi$ with $\Pi(\pi_i) = s_i^0 s_i^1 ...$ be an arbitrary path assignment and let $\mathcal{V}$ an arbitrary predicate valuation in the following cases.
	
	\textbf{Atomic propositions:}
	Doing the case for $a_{\pi_i}$, the negated case is analogous.
	It is easy to see that $\mathcal{L}(\mathcal{A}_{a_{\pi_i}}[X_1: \mathcal{L}(\mathcal{V}(X_1)(\Pi)),...,X_m:\mathcal{L}(\mathcal{V}(X_m)(\Pi))]) = \mathcal{L}(\mathcal{A}_{a_{\pi_i}})$.
	
	Let $v \in \llbracket a_{\pi_i} \rrbracket^{\mathcal{V}}(\Pi)$.
	By the definition of semantics we have $a \in L(\Pi(\pi_i)(v_i))$, thus for the first symbol $s_i^{v_i}$ of $w_{\Pi}[v]$ in direction $i$ we have $a \in L(s_i^{v_i})$, which induces an accepting run of $\mathcal{A}_{a_{\pi_i}}$.
	Therefore we have $w_{\Pi}[v] \in \mathcal{L}(\mathcal{A}_{a_{\pi_i}})$.
	
	Let $w_{\Pi}[v] \in \mathcal{L}(\mathcal{A}_{a_{\pi_i}})$.
	By construction, the accepting run of $\mathcal{A}_{a_{\pi_i}}$ on $w_{\Pi}[v]$ has to move to $(tt)$ with the first symbol read in direction $i$, which implies $a \in L(s_i^{v_i})$ for the first symbol $s_i^{v_i}$ of $w_{\Pi}[v]$ in direction $i$.
	By the definition of $w_{\Pi}$ we then have $a \in L(\Pi(\pi_i)(v_i))$, which by the definition of semantics directly implies $v \in \llbracket a_{\pi_i} \rrbracket^{\mathcal{V}}(\Pi)$.
	
	\textbf{Predicates:}
	We have $\mathcal{L}(\mathcal{A}_{X_i}[X_1: \mathcal{L}(\mathcal{V}(X_1)(\Pi)),...,X_m:\mathcal{L}(\mathcal{V}(X_m)(\Pi))]) = \mathcal{L}(\mathcal{A}_{X_i}[X_i : \mathcal{L}(\mathcal{V}(X_i)(\Pi))]) = \mathcal{L}(\mathcal{V}(X_i)(\Pi))$ since $\mathcal{A}_{X_i}$ only consists of the hole $X_i$.
	
	Let $v \in \llbracket X_i \rrbracket^{\mathcal{V}}(\Pi)$.
	By the definition of semantics, we have $v \in \mathcal{V}(X_i)(\Pi)$, which directly implies 
	$w_{\Pi}[v] \in \mathcal{L}(\mathcal{V}(X_i)(\Pi))$.
	
	Let $w_{\Pi}[v] \in \mathcal{L}(\mathcal{V}(X_i)(\Pi))$.
	By the definition of $w_{\Pi}$, this implies $v \in \mathcal{V}(X_i)(\Pi)$, which in turn implies $v \in \llbracket X_i \rrbracket^{\mathcal{V}}(\Pi)$.
	
	\textbf{Boolean connectives:}
	Doing the case for disjunction, other case is analogous.
	We divide the free variables $X_1,...,X_m$ into two sets $Y_1,...,Y_{m_1}$ and $Z_1,...,Z_{m_2}$ (which may be non-disjoint) such that $Y_1,...,Y_{m_1}$ are the free variables in $\psi_1$ and $Z_1,...,Z_{m_2}$ are the free variables in $\psi_2$.
	By construction, we have $\mathcal{L}(\mathcal{A}_{\psi_1 \lor \psi_2})[X_1: \mathcal{L}(\mathcal{V}(X_1)(\Pi)),...,X_m:\mathcal{L}(\mathcal{V}(X_m)(\Pi))] = \mathcal{L}(\mathcal{A}_{\psi_1})[Y_1:\mathcal{L}(\mathcal{V}(Y_1)(\Pi)),...,Y_{m_1}:\mathcal{L}(\mathcal{V}(Y_{m_1})(\Pi))] \cup \mathcal{L}(\mathcal{A}_{\psi_2})[Z_1:\mathcal{L}(\mathcal{V}(Z_1)(\Pi)),...,Z_{m_2}:\mathcal{L}(\mathcal{V}(Z_{m_2})(\Pi))]$.
	Then both directions of the claim ensue from the induction hypothesis.
	
	\textbf{Next:}
	Handling a formula of the form $\bigcirc_{\pi_i} \psi_1$.
	Note that $X_1,...,X_m$ are exactly the free variables of $\psi_1$ as well.
	By construction we have $w[v] \in \mathcal{L}(\mathcal{A}_{\bigcirc_{\pi_i} \psi_1}[X_1: \mathcal{L}(\mathcal{V}(X_1)(\Pi)),...,X_m:\mathcal{L}(\mathcal{V}(X_m)(\Pi))])$ iff $w[v + e_i] \in \mathcal{L}(\mathcal{A}_{\psi_1}[X_1: \mathcal{L}(\mathcal{V}(X_1)(\Pi)),...,X_m:\mathcal{L}(\mathcal{V}(X_m)(\Pi))])$ for all words $w$ and the unit vector $e_i$ with $1$ in component $i$.
	Then both directions of the claim ensue directly from the induction hypothesis.

	\textbf{Fixpoint expressions:}
	Doing the case for a $\mu$ formula $\mu X. \psi_1$, the other case is analogous.
	Since we assume that every path predicate is bound by a unique fixpoint expression, we have $X \in free(\psi_1)$ and $X \not\in free(\mu X. \psi_1)$.
	Note that the priority of state $(X)$ is an odd strict lower bound on the priorities in $\mathcal{A}_{\mu X. \psi_1}$.
	Thus, every path of an accepting run of $\mathcal{A}_{\mu X. \psi_1}$ can only visit $(X)$ finitely often.
	Also, by construction, any path of a run visiting $(X)$ at some point must then proceed from the start of the automaton in the next step.
	Therefore, $\mathcal{L}(\mathcal{A}_{\mu X. \psi_1}[X_1: \mathcal{L}(\mathcal{V}(X_1)(\Pi)),...,X_m:\mathcal{L}(\mathcal{V}(X_m)(\Pi))])$ can be characterised as the least fixpoint of the function $f: \mathcal{L} \mapsto \mathcal{L}(\mathcal{A}_{\psi_1}[X_1: \mathcal{L}(\mathcal{V}(X_1)(\Pi)),...,X_m:\mathcal{L}(\mathcal{V}(X_m)(\Pi)),X: \mathcal{L}])$, or by a union of its approximants $\bigcup_{\kappa \geq 0} f^\kappa(\emptyset)$ where $f^0(\mathcal{L}) = \mathcal{L}$, $f^{\kappa+1}(\mathcal{L}) = f(f^\kappa(\mathcal{L}))$ and $f^{\lambda}(\mathcal{L}) = \bigcup_{\kappa < \lambda} f^\kappa(\mathcal{L})$ for ordinals $\kappa$ and limit ordinals $\lambda$. 
	On the other hand, $\llbracket \mu X. \psi_1 \rrbracket^{\mathcal{V}}$ is the least fixpoint of the function $\alpha: \xi \mapsto \llbracket \psi_1 \rrbracket^{\mathcal{V}[X \mapsto \xi]}$ and can be characterised as a union of its approximants $\bigsqcup_{\kappa \geq 0} \alpha^\kappa(\bot)$ where $\alpha^0(\xi) = \xi$, $\alpha^{\kappa+1}(\xi) = \alpha(\alpha^\kappa(\xi))$ and $\alpha^{\lambda}(\xi) = \bigsqcup_{\kappa < \lambda} \alpha^\kappa(\xi)$ by \autoref{cor:welldefined}.
	We now show that $v \in \alpha^\kappa(\bot)(\Pi)$ iff $w_{\Pi}[v] \in f^\kappa(\emptyset)$ (*) for all ordinals $\kappa \geq 1$, which establishes the theorem for this case.
	Indeed, for arbitrary $\kappa$, this ensues directly from the induction hypothesis and the following claim:
	
	\textbf{Claim:} $w_{\Pi}[v] \in \mathcal{L}(\mathcal{A}_{\psi_1}[X_1 : \mathcal{L}(\mathcal{V}[X \mapsto \alpha^\kappa(\bot)](X_1)(\Pi)),...,X_m : \mathcal{L}(\mathcal{V}[X \mapsto \alpha^\kappa(\bot)](X_m)(\Pi))])$ iff $w_{\Pi}[v] \in  \mathcal{L}(\mathcal{A}_{\psi_1}[X_1 : \mathcal{L}(\mathcal{V}(X_1)(\Pi)),...,X_m : \mathcal{L}(\mathcal{V}(X_m)(\Pi)), X: f^\kappa(\emptyset)])$. 
	
	
	We prove the claim by a transfinite induction on $\kappa$.
	To avoid confusion, we will denote the structural induction's hypothesis by $(IH1)$ and this induction's hypothesis by $(IH2)$.
	Also, since (*) for some $\kappa$ follows directly from the claim for the same $\kappa$, we can use this as $(IH3)$.
	
	For the \textbf{base case} $\kappa = 0$, let $w_{\Pi}[v] \in \mathcal{L}(\mathcal{A}_{\psi_1}[X_1 : \mathcal{L}(\mathcal{V}[X \mapsto \bot](X_1)(\Pi)),...,X_m : \mathcal{L}(\mathcal{V}[X \mapsto \bot](X_m)(\Pi))])$.
	Since $\mathcal{V}[X \mapsto \bot](X)(\Pi) = \emptyset$, we can exchange $\mathcal{L}(\mathcal{V}[X \mapsto \bot](X)(\Pi))$ with $\emptyset$ in the substitutions.
	Furthermore, since $\mathcal{L}(\mathcal{V}[X \mapsto \bot](X_i)(\Pi))$ does not depend on $\mathcal{V}[X \mapsto \emptyset](X)$ for $X_i \neq X$, we can exchange the predicate valuations without changing language containment of $w_{\Pi}[v]$.
	
	On the other hand let $w_{\Pi}[v] \in \mathcal{L}(\mathcal{A}_{\psi_1}[X_1 : \mathcal{L}(\mathcal{V}(X_1)(\Pi)),...,X_m : \mathcal{L}(\mathcal{V}(X_m)(\Pi)), X : \emptyset])$.
	With the same arguments as before, we can exchange $\mathcal{V}$ with $\mathcal{V}[X \mapsto \bot]$ and $\emptyset$ with $\mathcal{L}(\mathcal{V}[X \mapsto \bot](X)(\Pi))$ in the substitutions without changing language containment of $w_{\Pi}[v]$, immediately giving us the desired result.
	
	In the \textbf{inductive step} $\kappa \mapsto \kappa+1$, let $w_{\Pi}[v] \in \mathcal{L}(\mathcal{A}_{\psi_1}[X_1 : \mathcal{L}(\mathcal{V}[X \mapsto \alpha^{\kappa+1}(\bot)](X_1)(\Pi)),...,X_m : \mathcal{L}(\mathcal{V}[X \mapsto \alpha^{\kappa+1}(\bot)](X_m)(\Pi))])$.
	We show that the accepting run is also an accepting run of $\mathcal{A}_{\psi_1}[X_1 : \mathcal{L}(\mathcal{V}(X_1)(\Pi)),...,X_m : \mathcal{L}(\mathcal{V}(X_m)(\Pi)), X : f^{\kappa+1}(\emptyset)]$ on $w_{\Pi}[v]$.
	First notice that the run being an accepting run only depends on $\mathcal{V}[X \mapsto \alpha^{\kappa+1}(\bot)](X)$ in states $(X)$, thus replacing the predicate valuation with $\mathcal{V}$ will result in the same behaviour up to states $(X)$.
	Secondly, notice that for all $v' \in \mathbb{N}_0^n$, when state $(X)$ is reached with directions according to $v'$, then $w_{\Pi}[v+v'] \in f^{\kappa+1}(\emptyset)$.
	Therefore let $v'$ be an arbitrary vector such that $(X)$ is reached with directions according to $v'$.
	By definition of $\mathcal{L}(\mathcal{V}[X \mapsto \alpha^{\kappa+1}(\bot)](X)(\Pi))$, we then have $v+v' \in \alpha^{\kappa+1}(\bot)(\Pi) = \alpha(\alpha^\kappa(\bot))(\Pi) = \llbracket \psi_1 \rrbracket^{\mathcal{V}[X \mapsto \alpha^\kappa(\bot)]}(\Pi)$.
	Using $(IH1)$, we get that $w_{\Pi}[v+v'] \in \mathcal{L}(\mathcal{A}_{\psi_1}[X_1 : \mathcal{L}(\mathcal{V}[X \mapsto \alpha^\kappa(\bot)](X_1)(\Pi)),...,X_m : \mathcal{L}(\mathcal{V}[X \mapsto \alpha^\kappa(\bot)](X_m)(\Pi))])$.
	Now $(IH2)$ applies and we have $w_{\Pi}[v+v'] \in \mathcal{L}(\mathcal{A}_{\psi_1}[X_1 : \mathcal{L}(\mathcal{V}(X_1)(\Pi)),...,X_m : \mathcal{L}(\mathcal{V}(X_m)(\Pi)), X : f^\kappa(\emptyset)]) = f^{\kappa+1}(\emptyset)$.
	Thus, we have an accepting run of $\mathcal{A}_{\psi_1}[X_1 : \mathcal{L}(\mathcal{V}(X_1)(\Pi)),...,X_m : \mathcal{L}(\mathcal{V}(X_m)(\Pi)), X : f^{\kappa+1}(\emptyset)]$ on $w_{\Pi}[v]$, which implies $w_{\Pi}[v] \in \mathcal{L}(\mathcal{A}_{\psi_1}[X_1 : \mathcal{L}(\mathcal{V}(X_1)(\Pi)),...,X_m : \mathcal{L}(\mathcal{V}(X_m)(\Pi)), X: f^{\kappa+1}(\emptyset)])$.
	
	On the other hand assume that $w_{\Pi}[v] \in \mathcal{L}(\mathcal{A}_{\psi_1}[X_1 : \mathcal{L}(\mathcal{V}(X_1)(\Pi)),...,X_m : \mathcal{L}(\mathcal{V}(X_m)(\Pi)), X: f^{\kappa+1}(\emptyset)])$.
	Again, we show that the accepting run is also an accepting run of $\mathcal{A}_{\psi_1}[X_1 : \mathcal{L}(\mathcal{V}[X \mapsto \alpha^{\kappa+1}(\bot)](X_1)(\Pi)),...,X_m : \mathcal{L}(\mathcal{V}[X \mapsto \alpha^{\kappa+1}(\bot)](X_m)(\Pi))]$ on $w_{\Pi}[v]$.
	We notice, that replacing $\mathcal{V}$ with $\mathcal{V}[X \mapsto \alpha^{\kappa+1}(\bot)]$ will result in the same behaviour up to states $(X)$.
	Next we notice, that for all vectors $v' \in \mathbb{N}_0^n$, when state $(X)$ is reached according to directions in $v'$ in the accepting run, then $w_{\Pi}[v+v'] \in \mathcal{L}(\mathcal{V}[X \mapsto \alpha^{\kappa+1}(\bot)](X)(\Pi))$.
	Therefore let $v'$ be an arbitrary vector such that $(X)$ is reached according to directions in $v'$.
	Due to the substitutions, this implies $w_{\Pi}[v+v'] \in f^{\kappa+1}(\emptyset) = \mathcal{L}(\mathcal{A}_{\psi_1}[X_1 : \mathcal{L}(\mathcal{V}(X_1)(\Pi)),...,X_m : \mathcal{L}(\mathcal{V}(X_m)(\Pi)), X : f^\kappa(\emptyset)])$.
	Now $(IH2)$ applies and we have $w_{\Pi}[v+v'] \in \mathcal{L}(\mathcal{A}_{\psi_1}[X_1 : \mathcal{L}(\mathcal{V}[X \mapsto \alpha^\kappa(\bot)](X_1)(\Pi)),...,X_m : \mathcal{L}(\mathcal{V}[X \mapsto \alpha^\kappa(\bot)](X_m)(\Pi))])$.
	Using $(IH1)$, we then get $v+v' \in \llbracket \psi_1 \rrbracket^{\mathcal{V}[X \mapsto \alpha^\kappa(\bot)]}(\Pi) = \alpha(\alpha^\kappa(\bot))(\Pi) = \alpha^{\kappa+1}(\bot)(\Pi)$ and thus $w_{\Pi}[v+v'] \in \mathcal{L}(\mathcal{V}[X \mapsto \alpha^{\kappa+1(\bot)}](X)(\Pi))$.
	These two facts imply that the run also witnesses $w_{\Pi}[v] \in \mathcal{L}(\mathcal{A}_{\psi_1}[X_1 : \mathcal{L}(\mathcal{V}[X \mapsto \alpha^{\kappa+1}(\bot)](X_1)(\Pi)),...,X_m : \mathcal{L}(\mathcal{V}[X \mapsto \alpha^{\kappa+1}(\bot)](X_m)(\Pi))])$.
	
	For the \textbf{limit case} $\kappa < \lambda \mapsto \lambda$, let $w_{\Pi}[v] \in \mathcal{L}(\mathcal{A}_{\psi_1}[X_1 : \mathcal{L}(\mathcal{V}[X \mapsto \alpha^{\lambda}(\bot)](X_1)(\Pi)),...,X_m : \mathcal{L}(\mathcal{V}[X \mapsto \alpha^{\lambda}(\bot)](X_m)(\Pi))])$.
	Just as in the inductive step, we have to show two claims: (i) replacing $\mathcal{V}[X \mapsto \alpha^{\lambda}(\bot)]$ with $\mathcal{V}$ results in the same behaviour up to states $(X)$ and (ii) for all $v' \in \mathbb{N}_0^n$, if $(X)$ is reached according to directions $v'$, then $w_{\Pi}[v + v'] \in f^{\lambda}(\emptyset)$.
	While claim (i) is trivial, we have to rely on more arguments for claim (ii).
	Thus let $v'$ be an arbitrary vector such that $(X)$ is reached according to directions $v'$ in the accepting run.
	By definition of $\mathcal{L}(\mathcal{V}[X \mapsto \alpha^{\lambda}(\bot)](X)(\Pi))$, we then have $v+v' \in \alpha^{\lambda}(\bot)(\Pi) = \bigsqcup_{\kappa < \lambda} \alpha^\kappa(\bot)(\Pi) = (\lambda \Pi'. \bigcup_{\kappa < \lambda} \alpha^\kappa(\bot)(\Pi'))(\Pi)$.
	Thus, there is a $\kappa$ such that $v+v' \in \alpha^\kappa(\bot)(\Pi)$.
	Using $(IH3)$ we then have $v+v' \in f^\kappa(\emptyset)$ implying $v+v' \in f^{\lambda}(\emptyset)$ and therefore claim (ii).
	Combining claims (i) and (ii), we can argue that the accepting run of $\mathcal{A}_{\psi_1}[X_1 : \mathcal{L}(\mathcal{V}[X \mapsto \alpha^{\lambda}(\bot)](X_1)(\Pi)),...,X_m : \mathcal{L}(\mathcal{V}[X \mapsto \alpha^{\lambda}(\bot)](X_m)(\Pi))]$ on $w_{\Pi}[v]$ is an accepting run of $\mathcal{A}_{\psi_1}[X_1 : \mathcal{L}(\mathcal{V}(X_1)(\Pi)),...,X_m : \mathcal{L}(\mathcal{V}(X_m)(\Pi)), X: f^{\lambda}(\emptyset)]$ on $w_{\Pi}[v]$ as well.
	
	On the other hand let $w_{\Pi}[v] \in \mathcal{L}(\mathcal{A}_{\psi_1}[X_1 : \mathcal{L}(\mathcal{V}(X_1)(\Pi)),...,X_m : \mathcal{L}(\mathcal{V}(X_m)(\Pi)), X: f^{\lambda}(\emptyset)])$.
	Again, we have to show two claims similar to those in the inductive step: (i) $\mathcal{V}$ can be replaced with $\mathcal{V}[X \mapsto \alpha^{\lambda}(\bot)]$ in the substitutions resulting in the same behaviour up to state $(X)$ and (ii) for all vectors $v' \in \mathbb{N}_0^n$, if state $(X)$ is reached according to directions in $v'$ in the accepting run, then $w_{\Pi}[v + v'] \in \mathcal{L}(\mathcal{V}[X \mapsto \alpha^{\lambda}(\bot)](X)(\Pi)) = \alpha^{\lambda}(\bot)(\Pi)$.
	The first claim is trivial.
	For the second claim, let $v'$ be a vector such that $(X)$ is reached according to directions in $v'$.
	Due to the substitutions, this implies $w_{\Pi}[v+v'] \in f^{\lambda}(\emptyset) = \bigcup_{\kappa < \lambda} f^\kappa(\emptyset)$.
	Thus, there is some $\kappa$ such that $w_{\Pi}[v + v'] \in f^\kappa(\emptyset)$.
	Now $(IH3)$ applies and we have $w_{\Pi}[v + v'] \in \alpha^{\kappa}(\bot)(\Pi) \subseteq (\bigsqcup_{\kappa < \lambda} \alpha^\kappa(\bot))(\Pi) = \alpha^{\lambda}(\bot)(\Pi)$, thus claim (ii) holds.
	Combining the two facts, we see that the accepting run of $\mathcal{A}_{\psi_1}[X_1 : \mathcal{L}(\mathcal{V}(X_1)(\Pi)),...,X_m : \mathcal{L}(\mathcal{V}(X_m)(\Pi)), X: f^{\lambda}(\emptyset)]$ on $w_{\Pi}[v]$ is an accepting run of $\mathcal{A}_{\psi_1}[X_1 : \mathcal{L}(\mathcal{V}[X \mapsto \alpha^{\lambda}(\bot)](X_1)(\Pi)),...,X_m : \mathcal{L}(\mathcal{V}[X \mapsto \alpha^{\lambda}(\bot)](X_m)(\Pi))]$ on $w_{\Pi}[v]$ as well.
\end{proof}

\subsection{Proof of \autoref{thm:aapahmuequivalence} part 2}

In this section, we prove that the construction presented in \autoref{subsec:aapatohmu} indeed results in a formula that is $\mathcal{K}$-equivalent to $\mathcal{A}$.
Therefore, we inductively construct intermediate automata $\mathcal{A}_i^h$ capturing exactly the behaviour of $\psi_i^h$.
We combine them in an automaton $\mathcal{A}_i$ moving into automata $\mathcal{A}_i^h$ according to $\mathcal{A}$'s starting function $\rho_0$.
The proof of our main result then ensues from two lemmas about $\mathcal{A}_i^h$ and $\mathcal{A}_i$.

\textbf{Construction of $\mathcal{A}_i^h$:}

For the construction of $\mathcal{A}_i^h$, the indices range from $0$ to $n-m$ and $1$ to $n$ just like in the construction for $\psi_i^h$.
The automaton is given as follows: $\mathcal{A}_i^h = (Q' \cup \hat{Q},\hat{q_h},\rho_i,\Omega')$ where
\begin{itemize}
	\item $\hat{Q} = \{\hat{q} | q \in Q\}$ is a copy of $Q$
	\item $Q'$ is a copy of $Q$ where for $j > i$, the state $q_j$ is substituted by a hole
	\item $\Omega'(q) = \Omega'(\hat{q}) = \Omega(q)$ for all $q \in Q$
	\item $\rho_i(\hat{q_h},(s,v),d) = \rho(q_h,(s,v),d)$
	\item $\rho_i(q_h,(s,v),d) = \rho(q_h,(s,v),d)$ for $h \leq i$
\end{itemize}

\textbf{Construction of $\mathcal{A}_i$:}

The automaton is given as $\mathcal{A}_i = (Q' \cup \hat{Q},\hat{\rho_0},\rho_i,\Omega')$ where
\begin{itemize}
	\item $Q', \hat{Q}, \rho_i$ and $\Omega'$ are taken from an arbitrary $\mathcal{A}_i^h$ and
	\item $\hat{\rho_0} = \rho_0 [q_1 / \hat{q_1}] ... [q_n / \hat{q_n}]$
\end{itemize}

\begin{lemma}\label{lem:aapatohmulemma1}
	$\mathcal{A}_i^h$ is $\mathcal{K}$-equivalent to $\psi_i^h$.
\end{lemma}
\begin{proof}
	The proof proceeds by an induction on $i$.
	
	\textbf{Base case:} $i = 0$: 
	For $h > n-m$, we have $\psi_0^h = X_h$ and $\mathcal{A}_0^h$ has (a copy of) the hole $X_h$ as its starting state.
	Then $v \in \llbracket X_h \rrbracket^{\mathcal{V}}(\Pi)$ iff $v \in \mathcal{V}(X_h)(\Pi)$ iff $w_{\Pi}[v] \in \mathcal{L}(\mathcal{V}(X_h)(\Pi))$ iff $w_{\Pi}[v] \in \mathcal{A}_0^h [X_1 : \mathcal{L}(\mathcal{V}(X_1)(\Pi)),...,X_m : \mathcal{L}(\mathcal{V}(X_m)(\Pi))]$.
	
	For $h \leq n-m$, all of the states $q_h$ are replaced by holes while the states $\hat{q_h}$ inherit their transitions from $q_h$ in $\mathcal{A}$.
	Thus $w_{\Pi}[v] \in \mathcal{L}(\mathcal{A}_0^h [X_1 : \mathcal{L}(\mathcal{V}(X_1)(\Pi)),...,X_m : \mathcal{L}(\mathcal{V}(X_m)(\Pi))])$ implies that there is a combination of holes $\{Y_1,...,Y_l\} \subseteq \{X_1,...,X_m\}$ making $\rho(q_h,w_{\Pi}(v)|_d,d)$ true for some $d$ such that $w_{\Pi}[v + e_{d}] \in \mathcal{L}(\mathcal{V}(Y_j)(\Pi))$ and therefore $v + e_{d} \in \mathcal{V}(Y_j)(\Pi)$ for all $1 \leq j \leq l$.
	One can easily see that for this $d$ and $\sigma = w_{\Pi}(v)|_d$, we then have $v \in \llbracket \sigma_{\pi_d} \land \bigcirc_{\pi_d} \hat{\rho}(q_h,\sigma,d) \rrbracket^{\mathcal{V}}(\Pi)$ and thus $v \in \llbracket \psi_0^h \rrbracket^{\mathcal{V}}(\Pi)$.
	The other direction works in a similar way when additionally noticing that by the definition of $w_{\Pi}$, the $\sigma$ making the disjunction in $\psi_0^h$ true has to match $v$ in direction $d$ and can therefore only be $w_{\Pi}(v)|_d$.
	
	\textbf{Inductive step:} $i \mapsto i+1$:
	We first show the claim for $h = i+1$ and then use that as a lemma for $h \neq i + 1$.
	Also, we consider the case where $\psi_{i+1}^{i+1} = \mu X_{i+1}. \psi_i^{i+1}$ and $q_{i+1}$ has an odd priority since the other case is analogous.
	Since all states $q_j$ with $j > i+1$ are substituted by holes, $\Omega(q_{i+1})$ is the lowest priority in all automata $\mathcal{A}_{i+1}^h$.
	In this case, where $\Omega(q_{i+1})$ is odd, any path in an accepting run of $\mathcal{A}_{i+1}^h$ on some word $w$ can only visit $q_{i+1}$ a finite number of times.
	Since $\hat{q_{i+1}}$ in $\mathcal{A}_i^{h}$ has the same transitions as $q_{i+1}$ in $\mathcal{A}_{i+1}^h$ up to the hole substitution for $q_{i+1}$, $\mathcal{L}(\mathcal{A}_{i+1}^{i+1}[X_{i+2} : \mathcal{L}(\mathcal{V}(X_{i+2})(\Pi),...,X_n : \mathcal{L}(\mathcal{V}(X_n)(\Pi)))])$ can be characterised as the least fixpoint of the function $f : \mathcal{L} \mapsto \mathcal{A}_i^{i+1}[X_{i+1} : \mathcal{L},X_{i+2} : \mathcal{L}(\mathcal{V}(X_{i+2})(\Pi),...,X_n : \mathcal{L}(\mathcal{V}(X_n)(\Pi)))]$ or as a union of its approximants $\bigcup_{\kappa \geq 0} f^\kappa(\emptyset)$.
	On the other hand semantics of $\psi_{i+1}^{i+1}$ can be characterised as a union of its approximants $\bigsqcup_{\kappa \geq 0} \alpha^\kappa(\bot)$ using \autoref{cor:welldefined}.
	
	We show that $w_{\Pi}[v] \in f^k(\emptyset)$ iff $v \in \alpha^k(\bot)(\Pi)$ for all $\kappa$.
	Indeed, for arbitrary $\kappa \geq 1$, this follows immediately from the definition of $f$ and $\alpha$, the induction hypothesis, and the following claim:
	
	\textbf{Claim:} $w_{\Pi}[v] \in \mathcal{L}( \mathcal{A}_i^{i+1}[X_{i+1} : f^{\kappa-1}(\emptyset),X_{i+2}: \mathcal{L}(\mathcal{V}(X_{i+2})(\Pi)),...,X_n : \mathcal{L}(\mathcal{V}(X_n)(\Pi))])$ holds iff we have $w_{\Pi}[v] \in \mathcal{L}(\mathcal{A}_i^{i+1}[X_{i+1}: \mathcal{L}(\mathcal{V}[X_{i+1} \mapsto \alpha^{\kappa-1}(\bot)](X_{i+1})(\Pi)),...,X_{n}: \mathcal{L}(\mathcal{V}[X_{i+1} \mapsto \alpha^{\kappa-1}(\bot)](X_{n})(\Pi))])$.
	
	This can be shown by a transfinite induction on $\kappa$, similar to the one in the proof of \autoref{thm:aapahmuequivalence} part 1.
	
	Now we consider the case where $h \neq i+1$ and let $v \in \llbracket \psi_{i+1}^{h} \rrbracket^{\mathcal{V}}(\Pi) = \llbracket \psi_{i}^h [X_{i+1} / \psi_{i+1}^{i+1}] \rrbracket^{\mathcal{V}}(\Pi)$.
	Due to the way the substitution was done, we have a set $V$ of vectors $v' \in \llbracket \psi_{i+1}^{i+1} \rrbracket^{\mathcal{V}}(\Pi)$ such that all $v' \in V$ combined \textit{witness} $v \in \llbracket \psi_{i+1}^h \rrbracket^{\mathcal{V}}(\Pi)$.
	Considering some $\mathcal{V}'$ with $V \subseteq \mathcal{V}'(X_{i+1})(\Pi)$, we then have $v \in \llbracket \psi_{i}^h \rrbracket^{\mathcal{V}'}(\Pi)$, in particular $v \in \llbracket \psi_i^h \rrbracket^{\mathcal{V}[X_{i+1} \mapsto \llbracket \psi_{i+1}^{i+1} \rrbracket^{\mathcal{V}}]}(\Pi)$.
	Using the induction hypothesis, we get $w_{\Pi}[v] \in \mathcal{L}(\mathcal{A}_i^h[X_{i+1}:\mathcal{V}[X_{i+1} \mapsto \llbracket \psi_{i+1}^{i+1} \rrbracket^{\mathcal{V}}](X_{i+1})(\Pi),...,X_{n}:\mathcal{V}[X_{i+1} \mapsto \llbracket \psi_{i+1}^{i+1} \rrbracket^{\mathcal{V}}](X_{n})(\Pi)])$.
	For $X_j \neq X_{i+1}$ we can remove the substitution in $\mathcal{V}$ without changing language containment.
	Now all that is left to show is that we can remove the substitution for hole $X_{i+1}$ and go over to $\mathcal{A}_{i+1}^h$ without changing language containment as well.
	This is done by using the claim for $h = i+1$ as a lemma and exchange the substitution of $\psi_{i+1}^{i+1}$'s semantics with the language of $\mathcal{A}_{i+1}^{i+1}$ (with appropriate substitutions) and then noticing that $\mathcal{A}_{i+1}^h$ is actually the same automaton as $\mathcal{A}_{i+1}^{i+1}$ with a different starting state.
	Thus, instead of moving to hole $X_{i+1}$ using the substitution of $\mathcal{A}_{i+1}^h$, we can instead move to $q_{i+1}$ and obtain the same behaviour.
	Since this is exactly the difference between $\mathcal{A}_i^h$ and $\mathcal{A}_{i+1}^h$, we obtain our result.
	
	On the other hand let $w_{\Pi}[v] \in \mathcal{L}(\mathcal{A}_{i+1}^{h}[X_{i+2} : \mathcal{L}(\mathcal{V}(X_{i+2})(\Pi)),...,X_n : \mathcal{L}(\mathcal{V}(X_n)(\Pi))])$.
	Just like in the proof of the other direction, we notice that $\mathcal{A}_{i+1}^h$ is the same automaton as $\mathcal{A}_{i+1}^{i+1}$ with a different starting state.
	Thus, whenever a path in the accepting run moves into $q_{i+1}$, we can instead move to a hole $X_{i+1}$ substituted with the language of $\mathcal{A}_{i+1}^{i+1}$ (with appropriate substitutions) and obtain the same behaviour having gone over to $\mathcal{A}_i^h$.
	Using the claim for $h = i+1$ as a lemma, we can exchange the substitution for hole $X_{i+1}$ with $\llbracket \psi_{i+1}^{i+1} \rrbracket^{\mathcal{V}}(\Pi)$.
	Since for $X_j \neq X_{i+1}$, a modification of $\mathcal{V}$ on $X_{i+1}$ does not change behaviour, we can condense all substitutions into a predicate environment $\mathcal{V}[X_{i+1} \mapsto \llbracket \psi_{i+1}^{i+1} \rrbracket^{\mathcal{V}}]$.
	Thus we obtain $w_{\Pi}[v] \in \mathcal{L}(\mathcal{A}_i^h[X_{i+1}:\mathcal{V}[X_{i+1} \mapsto \llbracket \psi_{i+1}^{i+1} \rrbracket^{\mathcal{V}}](X_{i+1})(\Pi),...,X_{n}:\mathcal{V}[X_{i+1} \mapsto \llbracket \psi_{i+1}^{i+1} \rrbracket^{\mathcal{V}}](X_{n})(\Pi)])$, where we can apply the induction hypothesis to get $v \in \llbracket \psi_i^h \rrbracket^{\mathcal{V}[X_{i+1} \mapsto \llbracket \psi_{i+1}^{i+1} \rrbracket^{\mathcal{V}}]}(\Pi)$.
	Instead of substituting $\psi_{i+1}^{i+1}$'s semantics for $X_{i+1}$ in the predicate environment $\mathcal{V}$, we can instead substitute $\psi_{i+1}^{i+1}$ for $X_{i+1}$ in the formula without changing behaviour.
	Therefore we have $v \in \llbracket \psi_{i}^h [X_{i+1} / \psi_{i+1}^{i+1}] \rrbracket^{\mathcal{V}}(\Pi) = \llbracket \psi_{i+1}^h \rrbracket^{\mathcal{V}}(\Pi)$.
\end{proof}

\begin{lemma}\label{lem:aapatohmulemma2}
	$\mathcal{L}(\mathcal{A}[X_1:\mathcal{A}_1,...,X_m:\mathcal{A}_m]) = \mathcal{L}(\mathcal{A}_{n-m}[X_1:\mathcal{A}_1,...,X_m:\mathcal{A}_m])$ for all AAPA $\mathcal{A}_1,...,\mathcal{A}_m$.
\end{lemma}
\begin{proof}
	Since for all $j > n-m$ the state $q_j$ is a hole, no states in $Q'$ are further substituted by holes.
	Therefore, the set of states in $\mathcal{A}_{m-n}$ consists of  two copies $Q'$ and $\hat{Q}$ of $Q$, where moves from $\hat{Q}$ to $Q'$ are possible, but not the other way around.
	Since $\hat{Q}$ is left after one transition, any run of $\mathcal{A}_{m-n}$ on some word $w$ behaves just like a run of $\mathcal{A}$ where the first state is substituted by its copy and vice versa.
	Thus, the two automata recognise the same language.
\end{proof}

\begin{proof}[Proof of \autoref{thm:aapahmuequivalence} part 2]
	The sought formula is given as $\rho_0[q_1 / \psi_{n-m}^1]...[q_n / \psi_{n-m}^n]$ when $\rho_0$ is the starting function of of $\mathcal{A}$.
	Our claim immediately follows from \autoref{lem:aapatohmulemma1} and \autoref{lem:aapatohmulemma2}.
\end{proof}

\section{Missing proofs from Section \ref{sec:modelchecking}}

\subsection{Proof of \autoref{thm:constructionproperties}}

\begin{proof}
	We show that the syntactic restrictions of the formula yield automata only having the respective kinds of runs.
	
	For a synchronous formula, $\bigcirc \psi$ constructs are used instead of $\bigcirc_i \psi$ constructs.
	These are translated to automata where a transition to $\mathcal{A}_{\psi}$ is performed only after a symbol from each direction is read.
	The different nodes reading their respective direction can then be merged into a single node reading a vector.
	Also, the construction for atomic propositions can straightforwardly be adapted to read a vector.
	Since the transition functions in all other constructions are defined inductively, this translates the AAPA into an APA.
	
	For $k$-synchronous and $k$-context-bounded formulas, notice that the construction is performed such that each node in a run of $\mathcal{A}_{\psi}$ corresponds to a node in the extended syntax tree of $\psi$ (but not necessarily the other way around due to disjunctions).
	Thus, a restriction on the extended syntax tree straightforwardly translates to a restriction on the runs of the corresponding AAPA.
\end{proof}

\subsection{Showing the $\mathcal{K}$-equivalence of quantified formulas to their respective automata}

\begin{proof}
	We show the proof by a structural induction with three cases.
	Here, we only consider the case for an innermost quantifier $\exists \pi_{n+1} . \psi$ in depth since the other cases are similar or trivial.
	
	\textbf{Existential Quantifiers:}
	Doing the case for an $\exists$ quantifier.
	We show $\Pi \models_k^{\mathcal{K}} \exists \pi_{n+1} . \varphi$ iff $w_{\Pi} \in \mathcal{A}_{\exists \pi_{n+1} . \varphi}$.
	
	Assume $\Pi \models_k^{\mathcal{K}} \exists \pi_{n+1} . \varphi$ holds.
	By the definition of semantics, this implies that there is a $p \in Paths(\mathcal{K})$ such that $\Pi[\pi_{n+1} \mapsto p] \models_k^{\mathcal{K}} \varphi$.
	We use the induction hypothesis and obtain that $w_{\Pi[\pi_{n+1} \mapsto p]} \in \mathcal{A}_{\varphi}$, i.e. $\mathcal{A}_{\varphi}$ has an accepting run $q_0' q_1' ...$ on $w_{\Pi[\pi_{n+1} \mapsto p]}$.
	Since $w_{\Pi[\pi_{n+1} \mapsto p]}$ has one additional $S$-component $s_0 s_1 ...$ compared to $w_{\Pi}$ and this component represents a path of $\mathcal{K}$ starting in $s_0$, we can simulate this additional component by the state space of $\mathcal{A}_{\varphi}$ and obtain an accepting run $q_0 q_1 ... = (q_0', s_0)(q_1', s_1)....$ of $\mathcal{A}_{\exists_{\pi_{n+1}} . \psi}$ on $w_{\Pi}$.
	
	On the other hand let $w_{\Pi} \in \mathcal{L}(\mathcal{A}_{\exists \pi_{n+1} . \varphi})$.
	Therefore we have an accepting run $q_0 q_1 ... = (q_0', s_0)(q_1', s_1)....$ of $\mathcal{A}_{\exists_{\pi_{n+1}} . \varphi}$ on $w_{\Pi}$.
	Due to the way $\mathcal{A}_{\exists \pi_{n+1} . \varphi}$ was constructed, the second component $s_0 s_1 ...$ represents a path $p$ of $\mathcal{K}$ starting in $s_0$.
	Additionally, the automaton makes sure that $q_0' q_1' ...$ is an accepting run of $\mathcal{A}_{\varphi}$ on $w_{\Pi[\pi_{n+1} \mapsto p]}$.
	We use the induction hypothesis to obtain that $\Pi[\pi_{n+1} \mapsto p] \models_k^{\mathcal{K}} \varphi$.
	Thus, the existence of $p$ witnesses $\Pi \models_k^{\mathcal{K}} \exists \pi_{n+1} . \varphi$.
	
	\textbf{Universal quantifiers:} Analogous to the previous case.
	
	\textbf{Innermost formula:} Here we use \autoref{thm:aapahmuequivalence} as a lemma to directly obtain the result from the semantics definition.
\end{proof}

\subsection{Proof of \autoref{lem:constructionsize}}

\begin{proof}
	The three claims follow from the observation that the analyses produce synchronous, $k$-synchronous and $k$-context-bounded AAPA in combination with \autoref{thm:aapaksynchronoustosynchronous}, \autoref{cor:contextboundedtoapa} and a further argument we present here.
	
	The mentioned argument is that the construction for quantifiers increases the size of the construction exponentially in the worst case for each alternation removal that is being performed.
	Since the Kripke Structure $\mathcal{K}$ is incorporated into the automaton only after the first alternation removal, the complexity in the size of $\mathcal{K}$ is one exponent lower than in the size of $\mathcal{\varphi}$.
	It remains to show that the number of alternation removals is equal to the number of quantifier alternations plus one rather than the number of quantifiers, which a naive look at the construction would imply.
	This is due to the fact that the construction results in an NPA and thus no further alternation is needed when no complementation constructions are performed in between the quantifier constructions.
	For existential quantifiers one can see this directly while for universal quantifiers this is due to the elimination of double negations.
	An outermost negation does not have to be included into the construction since flipping the result of the test performed on the automaton has the same effect.
\end{proof}

\end{document}